\documentclass[11pt,english]{article}
\usepackage{lmodern}

\usepackage[T1]{fontenc}
\usepackage[latin9]{inputenc}
\usepackage{geometry}
\usepackage{float}
\usepackage{units}
\usepackage{dsfont}
\usepackage{amsmath}
\usepackage{amsthm}
\usepackage{amssymb}

\makeatletter

\floatstyle{ruled}
\newfloat{algorithm}{tbp}{loa}
\providecommand{\algorithmname}{Algorithm}
\floatname{algorithm}{\protect\algorithmname}

\numberwithin{figure}{section}
\numberwithin{equation}{section}
\theoremstyle{plain}
\newtheorem{thm}{\protect\theoremname}[section]
\theoremstyle{definition}
\newtheorem{defn}[thm]{\protect\definitionname}
\theoremstyle{definition}
\newtheorem{problem}[thm]{\protect\problemname}
\theoremstyle{plain}
\newtheorem{prop}[thm]{\protect\propositionname}
\theoremstyle{plain}
\newtheorem{lem}[thm]{\protect\lemmaname}
\theoremstyle{remark}
\newtheorem{rem}[thm]{\protect\remarkname}
\theoremstyle{plain}
\newtheorem{assumption}[thm]{\protect\assumptionname}
\theoremstyle{remark}
\newtheorem*{acknowledgement*}{\protect\acknowledgementname}

\usepackage{comment,url,algorithm,algorithmic,graphicx,subcaption,relsize}
\usepackage{amssymb,amsfonts,amsmath,amsthm,amscd,dsfont,mathrsfs,mathtools,microtype,nicefrac,pifont}
\usepackage{float,psfrag,epsfig,color,url,hyperref}
\hypersetup{
  colorlinks,
  linkcolor={red!50!black},
  citecolor={blue!50!black},
  urlcolor={blue!80!black}
}
\usepackage{upgreek}
\usepackage[dvipsnames]{xcolor}
\usepackage{epstopdf,bbm,mathtools,enumitem}
\usepackage[toc,page]{appendix}
\usepackage{dsfont}
\usepackage{soul}

\usepackage[mathscr]{euscript}
\allowdisplaybreaks
\usepackage{custom}
\makeatletter
\renewcommand{\paragraph}{%
  \@startsection{paragraph}{4}%
  {\z@}{1.25ex \@plus 1ex \@minus .2ex}{-1em}%
  {\normalfont\normalsize\bfseries}%
}
\makeatother

\makeatother

\usepackage{babel}
\providecommand{\acknowledgementname}{Acknowledgement}
\providecommand{\assumptionname}{Assumption}
\providecommand{\definitionname}{Definition}
\providecommand{\lemmaname}{Lemma}
\providecommand{\problemname}{Problem}
\providecommand{\propositionname}{Proposition}
\providecommand{\remarkname}{Remark}
\providecommand{\theoremname}{Theorem}

\begin{document}
%--------------------------------------------------------------------------------------------------------------------------------
% Environment shortcuts
%--------------------------------------------------------------------------------------------------------------------------------
\def\balign#1\ealign{\begin{align}#1\end{align}}
\def\baligns#1\ealigns{\begin{align*}#1\end{align*}}
\def\balignat#1\ealign{\begin{alignat}#1\end{alignat}}
\def\balignats#1\ealigns{\begin{alignat*}#1\end{alignat*}}
\def\bitemize#1\eitemize{\begin{itemize}#1\end{itemize}}
\def\benumerate#1\eenumerate{\begin{enumerate}#1\end{enumerate}}

% Align environments that use textstyle instead of displaystyle
\newenvironment{talign*}
 {\let\displaystyle\textstyle\csname align*\endcsname}
 {\endalign}
\newenvironment{talign}
 {\let\displaystyle\textstyle\csname align\endcsname}
 {\endalign}

\def\balignst#1\ealignst{\begin{talign*}#1\end{talign*}}
\def\balignt#1\ealignt{\begin{talign}#1\end{talign}}
%---------------------------------------------------

%--------------------------------------------------------------------------------------------------------------------------------
% Redefine left and right to remove initial and trailing space
%--------------------------------------------------------------------------------------------------------------------------------
\let\originalleft\left
\let\originalright\right
\renewcommand{\left}{\mathopen{}\mathclose\bgroup\originalleft}
\renewcommand{\right}{\aftergroup\egroup\originalright}

%--------------------------------------------------------------------------------------------------------------------------------
% Words with special symbols
%--------------------------------------------------------------------------------------------------------------------------------
\def\Gronwall{Gr\"onwall\xspace}
\def\Holder{H\"older\xspace}
\def\Ito{It\^o\xspace}
\def\Nystrom{Nystr\"om\xspace}
\def\Schatten{Sch\"atten\xspace}
\def\Matern{Mat\'ern\xspace}

%--------------------------------------------------------------------------------------------------------------------------------
% Smaller citations
%--------------------------------------------------------------------------------------------------------------------------------
\def\tinycitep*#1{{\tiny\citep*{#1}}}
\def\tinycitealt*#1{{\tiny\citealt*{#1}}}
\def\tinycite*#1{{\tiny\cite*{#1}}}
\def\smallcitep*#1{{\scriptsize\citep*{#1}}}
\def\smallcitealt*#1{{\scriptsize\citealt*{#1}}}
\def\smallcite*#1{{\scriptsize\cite*{#1}}}

%--------------------------------------------------------------------------------------------------------------------------------
% Colors
%--------------------------------------------------------------------------------------------------------------------------------
\def\blue#1{\textcolor{blue}{{#1}}}
\def\green#1{\textcolor{green}{{#1}}}
\def\orange#1{\textcolor{orange}{{#1}}}
\def\purple#1{\textcolor{purple}{{#1}}}
\def\red#1{\textcolor{red}{{#1}}}
\def\teal#1{\textcolor{teal}{{#1}}}

%--------------------------------------------------------------------------------------------------------------------------------
% Font styles
%--------------------------------------------------------------------------------------------------------------------------------
\def\mbi#1{\boldsymbol{#1}} % Bold and italic (math bold italic)
\def\mbf#1{\mathbf{#1}}
\def\mrm#1{\mathrm{#1}}
\def\tbf#1{\textbf{#1}}
\def\tsc#1{\textsc{#1}}

%--------------------------------------------------------------------------------------------------------------------------------
% Bold and italic variables
%--------------------------------------------------------------------------------------------------------------------------------
\def\mbiA{\mbi{A}}
\def\mbiB{\mbi{B}}
\def\mbiC{\mbi{C}}
\def\mbiDelta{\mbi{\Delta}}
\def\mbif{\mbi{f}}
\def\mbiF{\mbi{F}}
\def\mbih{\mbi{g}}
\def\mbiG{\mbi{G}}
\def\mbih{\mbi{h}}
\def\mbiH{\mbi{H}}
\def\mbiI{\mbi{I}}
\def\mbim{\mbi{m}}
\def\mbiP{\mbi{P}}
\def\mbiQ{\mbi{Q}}
\def\mbiR{\mbi{R}}
\def\mbiv{\mbi{v}}
\def\mbiV{\mbi{V}}
\def\mbiW{\mbi{W}}
\def\mbiX{\mbi{X}}
\def\mbiY{\mbi{Y}}
\def\mbiZ{\mbi{Z}}

%--------------------------------------------------------------------------------------------------------------------------------
% Textstyle vs. displaystyle
%--------------------------------------------------------------------------------------------------------------------------------
\def\textsum{{\textstyle\sum}} % Sum in textstyle form
\def\textprod{{\textstyle\prod}} % Prod in textstyle form
\def\textbigcap{{\textstyle\bigcap}} % Bigcap in textstyle form
\def\textbigcup{{\textstyle\bigcup}} % Bigcup in textstyle form

%--------------------------------------------------------------------------------------------------------------------------------
% Mathematical sets
%--------------------------------------------------------------------------------------------------------------------------------
\def\reals{\mathbb{R}} % Real number symbol
\def\integers{\mathbb{Z}} % Integer symbol
\def\rationals{\mathbb{Q}} % Rational numbers
\def\naturals{\mathbb{N}} % Natural numbers
\def\complex{\mathbb{C}} % Complex numbers

\def\what#1{\widehat{#1}}

\def\twovec#1#2{\left[\begin{array}{c}{#1} \\ {#2}\end{array}\right]}
\def\threevec#1#2#3{\left[\begin{array}{c}{#1} \\ {#2} \\ {#3} \end{array}\right]}
\def\nvec#1#2#3{\left[\begin{array}{c}{#1} \\ {#2} \\ \vdots \\ {#3}\end{array}\right]} % An n-vector with three arguments

%--------------------------------------------------------------------------------------------------------------------------------
% Eigenvalues
%--------------------------------------------------------------------------------------------------------------------------------
\def\maxeig#1{\lambda_{\mathrm{max}}\left({#1}\right)}
\def\mineig#1{\lambda_{\mathrm{min}}\left({#1}\right)}

%--------------------------------------------------------------------------------------------------------------------------------
% Operators
%--------------------------------------------------------------------------------------------------------------------------------
\def\Re{\operatorname{Re}} % Real part
\def\indic#1{\mbb{I}\left[{#1}\right]} % Indicator function
\def\logarg#1{\log\left({#1}\right)} % log with argument
\def\polylog{\operatorname{polylog}}
\def\maxarg#1{\max\left({#1}\right)} % max with argument
\def\minarg#1{\min\left({#1}\right)} % min with argument
\def\Earg#1{\E\left[{#1}\right]}
\def\Esub#1{\E_{#1}}
\def\Esubarg#1#2{\E_{#1}\left[{#2}\right]}
\def\bigO#1{\mathcal{O}\left(#1\right)} % big-oh notation
\def\littleO#1{o(#1)} % big-oh notation
\def\P{\mbb{P}} % Probability symbol
\def\Parg#1{\P\left({#1}\right)}
\def\Psubarg#1#2{\P_{#1}\left[{#2}\right]}
\def\Trarg#1{\Tr\left[{#1}\right]} % Trace with argument
\def\trarg#1{\tr\left[{#1}\right]} % trace with argument
\def\Var{\mrm{Var}} % Variance symbol
\def\Vararg#1{\Var\left[{#1}\right]}
\def\Varsubarg#1#2{\Var_{#1}\left[{#2}\right]}
\def\Cov{\mrm{Cov}} % Covariance symbol
\def\Covarg#1{\Cov\left[{#1}\right]}
\def\Covsubarg#1#2{\Cov_{#1}\left[{#2}\right]}
\def\Corr{\mrm{Corr}} % Covariance symbol
\def\Corrarg#1{\Corr\left[{#1}\right]}
\def\Corrsubarg#1#2{\Corr_{#1}\left[{#2}\right]}
\newcommand{\info}[3][{}]{\mathbb{I}_{#1}\left({#2};{#3}\right)} % Information symbol
\newcommand{\staticexp}[1]{\operatorname{exp}(#1)} % An exponential with parens that do not resize with input
\newcommand{\loglihood}[0]{\mathcal{L}} % log likelihood

% Copied from mathrsfs.sty

%--------------------------------------------------------------------------------------------------------------------------------
% Optimization macros
%--------------------------------------------------------------------------------------------------------------------------------
%\providecommand{\argmax}{\mathop\mathrm{arg max}} % Defining math symbols
%\providecommand{\argmin}{\mathop\mathrm{arg min}}
\providecommand{\arccos}{\mathop\mathrm{arccos}}
\providecommand{\dom}{\mathop\mathrm{dom}}
\providecommand{\diag}{\mathop\mathrm{diag}}
\providecommand{\tr}{\mathop\mathrm{tr}}
\providecommand{\card}{\mathop\mathrm{card}}
\providecommand{\sign}{\mathop\mathrm{sign}}
\providecommand{\conv}{\mathop\mathrm{conv}} % Convex hull
\def\rank#1{\mathrm{rank}({#1})}
\def\supp#1{\mathrm{supp}({#1})}

\providecommand{\minimize}{\mathop\mathrm{minimize}}
\providecommand{\maximize}{\mathop\mathrm{maximize}}
\providecommand{\subjectto}{\mathop\mathrm{subject\;to}}

\def\openright#1#2{\left[{#1}, {#2}\right)}

%--------------------------------------------------------------------------------------------------------------------------------
% Proof environments
%--------------------------------------------------------------------------------------------------------------------------------
\ifdefined\nonewproofenvironments\else
% The Theorems are numbered consecutively
% Lemmas are numbered by section, and observations, claims, facts, and 
% assumptions take their numbering. Propositions and definitions have their
% own numbering by section.
\ifdefined\ispres\else
% These conflict with Beamer definitions in pres mode
% \newtheorem{theorem}{Theorem}
% \newtheorem{lemma}[theorem]{Lemma}
% \newtheorem{corollary}[theorem]{Corollary}
% \newtheorem{definition}[theorem]{Definition}
% \newtheorem{fact}[theorem]{Fact}
% \renewenvironment{proof}{\noindent\textbf{Proof.}\hspace*{.3em}}{\qed \vspace{.1in}}
% \newenvironment{proof-sketch}{\noindent\textbf{Proof Sketch}
%   \hspace*{1em}}{\qed\bigskip\\}
% \newenvironment{proof-idea}{\noindent\textbf{Proof Idea}
%   \hspace*{1em}}{\qed\bigskip\\}
% \newenvironment{proof-of-lemma}[1][{}]{\noindent\textbf{Proof of Lemma {#1}}
%   \hspace*{1em}}{\qed\\}
%   \newenvironment{proof-of-proposition}[1][{}]{\noindent\textbf{Proof of Proposition {#1}}
%   \hspace*{1em}}{\qed\\}
% \newenvironment{proof-of-theorem}[1][{}]{\noindent\textbf{Proof of Theorem {#1}}
%   \hspace*{1em}}{\qed\\}
% \newenvironment{proof-attempt}{\noindent\textbf{Proof Attempt}
%   \hspace*{1em}}{\qed\bigskip\\}
% \newenvironment{proofof}[1]{\noindent\textbf{Proof of {#1}}
%   \hspace*{1em}}{\qed\bigskip\\}
 
% \newtheorem*{remark*}{Remark}
% \newenvironment{remark}{\noindent\textbf{Remark.}
%   \hspace*{0em}}{\smallskip}%\bigskip}
% \newenvironment{remarks}{\noindent\textbf{Remarks}
%   \hspace*{1em}}{\smallskip}
% \fi
% \newtheorem{observation}[theorem]{Observation}
% \newtheorem{proposition}[theorem]{Proposition}
% \newtheorem{claim}[theorem]{Claim}
% \newtheorem{assumption}{Assumption}
% \theoremstyle{definition}
% \newtheorem{example}[theorem]{Example}
%\renewcommand{\theassumption}{\Alph{assumption}} % Set counter for assumptions
                                                 % to be alphabetical
\fi
% Makes equation numbers have (1.1) style
% \numberwithin{equation}{section}
% \numberwithin{equation}{subsection}
\makeatletter
\@addtoreset{equation}{section}
\makeatother
\def\theequation{\thesection.\arabic{equation}}

\newcommand{\cmark}{\ding{51}}

\newcommand{\xmark}{\ding{55}}

%--------------------------------------------------------------------------------------------------------------------------------
% Equation environments
%--------------------------------------------------------------------------------------------------------------------------------
\newcommand{\eq}[1]{\begin{align}#1\end{align}}
\newcommand{\eqn}[1]{\begin{align*}#1\end{align*}}
\renewcommand{\Pr}[1]{\mathbb{P}\left( #1 \right)}
\newcommand{\Ex}[1]{\mathbb{E}\left[#1\right]}
%\newcommand{\var}[1]{\text{Var}\left(#1\right)}
%\newcommand{\ind}[1]{{\mathbbm{1}}_{\{ #1 \}} }
%\newcommand{\abs}[1]{\left|#1\right|}

%--------------------------------------------------------------------------------------------------------------------------------
% Comment environments
%--------------------------------------------------------------------------------------------------------------------------------
\newcommand{\matt}[1]{{\textcolor{Maroon}{[Matt: #1]}}}
\newcommand{\kook}[1]{{\textcolor{blue}{[Kook: #1]}}}
\definecolor{OliveGreen}{rgb}{0,0.6,0}
\newcommand{\sv}[1]{{\textcolor{OliveGreen}{[Santosh: #1]}}}

\global\long\def\on#1{\operatorname{#1}}%

\global\long\def\bw{\mathsf{Ball\ Walk}}%
\global\long\def\sw{\mathsf{Speedy\ Walk}}%
\global\long\def\gw{\mathsf{Gaussian\ Walk}}%
\global\long\def\ps{\mathsf{Proximal\ Sampler}}%
\global\long\def\sps{\mathsf{Proximal\ Gaussian\ Cooling}}%

\global\long\def\har{\mathsf{Hit\text{-}and\text{-}Run}}%
\global\long\def\gc{\mathsf{Gaussian\ Cooling}}%
\global\long\def\ino{\mathsf{\mathsf{In\text{-}and\text{-}Out}}}%
\global\long\def\tgc{\mathsf{Tilted\ Gaussian\ Cooling}}%
\global\long\def\PS{\mathsf{PS}}%
\global\long\def\psexp{\mathsf{PS}_{\textup{exp}}}%
\global\long\def\psann{\mathsf{PS}_{\textup{ann}}}%
\global\long\def\eval{\mathsf{Eval}}%

\global\long\def\O{\mathcal{O}}%
\global\long\def\Otilde{\widetilde{\mathcal{O}}}%
\global\long\def\Omtilde{\widetilde{\Omega}}%
\global\long\def\Thetilde{\widetilde{\Theta}}%

\global\long\def\E{\mathbb{E}}%
\global\long\def\Z{\mathbb{Z}}%
\global\long\def\P{\mathbb{P}}%
\global\long\def\N{\mathbb{N}}%
\global\long\def\K{\mathcal{K}}%

\global\long\def\R{\mathbb{R}}%
\global\long\def\Rext{\overline{\mathbb{R}}}%
\global\long\def\Rd{\mathbb{R}^{d}}%
\global\long\def\Rdd{\mathbb{R}^{d\times d}}%
\global\long\def\Rn{\mathbb{R}^{n}}%
\global\long\def\Rnn{\mathbb{R}^{n\times n}}%

\global\long\def\psd{\mathbb{S}_{+}^{d}}%
\global\long\def\pd{\mathbb{S}_{++}^{d}}%

\global\long\def\defeq{\stackrel{\mathrm{{\scriptscriptstyle def}}}{=}}%

\global\long\def\veps{\varepsilon}%
\global\long\def\lda{\lambda}%
\global\long\def\vphi{\varphi}%
\global\long\def\cpi{C_{\mathsf{PI}}}%
\global\long\def\clsi{C_{\mathsf{LSI}}}%
\global\long\def\cch{C_{\mathsf{Ch}}}%
\global\long\def\clch{C_{\mathsf{logCh}}}%

\global\long\def\half{\frac{1}{2}}%
\global\long\def\nhalf{\nicefrac{1}{2}}%
\global\long\def\texthalf{{\textstyle \frac{1}{2}}}%

\global\long\def\ind{\mathds{1}}%
\global\long\def\op{\mathsf{op}}%

\global\long\def\chooses#1#2{_{#1}C_{#2}}%

\global\long\def\inrad{\on{inrad}}%

\global\long\def\vol{\on{vol}}%

\global\long\def\supp{\on{supp}}%

\global\long\def\law{\on{law}}%

\global\long\def\tr{\on{tr}}%

\global\long\def\diag{\on{diag}}%

\global\long\def\Diag{\on{Diag}}%

\global\long\def\inter{\on{int}}%

\global\long\def\esssup{\on{ess\,sup}}%

\global\long\def\poly{\on{poly}}%

\global\long\def\polylog{\on{polylog}}%

\global\long\def\var{\on{var}}%

\global\long\def\ent{\on{Ent}}%

\global\long\def\cov{\on{Cov}}%

\global\long\def\e{\mathrm{e}}%

\global\long\def\id{\mathrm{id}}%

\global\long\def\spanning{\on{span}}%

\global\long\def\rows{\on{row}}%

\global\long\def\cols{\on{col}}%

\global\long\def\rank{\on{rank}}%

\global\long\def\T{\mathsf{T}}%

\global\long\def\bs#1{\boldsymbol{#1}}%

\global\long\def\eu#1{\EuScript{#1}}%

\global\long\def\mb#1{\mathbf{#1}}%

\global\long\def\mbb#1{\mathbb{#1}}%

\global\long\def\mc#1{\mathcal{#1}}%

\global\long\def\mf#1{\mathfrak{#1}}%

\global\long\def\ms#1{\mathscr{#1}}%

\global\long\def\mss#1{\mathsf{#1}}%

\global\long\def\msf#1{\mathsf{#1}}%

\global\long\def\textint{{\textstyle \int}}%
\global\long\def\Dd{\mathrm{D}}%
\global\long\def\D{\mathrm{d}}%
\global\long\def\grad{\nabla}%
 
\global\long\def\hess{\nabla^{2}}%
 
\global\long\def\lapl{\triangle}%
 
\global\long\def\deriv#1#2{\frac{\D#1}{\D#2}}%
 
\global\long\def\pderiv#1#2{\frac{\partial#1}{\partial#2}}%
 
\global\long\def\de{\partial}%
\global\long\def\lagrange{\mathcal{L}}%
\global\long\def\Div{\on{div}}%

\global\long\def\Gsn{\mathcal{N}}%
 
\global\long\def\BeP{\textnormal{BeP}}%
 
\global\long\def\Ber{\textnormal{Ber}}%
 
\global\long\def\Bern{\textnormal{Bern}}%
 
\global\long\def\Bet{\textnormal{Beta}}%
 
\global\long\def\Beta{\textnormal{Beta}}%
 
\global\long\def\Bin{\textnormal{Bin}}%
 
\global\long\def\BP{\textnormal{BP}}%
 
\global\long\def\Dir{\textnormal{Dir}}%
 
\global\long\def\DP{\textnormal{DP}}%
 
\global\long\def\Expo{\textnormal{Expo}}%
 
\global\long\def\Gam{\textnormal{Gamma}}%
 
\global\long\def\GEM{\textnormal{GEM}}%
 
\global\long\def\HypGeo{\textnormal{HypGeo}}%
 
\global\long\def\Mult{\textnormal{Mult}}%
 
\global\long\def\NegMult{\textnormal{NegMult}}%
 
\global\long\def\Poi{\textnormal{Poi}}%
 
\global\long\def\Pois{\textnormal{Pois}}%
 
\global\long\def\Unif{\textnormal{Unif}}%

\global\long\def\bpar#1{\bigl(#1\bigr)}%
\global\long\def\Bpar#1{\Bigl(#1\Bigr)}%

\global\long\def\snorm#1{\|#1\|}%
\global\long\def\bnorm#1{\bigl\Vert#1\bigr\Vert}%
\global\long\def\Bnorm#1{\Bigl\Vert#1\Bigr\Vert}%

\global\long\def\sbrack#1{[#1]}%
\global\long\def\bbrack#1{\bigl[#1\bigr]}%
\global\long\def\Bbrack#1{\Bigl[#1\Bigr]}%

\global\long\def\sbrace#1{\{#1\}}%
\global\long\def\bbrace#1{\bigl\{#1\bigr\}}%
\global\long\def\Bbrace#1{\Bigl\{#1\Bigr\}}%

\global\long\def\abs#1{\lvert#1\rvert}%
\global\long\def\Par#1{\left(#1\right)}%
\global\long\def\Brack#1{\left[#1\right]}%
\global\long\def\Brace#1{\left\{  #1\right\}  }%

\global\long\def\inner#1{\langle#1\rangle}%
 
\global\long\def\binner#1#2{\left\langle {#1},{#2}\right\rangle }%

\global\long\def\norm#1{\lVert#1\rVert}%
\global\long\def\onenorm#1{\norm{#1}_{1}}%
\global\long\def\twonorm#1{\norm{#1}_{2}}%
\global\long\def\infnorm#1{\norm{#1}_{\infty}}%
\global\long\def\fronorm#1{\norm{#1}_{\text{F}}}%
\global\long\def\nucnorm#1{\norm{#1}_{*}}%
\global\long\def\staticnorm#1{\|#1\|}%
\global\long\def\statictwonorm#1{\staticnorm{#1}_{2}}%

\global\long\def\mmid{\mathbin{\|}}%

\global\long\def\wtilde{\widetilde{W}}%
\global\long\def\wt#1{\widetilde{#1}}%

\global\long\def\KL{\msf{KL}}%
\global\long\def\dtv{d_{\textrm{\textup{TV}}}}%
\global\long\def\FI{\msf{FI}}%
\global\long\def\tv{\msf{TV}}%

\global\long\def\PI{\msf{PI}}%

\global\long\def\cred#1{\textcolor{red}{#1}}%
\global\long\def\cblue#1{\textcolor{blue}{#1}}%
\global\long\def\cgreen#1{\textcolor{green}{#1}}%
\global\long\def\ccyan#1{\textcolor{cyan}{#1}}%

\global\long\def\iff{\Leftrightarrow}%
 
\global\long\def\textfrac#1#2{{\textstyle \frac{#1}{#2}}}%

\title{Sampling and Integration of Logconcave Functions by\\
Algorithmic Diffusion\date{}\author{Yunbum Kook\\ Georgia Tech\\  \texttt{yb.kook@gatech.edu} \and Santosh S. Vempala\\ Georgia Tech\\ \texttt{vempala@gatech.edu}}}
\maketitle
\begin{abstract}
We study the complexity of sampling, rounding, and integrating arbitrary
logconcave functions. Our new approach provides the first complexity
improvements in nearly two decades for general logconcave functions
for all three problems, and matches the best-known complexities for
the special case of uniform distributions on convex bodies. For the
sampling problem, our output guarantees are significantly stronger
than previously known, and lead to a streamlined analysis of statistical
estimation based on dependent random samples. 
\end{abstract}
\tableofcontents{}

\setcounter{page}{0}
\thispagestyle{empty}

\newpage{}

\section{Introduction}

Sampling and integration of logconcave functions are fundamental problems
with numerous applications and important special cases such as uniform
distributions over convex bodies and strongly logconcave densities.
The study of these problems has led to many useful techniques. Both
mathematically and algorithmically, general logconcave functions typically
provide the ``right'' general abstraction. For example, many classical
inequalities for convex bodies have natural extensions to logconcave
functions (e.g., Gr\"unbaum's theorem, Brunn-Minkowski and Pr\'ekopa-Leindler,
isotropic constant, etc.). The KLS hyperplane conjecture, first motivated
by the analysis of the $\bw$ for sampling convex bodies, is in the
setting of general logconcave densities. The current fastest algorithm
for estimating the volume of a convex body crucially uses sampling
from a sequence of logconcave densities, which is provably more efficient
than using a sequence of uniform distributions. Sampling logconcave
densities has many other applications as well, such as portfolio optimization,
simulated annealing, Bayesian inference, differential privacy etc. 

Sampling in high dimension is done by randomized algorithms based
on Markov chains. These chains are set up to have a desired stationary
distribution, which is relatively easy to ensure. For example, to
sample uniformly, it suffices that the Markov chain is symmetric.
Generally, to sample proportional to a desired function, it suffices
to ensure the ``detailed balance'' condition. The main challenge is
showing rapid mixing of the Markov chain, i.e., the convergence rate
to the stationary distribution is bounded by a (small) polynomial
in the dimension and other relevant parameters. 

The traditional analysis of Markov chains for sampling high-dimensional
distributions proceeds by analyzing the \emph{conductance} of the
Markov chain, the minimum conditional escape probability over all
subsets of the state space of measure at most half. Bounding this
is done by relating probabilistic (one-step distribution) distance
to geometric distance, and then using a purely isoperimetric inequality
for subsets of the support. Indeed, this approach led to several interesting
questions and useful techniques, in particular the development of
isoperimetric inequalities and the discovery of (nearly) tight bounds
in many settings. To describe background and known results, let us
first define the sampling problem (readers familiar with the problem
can skip to \S\ref{subsec:Results}).

\paragraph{Model and problems. }

We assume access to an integrable logconcave function $\exp(-V(x))$
in $\R^{n}$, via an evaluation oracle for a convex function $V:\Rn\to\R\cup\{\infty\}$.
We also assume that  there exists a point $x_{0}$ and parameters
$r,R$ such that for the distribution $\pi$ with density $\D\pi\propto\exp(-V)\,\D x$,
a ball of radius $r$ centered at $x_{0}$ is contained in the level
set of $\pi$ of measure $1/8$\footnote{We use the standard convention of $1/8$; any constant bounded away
from $1$ would suffice. In fact, we will use an even weaker condition;
see Problem~\ref{eq:ground_set} and Lemma~\ref{lem:relaxed-regularity}.} and $\E_{\pi}[\norm{X-x_{0}}^{2}]\leq R^{2}$. We refer to this as
a\emph{ well-defined function oracle}\footnote{Without loss of generality, we will assume by scaling that $r=1$.},
denoted by $\msf{Eval}_{\mc P}(V)$ where $\mc P$ indicates access
to the actual values of parameters in $\mc P$ (e.g., $\msf{Eval}_{x_{0},R}(V)$
presents both $x_{0}$ and $R$, while $\msf{Eval}(V)$ does not provide
any parameters).

Given this oracle, in this paper we consider three central problems:
(1) sample from the distribution $\pi$, (2) find an affine transformation
that places $\pi$ in near-isotropic position, and (3) estimate the
integral $\int e^{-V(x)}\,\D x$ (i.e., the normalization constant
of $\pi$). We measure the complexity in terms of the number of oracle
calls and the total number of arithmetic operations. 

\paragraph{Complexity of logconcave sampling. }

Sampling general logconcave functions, as an algorithmic problem,
was first studied by Applegate and Kannan \cite{applegate1991sampling},
who established an algorithm with complexity polynomial in the dimension
assuming Lipschitzness of the function over its support. They did
this by creating a sufficiently fine grid (so that the function changed
by at most a small constant within each grid cell) and used a discrete
``grid walk''. They sampled a grid point and used rejection sampling
to output a point from the cube associated with the grid point. The
output distribution is guaranteed to be within desired total variation
($\tv$) distance from the target distribution. Following many developments,
Lov\'asz and Vempala improved the bound to scale as $n^{2}R^{2}/r^{2}$
\cite{lovasz2007geometry,lovasz2006fast}. More specifically, \cite[Theorem 2.2, 2.3]{lovasz2007geometry}
analyzed the $\bw$ and $\har$ and proved a bound of 
\[
\frac{M^{4}}{\varepsilon^{4}}\frac{n^{2}R^{2}}{r^{2}}\log\frac{M}{\varepsilon}
\]
steps/queries to reach a distribution within $\varepsilon$ $\tv$-distance
of the target starting from an $M$-warm distribution. \cite[Theorem 1.1]{lovasz2006fast}
showed that $\har$ achieves the same guarantee in 
\[
\frac{n^{2}R^{2}}{r^{2}}\log^{\O(1)}\frac{nMR}{\varepsilon r}
\]
steps. We note that in all previous work on general logconcave sampling
the output guarantees are in $\tv$-distance. It is often desirable
to have stronger guarantees such as the $\KL$ or R\'enyi divergences. 

\paragraph{Classical sampling algorithms.}

For target density proportional to $e^{-V}$, the $\bw$ with parameter
$\delta$ proceeds as follows: at a point $x$ with $e^{-V(x)}>0$,
sample a uniform random point $y$ in the ball of radius $\delta$
centered at $x$; go to $y$ with probability $\min(1,\frac{e^{-V(y)}}{e^{-V(x)}})$,
staying at $x$ with the remaining probability. $\har$ does not need
a parameter: at current point $x$, pick a uniform random line $\ell$
through $x$, and go to a random point $y$ along $\ell$ with probability
proportional to $e^{-V(y)}$ (marginal along $\ell$). For technical
reasons, both of these walks (and almost all other walks) need to
be \emph{lazy}, i.e., they do nothing with probability $1/2$, and
do the above step with probability $1/2$.\emph{ }

\paragraph{Sampling and isoperimetry. }

The complexity of sampling is intuitively tied closely to the isoperimetry
of the target distribution. If the target has poor isoperimetry (roughly,
a small measure surface can partition the support into two large measure
subsets), then any ``local'' Markov chain will have difficulty moving
from one large subset to its complement. There are two alternative
views of isoperimetry --- functional and geometric. We recap the
definitions of isoperimetric constants.
\begin{defn}
We say that a probability measure $\pi$ on $\R^{n}$ satisfies a
\emph{Poincar\'e inequality} (PI) with parameter $C_{\msf{PI}}(\pi)$
if for all smooth functions $f:\R^{n}\to\R$,
\begin{equation}
\var_{\pi}f\leq\cpi(\pi)\,\E_{\pi}[\norm{\nabla f}^{2}]\,,\tag{\ensuremath{\msf{PI}}}\label{eq:pi}
\end{equation}
where $\var_{\pi}f=\E_{\pi}[\abs{f-\E_{\pi}f}^{2}]$.
\end{defn}

The Poincar\'e inequality is implied by the generally stronger log-Sobolev
inequality.
\begin{defn}
We say that a probability measure $\pi$ on $\R^{n}$ satisfies a
\emph{log-Sobolev inequality} (LSI) with parameter $C_{\msf{LSI}}(\pi)$
if for all smooth functions $f:\Rn\to\R$,
\begin{equation}
\ent_{\pi}(f^{2})\leq2\clsi(\pi)\,\E_{\pi}[\norm{\nabla f}^{2}]\,,\tag{\ensuremath{\msf{LSI}}}\label{eq:lsi}
\end{equation}
where $\ent_{\pi}(f^{2}):=\E_{\pi}[f^{2}\log f^{2}]-\E_{\pi}[f^{2}]\log\E_{\pi}[f^{2}]$. 
\end{defn}

For the geometric view, we define $\pi(\de S):=\liminf_{\veps\downarrow0}\frac{\pi(S_{\veps})-\pi(S)}{\veps}$,
where $S_{\veps}:=\{x:d(x,S)\leq\veps\}$. 
\begin{defn}
The \emph{Cheeger constant} $\cch(\pi)$ of a probability measure
$\pi$ on $\R^{n}$ is defined as 
\[
\cch(\pi)=\inf_{S\subset\R^{n},\pi(S)\le\frac{1}{2}}\frac{\pi(\partial S)}{\pi(S)}\,.
\]
\end{defn}

\begin{defn}
The \emph{log-Cheeger constant} $\clch(\pi)$ of a probability measure
$\pi$ on $\R^{n}$ is defined as 
\[
\clch(\pi)=\inf_{S\subset\R^{n},\pi(S)\le\frac{1}{2}}\frac{\pi(\partial S)}{\pi(S)\sqrt{\log(1/\pi(S))}}\,.
\]
\end{defn}

It is known that for logconcave measures, $C_{\msf{PI}}(\pi)=\Theta(\cch^{-2}(\pi))$
\cite{milman2009role} and $\clsi(\pi)=\Theta(\clch^{-2}(\pi))$ \cite{ledoux1994simple}.
Bounding these constants has been a major research topic for decades.
In recent years, following many improvements, it has been shown that
for isotropic logconcave measures, $\cpi(\pi)\lesssim\log n$ \cite{klartag2023logarithmic}
and for isotropic logconcave ones with support of diameter $D$, we
have $C_{\msf{LSI}}(\pi)\lesssim D$ \cite{lee2024eldan}. The former
is conjectured to be $\O(1)$ (the KLS conjecture), and the latter
is the best possible. 

The significance of these constants for algorithmic sampling became
clear with the analysis of the $\bw$, where the bound on its convergence
from a warm start to a logconcave distribution depends directly on
$\cch^{-2}(\pi)$ \cite{kannan1997random} (this was the original
motivation for the KLS conjecture) --- the mixing time of the $\bw$
starting from an $M$-warm distribution to reach $\veps$ $\tv$-distance
is bounded by $\O(n^{2}\cch^{-2}\poly\frac{M}{\veps})$. For the special
case of uniformly sampling convex bodies, using the notion of a \emph{$\sw$,
}this can be improved to $\O(Mn^{2}\cch^{-2}\log\frac{M}{\veps})$,
but using it in the analysis entails handling several technical difficulties.
In principle, this should be extendable to general logconcave functions,
albeit with formidable technical complications. We take a different
approach. 

\paragraph{Sampling, isoperimetry, and diffusion. }

The connection of isoperimetric constants to the convergence of continuous-time
\emph{diffusion} is a classical subject. The Langevin diffusion is
a canonical stochastic differential equation (SDE) given by
\[
\D X_{t}=-\nabla V(X_{t})\,\D t+\sqrt{2}\,\D B_{t}\quad X_{0}\sim\nu_{0}\,,
\]
where $B_{t}$ is the standard Brownian process, and $V$ is a smooth
function. Then, under mild conditions, the law $\nu_{t}$ of $X_{t}$
converges to the distribution with density $\nu$ proportional to
$e^{-V}$ in the limit. Moreover, it can be shown that
\[
\chi^{2}(\nu_{t}\mmid\nu)\leq\exp\bpar{-\frac{2t}{\cpi(\nu)}}\,\chi^{2}(\nu_{0}\mmid\nu)\,,\quad\text{and}\quad\KL(\nu_{t}\mmid\nu)\le\exp\bpar{-\frac{2t}{\clsi(\nu)}}\,\KL(\nu_{0}\mmid\nu)\,.
\]
A natural idea for a sampling algorithm is to discretize the diffusion
equation. This discretized algorithm was shown to converge in $\text{poly}(n,\beta,\cpi(\nu))$
and $\text{poly}(n,\beta,C_{\msf{LSI}}(\nu))$ iterations under \eqref{eq:pi}
\cite{chewi2021analysis} and \eqref{eq:lsi} \cite{vempala2019rapid}
respectively, assuming $\beta$-Lipschitzness of $\nabla V$. Can
this approach be extended to sampling general logconcave densities
without the gradient Lipschitzness?

\paragraph{Algorithmic diffusion.}

By \emph{algorithmic diffusion}, we mean the general idea of discretizing
a diffusion process and proving guarantees on the discretization error
and query complexity. Recently, \cite{kook2024inout} showed that
such an algorithm (called $\ino$ there), which can be viewed as a
$\ps$, works for sampling from the uniform distribution over a convex
body and recovers state-of-the-art complexity guarantees with substantially
improved output guarantees (R\'enyi-divergences). The guarantee was
extended to the R\'enyi-infinity (or pointwise) distance as well
\cite{kook2024renyi}. A nice aspect of the analysis is that it shows
convergence directly in terms of isoperimetric constants of target
distributions. As a result, it provides a unifying point of view for
analysis without requiring the use of technically sophisticated tools
such as the $\sw$ and $s$-conductance \cite{lovasz1993random,kannan1997random,lovasz2007geometry}.
We discuss the approach in more detail presently.

This brings us to the main motivation of the current paper: \emph{Can
we use algorithmic diffusion to obtain faster algorithms with stronger
guarantees for sampling, rounding, and integrating general logconcave
functions?}

\subsection{Results\label{subsec:Results}}

In this paper, we propose new approaches for general logconcave sampling,
rounding and integration, leading to the first improvements in the
complexity of these problems in nearly two decades \cite{lovasz2006fast}.
Our methods crucially rely on a reduction from general logconcave
sampling to \emph{exponential sampling} in one higher dimension. With
this reduction in hand, we develop a new framework for sampling and
improve the complexity of four fundamental problems: (1) logconcave
sampling from an $\O(1)$-warm start, (2) warm-start generation, (3)
isotropic rounding and (4) integration. For each problem, our improved
complexity for general logconcave distributions matches the current
best complexity for the uniform distribution over a convex body, a
setting that has received much attention for more than three decades
\cite{dyer1991random,lovasz1990mixing,lovasz1991compute,lovasz1993random,kannan1997random,lovasz2006hit,cousins2018gaussian,kook2024inout,kook2024renyi}.
We now present our main results, followed by a detailed discussion
of the challenges and techniques in \S\ref{ssec:techniques}, and
notations and preliminaries in \S\ref{sec:Preliminaries}.

\paragraph{Result 1: Logconcave sampling from a warm start.}

When a logconcave distribution $\pi^{X}\propto e^{-V}$ is presented
with the evaluation oracle for $V$, we instead consider the ``lifted''
distribution 
\[
\pi(x,t)\propto\exp(-nt)\cdot\ind[(x,t):V(x)\leq nt]
\]
in the augmented $(x,t)$-space, the $X$-marginal of which is exactly
$\pi^{X}$ (see Proposition~\ref{prop:x-marginal}). We quantify
in \S\ref{sec:lc-sampling} how this affects parameters pertinent
to sampling (e.g., isoperimetric constants). 

\begin{algorithm}[t]
\hspace*{\algorithmicindent} \textbf{Input:} initial pt. $z_{0}\sim\pi_{0}\in\mc P(\R^{n+1})$,
$\K=\{(x,t):V(x)\leq nt\}$, $k\in\mathbb{N}$, threshold $N$, variance
$h$.

\hspace*{\algorithmicindent} \textbf{Output:} $z_{k+1}=(x_{k+1},t_{k+1})$.

\begin{algorithmic}[1] \caption{$\protect\ps$ $\protect\psexp$\label{alg:prox-exp}}
\FOR{$i=0,\dotsc,k$}

\STATE Sample $y_{i+1}\sim\mc N(z_{i},hI_{n+1})$.\label{line:forward}

\STATE Sample $z_{i+1}\sim\mc N(y_{i+1}-hne_{n+1},hI_{n+1})|_{\K}$.
\label{line:backward}

\STATE$\quad(\uparrow)$ {\small Repeat $z_{i+1}\sim\mc N(y_{i+1}-hne_{n+1},hI_{n+1})$
until $z_{i+1}\in\mc K$. If $\#$attempts$_{i}\,\geq N$, declare
\textbf{Failure}.}

\ENDFOR \end{algorithmic}
\end{algorithm}

To sample from the exponential target $\pi$ using a given $\eu R_{\infty}$
warm-start, we use the $\ps$ (see Figure~\ref{fig:fig1}) with an
evaluation oracle that returns $nt$ conditioned on $V(x)\leq nt$.
One iteration of the sampler is reversible with respect to the stationary
distribution $\pi$, while it can be interpreted as a composition
of two complimentary diffusion processes. Discarding the $T$-component
of a sample, we are left with the $X$-component with law close to
our desired target $\pi^{X}$. In this regard, our sampler is a natural
generalization of $\ino$ \cite{kook2024inout}, the $\ps$ for uniform
sampling. Concretely, we establish mixing guarantees and query complexities
of our sampler for the target logconcave distribution $\pi^{X}$ in
terms of its Poincar\'e constant (which in general satisfies $\norm{\cov\pi^{X}}\leq\cpi(\pi^{X})\lesssim\norm{\cov\pi^{X}}\log n$
\cite{klartag2023logarithmic}). 
\begin{thm}
\label{thm:lc-warmstart-intro} For any logconcave distribution $\pi^{X}$
specified by a well-defined function oracle $\eval(V)$, for any given
$\eta,\veps\in(0,1)$, $q\geq2$, and $\pi_{0}^{X}$ with $\eu R_{\infty}(\pi_{0}^{X}\mmid\pi^{X})=\log M$,
we can use the $\ps$ $\psexp$ (Algorithm~\ref{alg:prox-exp}) with
suitable choices of parameters, so that with probability at least
$1-\eta$, we obtain a sample $X$ such that $\eu R_{q}(\law X\mmid\pi^{X})\leq\veps$,
using $\Otilde(qMn^{2}\,(\norm{\cov\pi^{X}}\vee1)\polylog\frac{1}{\eta\veps})$
evaluation queries in expectation. 
\end{thm}

See \S\ref{subsec:exp-mixing} for a more detailed description of
the heat flow perspective and its benefits for sampling. Our specific
setting of parameters can be found in Theorem~\ref{thm:exp-sampling}.
In comparison with the previous best complexity of $\Otilde(n^{2}\tr(\cov\pi^{X})\polylog\frac{M}{\veps})$,
which is for $\tv$-distance, the $\ps$ $\psexp$ achieves a provably
better rate (since $\norm{\cov\pi^{X}}\leq\tr(\cov\pi^{X})$) from
an $\O(1)$-warm start, and moreover does so in general $\eu R_{q}$-divergences.
This mixing rate matches the previously known best rate for the uniform
sampling by the $\bw$ \cite{kannan1997random} as well as $\ino$
\cite{kook2024inout}.

\begin{figure}
\includegraphics[width=\columnwidth]{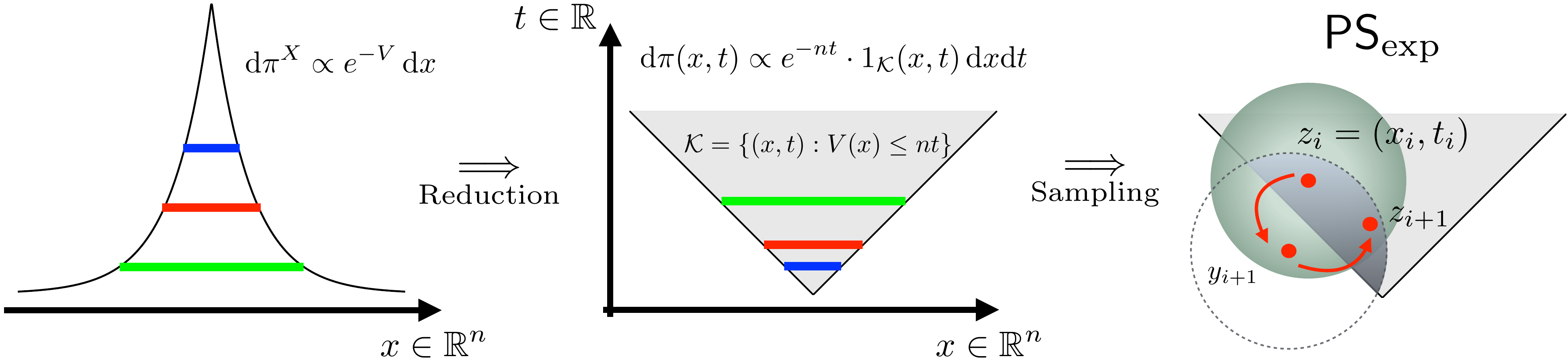} 

\caption{Reduction to an exponential distribution and sampling via the $\protect\ps$
$\protect\psexp$.\label{fig:fig1}}
\end{figure}

\paragraph{Result 2: Warm-start generation (sampling without a warm start).}

The $\ps$ $\psexp$ assumes access to a warm start in $\eu R_{\infty}$,
and generating such a good warm-start is an important and challenging
algorithmic problem in its own right. In \S\ref{sec:warm-start},
we propose $\tgc$ which generates an $\O(1)$-warm start in the $\eu R_{\infty}$-divergence
for any target logconcave distribution. This algorithm generalizes
$\gc$ \cite{cousins2018gaussian}, the known method for generating
a warm-start for the uniform distribution over a convex body. The
high-level idea is to follow a sequence $\{\mu_{i}\}_{i\in[m]}$ of
distributions, where $\mu_{1}$ is easy to sample from, $\mu_{i}$
and $\mu_{i+1}$ are close in some probability divergences, and $\mu_{m}$
is the desired target. 

Our algorithm first reduces the original target to the exponential
distribution and then follows annealing distributions of the form
\[
\mu_{\sigma^{2},\rho}(x,t)\propto\exp\bigl(-\frac{1}{2\sigma^{2}}\,\norm x^{2}-\rho t\bigr)\cdot\ind[\{V(x)\leq nt\}\cap\{\norm x=\O(R),\abs t=\O(1)\}]\,,
\]
with a carefully chosen schedule for updating the parameters $\sigma^{2}$
and $\rho$. Roughly speaking, we update the $X$-marginal through
$\gc$ and the $T$-marginal using an exponential \emph{tilt}. In
order to move across the annealing distributions, $\tgc$ requires
an efficient sampler for the intermediate annealing distributions,
ideally with guarantees in the $\eu R_{\infty}$-divergence in order
to relay $\eu R_{\infty}$-warmness guarantees along the annealing
scheme.  We use the $\ps$ for these intermediate distributions with
a $\eu R_{\infty}$-divergence guarantee, namely the $\ps$ $\psann$
(Algorithm~\ref{alg:prox-ann}) in \S\ref{ssec:prox-annealing}.
In our computational model, this sampler started at a previous annealing
distribution returns a sample with law $\mu$ satisfying $\eu R_{\infty}(\mu\mmid\mu_{\sigma^{2},\rho})\leq\veps$,
using $\Otilde(n^{2}\sigma^{2}\polylog\nicefrac{R}{\eta\veps})$ evaluation
queries in expectation (see Theorem~\ref{thm:annealing-query-comp}).

By sampling these annealing distributions through $\psann$, $\tgc$
obtains an $\O(1)$-warm start for $\pi$ (and hence for the desired
target $\pi^{X}$); then runs the $\ps$ $\psexp$ one final time
to obtain a sample with $\eu R_{\infty}$-guarantees (this is a very
strong notion of probability divergence that recovers all commonly
used distances such as $\tv,\KL,\chi^{2},$ or $\eu R_{q}$).
\begin{thm}
\label{thm:tgc-intro} For any logconcave distributions $\pi^{X}$
specified by $\eval_{x_{0},R}(V)$, for any given $\eta,\veps\in(0,1)$,
$\tgc$ (Algorithm~\ref{alg:tgc}) with probability at least $1-\eta$,
returns a sample with law $\nu$ such that $\eu R_{\infty}(\nu\mmid\pi^{X})\leq\veps$,
using $\Otilde(n^{2}(R^{2}\vee n)\polylog\nicefrac{1}{\eta\veps})$
evaluation queries in expectation. Hence, if $\pi^{X}$ is well-rounded
(i.e., $R^{2}\lesssim n$), then $\Otilde(n^{3}\polylog\nicefrac{1}{\eta\veps})$
queries suffice.
\end{thm}

This improves the prior best complexity of $\Otilde(n^{3}R^{2}\polylog\frac{1}{\veps})$
for general logconcave sampling (in $\tv$-distance) from scratch
by \cite[Collorary 1.2]{lovasz2006fast}, and provides a much stronger
$\eu R_{\infty}$-guarantee. Moreover, this complexity matches the
best-known complexity for uniform sampling with the $\tv$-guarantee
by \cite{cousins2018gaussian} and with the $\eu R_{\infty}$-guarantee
by \cite{kook2024renyi}.

\paragraph{Result 3: Isotropic rounding of logconcave distributions.}

The above results on the complexity of logconcave sampling still have
dependence on the second moment $R^{2}$, so those are not fully polynomial
in the problem parameters. In \S\ref{sec:rounding}, we address this
issue by initially running an algorithm for isotropic rounding which
makes the covariance matrix of a given logconcave distribution near-isotropic
(i.e., $\cov\pi^{X}\approx I_{n}$). After this rounding, $R^{2}=\O(n)$,
and thus the guarantee above turns into $n^{3}$ at the cost of a
multiplicative factor of $\polylog R$ in the complexity. Isotropic
rounding is also a useful tool for many other high-dimensional algorithms. 

At a high level, just as in the warm-start generation, our algorithm
follows a sequence $\{\mu_{i}\}$ of distributions while updating
an affine map $F:\Rn\to\Rn$ along the way. An important subroutine
is to convert a well-rounded distribution (i.e., with $\tr(\cov)=\O(n)$)
to one that is near-isotropic (i.e., with $\cov\approx I_{n}$). To
this end, we generalize the approach taken in \cite{jia2021reducing}
for the uniform distribution --- repeatedly {[}draw a few samples
$\to$ compute a crude estimate of the covariance $\to$ identify
skewed directions from the estimation and upscale them{]}. In doing
so, we rely crucially on our sampler's mixing rate in terms of $\norm{\cov\pi^{X}}$
(instead of $\tr(\cov\pi^{X})$) in \textbf{Result 1} and its improved
complexity for warm-start generation in \textbf{Result 2}.
\begin{thm}
\label{thm:rounding-intro} For any logconcave distribution $\pi^{X}$
specified by  $\eval_{x_{0},R}(V)$, there exists a randomized algorithm
with query complexity of $\Otilde(n^{3.5}\polylog R)$ that finds
an affine map $F$ such that the pushforward of $\pi^{X}$ via $F$
is $1.01$-isotropic with probability at least $1-\O(n^{-1/2})$.
\end{thm}

Our complexity improves upon the previous best complexity of $\Otilde(n^{4}\polylog R)$
for general logconcave distributions \cite[Theorem 6.1]{lovasz2006fast}.
Moreover, our rate matches that of isotropic rounding for the uniform
distribution \cite{jia2021reducing,jia2024reducingisotropyvolumekls}.

\paragraph{Result 4: Integration of logconcave functions.}

Finally, we examine the complexity of integrating a logconcave function,
a classical application of logconcave sampling. To this end, we extend
an annealing approach in \cite{cousins2018gaussian} to general logconcave
integration in \S\ref{sec:integration}. Similar to $\tgc$, we follow
a sequence $\{f_{i}\}_{i\in[m]}$ of logconcave functions, where $f_{1}$
is easy to integrate, $\mu_{i}\propto f_{i}$ is an $\O(1)$-warm
start for $\mu_{i+1}\propto f_{i+1}$, the variance of $f_{i+1}/f_{i}$
with respect to $\mu_{i}$ is small enough (e.g., $\chi^{2}(\mu_{i+1}\mmid\mu_{i})=\O(m^{-1})$),
and $f_{m}$ is the target logconcave function. 
\begin{thm}
\label{thm:integration-intro} For any $\veps>0$ and  an integrable
well-rounded logconcave function $f=e^{-V}:\R^{n}\rightarrow\R$ presented
by $\eval_{x_{0},R}(V)$, there exists an algorithm that with probability
at least $3/4$, returns a $(1+\varepsilon)$-multiplicative approximation
to the integral of $f$ using $\Otilde(n^{3}/\varepsilon^{2})$ queries.
For an arbitrary logconcave $f$ given by a well-defined function
oracle, the complexity is bounded by $\Otilde(n^{3.5}\polylog R+n^{3}/\varepsilon^{2})$. 
\end{thm}

This result improves the previous best complexity $\Otilde(n^{4}/\veps^{2})$
of integrating logconcave functions \cite[Theorem 1.3]{lovasz2006fast}.
Again, our improved complexity matches that for the uniform distribution
(i.e., volume computation) \cite{cousins2018gaussian,jia2021reducing,jia2024reducingisotropyvolumekls}.
We note that due to the strong guarantees of the new sampler, we can
streamline the analysis of errors and dependence among samples used
to estimate $f_{i+1}/f_{i}$, simplifying and strengthening earlier
analyses \cite{kannan1997random,lovasz2006simulated,cousins2018gaussian}.

\subsection{Techniques and background\label{ssec:techniques}}

In proving our results --- faster algorithms for the basic problems
of sampling (with and without a warm start), rounding and integration
--- we do not take the most direct route. Instead we develop several
techniques that appear to be interesting and we imagine will be useful
in other contexts as well. These include the reversible heat-flow
perspective for the design of polynomial-time algorithms providing
a direct connection to isoperimetric constants, sampling guarantees
in the strong notion of $\eu R_{\infty}$-distance (which has other
motivations such as differential privacy), relevant geometry of logconcave
functions, and a streamlined analysis of estimation errors of dependent
samples. 

\subsubsection{Going beyond the uniform distribution}

The uniform distribution over a convex body is  a special case of
a logconcave distribution. The best-known polynomial complexity for
uniform sampling is achieved by addressing three subproblems: (1)
sampling from a warm start, (2) generating a warm start and (3) performing
isotropic rounding. Extending the best-known guarantees and algorithms
for each of these problems to general logconcave densities is the
primary challenged addressed in this paper. In our proofs, we use
general properties of logconcave functions and avoid any structural
assumptions on the target density. 

The first challenge is to establish a logconcave sampler with a mixing
rate matching that of the best uniform samplers, such as the $\bw$
\cite{kannan1997random} or $\ino$ \cite{kook2024inout}. These uniform
samplers have a mixing rate of $n^{2}\norm{\cov\pi^{X}}$, while $\har$
in \cite{lovasz2006fast}, previously the best for general logconcave
distributions, has a rate of $n^{2}\tr(\cov\pi^{X})$. The improved
complexity bounds of these uniform samplers come from a better understanding
of isoperimetric properties, such as the Cheeger and Poincar\'e constants
for logconcave distributions, and this aspect is more direct for $\ino$.
Here we extend $\ino$ to general logconcave distributions. As we
elaborate in \S\ref{sssec:lc-sampling-warm}, this requires implementing
rejection sampling for the distributions of the form $\exp(-V(x)-\frac{1}{2h}\,\norm x^{2})$
(for a suitable $h$) using $\Otilde(1)$ queries. However, with only
logconcavity assumed for $e^{-V}$, determining a suitable proposal
distribution and analyzing its query complexity present significant
challenges.

\paragraph{Reduction to an exponential distribution.}

To extend methods from uniform to general logconcave distributions,
we leverage a conceptual connection between logconcave sampling and
convex optimization --- sampling from $\exp(-V)$ is analogous to
minimization of $V$. In optimization, we bound $V(x)$ with a new
variable $t$, and add the convex constraint $V(x)\leq t$. Inspired
by this, we consider an exponential distribution with density $\pi(x,t)\propto\exp(-nt)\cdot\ind[(x,t):V(x)\leq nt]$
on $\R^{n+1}$, with $X$-marginal $\pi^{X}\propto\exp(-V)$ (Proposition~\ref{prop:x-marginal}).
The scaling of $n$ in the exponential function is a natural choice
as we will see later. 

This reduction offers several advantages. First, the potential $V$
becomes linear, opening up a possible extension of previous analyses.
Second, it points to a clear path to generalizing ideas/guarantees
devised for uniform sampling. Specifically, the conditional law $\pi^{X|T=t}$
of $X$ given $T=t$ is the uniform distribution over the convex level
set $L_{nt}=\{V(x)\leq nt\}$. Thus, $\pi$ can be interpreted as
an average of the uniform distribution over $L_{nt}$ weighted by
$e^{-nt}$, and this idea is crucial to our subsequent algorithms.

A similar exponential reduction was discussed conceptually in \cite{kook2022sampling}
and developed further in \cite{kook2024gaussian} for a sampling analogue
of the interior-point method from optimization. Here, however, we
apply this reduction in our general setup, without specific assumptions
on the epigraph, and analyze structural properties of the reduced
distribution. For example, we bound the mean and variance along the
$T$-direction. This allows us to analyze how this reduction impacts
key parameters of the original logconcave distribution, including
the largest eigenvalue and trace of the covariance. 

\subsubsection{Logconcave sampling from a warm start\label{sssec:lc-sampling-warm}}

\paragraph{Sampling through diffusion.}

The proximal sampler \cite{lee2021structured} is essentially a Gibbs
sampler from two conditional distributions. For a target distribution
$\pi^{X}\propto\exp(-V)$, it introduces a new variable $Y$ and considers
the augmented distribution $\pi^{X,Y}\propto\exp(-V(x)-\frac{1}{2h}\,\norm{x-y}^{2})$
with parameter $h>0$. One iteration involves two steps: (i) $y\sim\pi^{Y|X=x}=\mc N(x,hI_{n})$
and (ii) $x\sim\pi^{X|Y=y}\propto\exp(-V(\cdot)-\frac{1}{2h}\,\norm{\cdot-y}^{2})$.
While step (i) is straightforward, step (ii) requires a nontrivial
sampling procedure.

The complexity of the $\ps$ involves a mixing analysis (to determine
the number of iterations for a desired accuracy to the target) and
the query complexity of implementing step (ii). For the first part,
\cite{chen2022improved} demonstrated that one iteration corresponds
to simulating the heat flow for (i) and then a \emph{time reversal}
of the heat flow for (ii) (see \S\ref{subsec:exp-mixing} for details).
This leads to exponential decay of the $\eu R_{q}$-divergence, with
the decay rate dependent on the isoperimetry of the target distribution
(e.g, \eqref{eq:pi} and \eqref{eq:lsi}). For the second part, prior
approaches typically use rejection sampling or a different logconcave
sampler (e.g., $\msf{MALA}$, $\msf{ULMC}$) under additional assumptions
such as the smoothness of $V$ or access to first-order and proximal
oracles for $V$.

\paragraph{Sampling without smoothness.}

Previous studies of the $\ps$ focused on smooth unconstrained distributions,
where the potential satisfies $\hess V\preceq\beta I_{n}$ for some
$\beta<\infty$, leaving open the complexities of uniform or general
logconcave sampling with hard constraints. \cite{kook2024inout} introduced
$\ino$, a version of the $\ps$ for uniform sampling over convex
bodies. With $V(x)=(\ind[x\in\K])^{-1}$ for a convex body $\K$,
one iteration of $\ino$ draws $y\sim\mc N(x,hI_{n})$ and then $x\sim\mc N(y,hI_{n})|_{\K}$
(the Gaussian truncated to $\mc K$) using rejection sampling on the
proposal $\mc N(y,hI_{n})$. They introduce a \emph{threshold} parameter
on the number of rejection trials which ensures that the algorithm
does not use too many queries. 

This diffusion-based approach turns out to be stronger than the $\bw$,
the previous best uniform sampler with query complexity $\Otilde(n^{2}\norm{\cov\pi^{X}}\polylog\frac{1}{\veps})$
for obtaining an $\veps$-close sample in $\tv$-distance. $\ino$
provides $\eu R_{q}$-divergence guarantees with a matching rate of
$n^{2}\norm{\cov\pi^{X}}$ for $\tv$-distance, and its analysis is
simpler than the $\bw$. Since the latter's analysis in \cite{kannan1997random}
goes through its biased version (called the $\sw$), it involves an
understanding of an additional rejection step for making the biased
distribution close to the uniform target, as well as the isoperimetric
constant of the biased distribution. In contrast, $\ino$ achieves
direct contraction towards the uniform target, with rate dependent
on isoperimetric constants of the original target, not a biased one,
and achieves stronger output guarantees. $\ino$'s approach has also
been extended to truncated Gaussian sampling \cite{kook2024renyi},
leading to a $\eu R_{\infty}$-guarantee that improves upon the $\tv$-guarantee
of the $\bw$ \cite{cousins2018gaussian}.

\paragraph{Proximal sampler for general logconcave distributions.}

This prompts a natural question if we can extend $\ino$ beyond constant
and quadratic potentials. Prior studies of the $\ps$ by \cite{chen2022improved,kook2024inout}
allow for an immediate mixing analysis through the isoperimetry of
logconcave distributions. However, we also need a query complexity
bound for sampling from $\pi^{X|Y=y}\propto\exp(-V(x)-\frac{1}{2h}\,\norm{x-y}^{2})$.
While one might consider using $\har$ \cite{lovasz2006fast} for
the second step, it requires \emph{roundedness} of this distribution,
and its current complexity is also unsatisfactory. Rejection sampling
is another option, but without smoothness assumptions for general
logconcave distributions, it is challenging to identify an appropriate
proposal distribution and to bound the expected number of trials.

To handle an arbitrary convex potential $V$, we apply the exponential
reduction, turning $V$ into a convex constraint and work with the
linear potential $nt$ in one higher dimension. Extending $\ino$
to this exponential distribution, we propose the $\ps$ $\psexp$
(Algorithm~\ref{alg:prox-exp}). Since $\cpi(\pi)\lesssim(\cpi(\pi^{X})+1)\log n$
(Lemma~\ref{lem:exp-cov}), the mixing rate of $\psexp$ for $\pi$
is close to that of the $\ps$ for $\pi^{X}$. Implementing the second
step is now simpler due to the linear potential, allowing us to extend
the analysis of $\ino$ for uniform and truncated Gaussian distributions
to exponential distributions (Lemma~\ref{lem:exp-perstep}). A new
technical ingredient is bounding the rate of increase of $\int_{\K_{\delta}}e^{-nt}/\int_{\K}e^{-nt}$
as $\delta$ grows, where $\K_{\delta}=\{z\in\R^{n+1}:d(z,\K)\leq\delta\}$.
We show that this growth rate is bounded by $e^{\delta n}$ (Lemma~\ref{lem:relaxed-regularity}
and~\ref{lem:exp-expansion}).

\subsubsection{Warm-start generation\label{subsec:warm-start-technical}}

Prior work on warm-start generation \cite{kannan1997random,lovasz2006simulated,cousins2018gaussian}
is based on using a sequence $\{\mu_{i}\}_{i\in[m]}$ of distributions,
where $\mu_{1}$ is easy to sample (e.g., uniform distribution over
a unit ball), each $\mu_{i}$ is close to $\mu_{i+1}$ in probabilistic
distance, and $\mu_{m}$ is the target. By moving along this sequence
with a suitable sampler for intermediate annealing distributions,
one can generate a warm start for the desired distribution. This approach
is more efficient than trying to directly go from $\mu_{1}$ to $\mu_{m}$.

\paragraph{Warm-start generation for uniform distributions.}

The state-of-the-art algorithm for uniform distributions over convex
bodies is $\gc$ \cite{cousins2018gaussian}. They set $\mu_{i}(x)\propto\exp(-\norm x^{2}/(2\sigma_{i}^{2}))$,
where $\sigma_{i}^{2}$ increases from $n^{-1}$ to $R^{2}$ according
to a suitable schedule. This approach ensures that $\mu_{i}$ is $\O(1)$-warm
with respect to $\mu_{i+1}$, and uses the $\bw$ to sample from truncated
Gaussians with $\tv$-guarantees. Using a coupling argument, they
showed that with high probability, this scheme outputs an $\veps$-close
sample to the uniform distribution in $\tv$-distance, using $\Otilde(n^{2}(R^{2}\vee n)\polylog\frac{1}{\veps})$
membership queries. \cite{kook2024renyi} later improved this by replacing
the $\bw$ with the $\ps$, achieving $\eu R_{\infty}$-guarantees
with the same complexity. As a result, $\gc$ transfers $\eu R_{\infty}$-warmness
across the sequence of distributions, and thus the complexity for
uniform sampling from a convex body remains the same even for $\eu R_{\infty}$-divergence
guarantees.

\paragraph{Going beyond uniform distributions.}

In \cite{cousins2018gaussian}, Cousins and Vempala raised the question
of whether their annealing strategy can be extended to arbitrary logconcave
distributions with complexity $\Otilde(n^{2}(R^{2}\vee n))$. A natural
choice for annealing distributions is $\mu(x)\propto\exp(-V(x)-\frac{1}{2\sigma^{2}}\,\norm x^{2})$,
which would still provide $\eu R_{\infty}$-closeness of consecutive
distributions and allow for accelerated updates to $\sigma^{2}$.
However, prior samplers lack the necessary guarantees for these distributions,
so we use the exponential reduction.

To generate a warm start for $\pi(x,t)\propto e^{-nt}|_{\K}$, it
seems natural to consider an annealing distribution obtained by multiplying
$\pi$ by a Gaussian in ``$(x,t)$'' for a direct application of $\gc$
in $\R^{n+1}$. However, due to different rates of changes in the
quadratic term and linear term in $t$ over an interval of length
$\O(1)$, these two terms do not properly cancel each other, which
implies that the warmness of $\mu_{m-1}$ with respect to $\mu_{m}=\pi$
is no longer $\O(1)$-bounded.

We address this by introducing $\mu_{\sigma^{2},\rho}(x,t)\propto\exp(-\frac{1}{2\sigma^{2}}\,\norm x^{2}-\rho t)\cdot\ind[(x,t)\in\bar{\K}]$
where $\sigma^{2}\in(0,R^{2}]$, $\rho\in(0,n]$, and $\bar{\K}=\{V(x)\leq nt\}\cap\{\norm x=\O(R),\norm t=\O(1)\}$.
Essentially, this runs $\gc$ along the $x$-direction with an exponential
tilt in the $t$-direction, ensuring that consecutive distributions
are $\O(1)$-close in $\eu R_{\infty}$. However, we need an efficient
sampler for $\mu_{\sigma^{2},\rho}$ with $\eu R_{\infty}$-guarantees
to maintain $\eu R_{\infty}$-warmness across the annealing scheme.

\paragraph{Sampling from annealing distributions.}

Given the form of the annealing distribution (the potential is a combination
of linear and quadratic terms), we use the $\ps$ to develop $\psann$
(Algorithm~\ref{alg:prox-ann}). The query complexity for rejection
sampling in the second step now can be derived from our analysis of
$\psexp$ (Lemma~\ref{lem:exp-perstep}) and the $\ps$ for truncated
Gaussians \cite{kook2024renyi}. For mixing with $\eu R_{\infty}$
guarantees, we apply a technique from \cite{kook2024renyi}. We obtain
a mixing rate of $\O(h^{-1}\clsi(\mu_{\sigma^{2},\rho}))$ for the
$\ps$ based on \eqref{eq:lsi} (Proposition~\ref{prop:ps-smooth}).
This results in only a \emph{doubly logarithmic} dependence on the
initial warmness, implying convergence from any \emph{feasible} start
with an overhead of $\polylog(n,R)$ (Lemma~\ref{lem:ps-annealing-iter}).
This \emph{uniform ergodicity} implies $L^{\infty}$-norm contraction
of density toward the target \cite{del2003contraction} (see Lemma~\ref{lem:boosting}),
leading to a $\eu R_{\infty}$-guarantee of $\psann$ without significant
overhead (Theorem~\ref{thm:warm-start}). 

In Lemma~\ref{lem:lsi-annealing}, we bound the LSI constant of the
annealing distribution $\mu_{\sigma^{2},\rho}$ by $\sigma^{2}\vee1$
via the Bakry-\'Emery criterion (Lemma~\ref{lem:bakry-emery}) and
Holley--Stroock perturbation principle (Lemma~\ref{lem:lsi-bdd-perturbation}).
For $\nu\propto\exp(-\frac{1}{2\sigma^{2}}\,\norm x^{2}-\frac{t^{2}}{2}-\rho t)|_{\bar{\K}}$,
since the potential of $\nu$ is $\min(\sigma^{-2},1)$-strongly convex,
its LSI without convex truncation is bounded by $\sigma^{2}\vee1$
through Bakry-\'Emery, and convex truncation to $\bar{\K}$ only
helps in satisfying the criterion \cite{bakry2014analysis}. Also,
as $\sup t-\inf t=\Theta(1)$ over $\bar{\K}$, the ratio of $\mu_{\sigma^{2},\rho}$
to $\nu$ is bounded below and above by $\Theta(1)$, so the perturbation
principle ensures that $\clsi(\mu_{\sigma^{2},\rho})\lesssim\clsi(\nu)\leq\sigma^{2}\vee1$.

\paragraph{Gaussian cooling with exponential tilt.}

With the query complexity of $\psann$ in mind, we design $\tgc$
(Algorithm~\ref{alg:tgc}) for warm-start generation for $\pi(x,t)\propto\exp(-nt)|_{\K}$.
We run rejection sampling with proposal $\mc N(0,n^{-1}I_{n})\otimes\text{Unif}\,(I_{t})$
for some interval $I_{t}$ of length $\Theta(1)$, and initial distribution
$\exp(-\frac{n}{2}\,\norm x^{2})|_{\bar{\K}}$, which is $\O(1)$-warm
with respect to $\mu_{1}=\exp(-\frac{n}{2}\,\norm x^{2}-t)|_{\bar{\K}}$.
In Phase I, we update the two parameters according to $\sigma^{2}\gets\sigma^{2}(1+n^{-1})$
and $\rho\gets\rho(1+n^{-1})$ while $\rho\leq n$ and $\sigma^{2}\leq1$.
Since Phase I involves $\Otilde(n)$ inner phases with complexities
of sampling from each annealing distributions being $\Otilde(n^{2}(\sigma^{2}\vee1))=\Otilde(n^{2})$,
the total complexity is $\Otilde(n^{3})$. In Phase II, we accelerate
$\sigma^{2}$-updates via $\sigma^{2}\gets\sigma^{2}(1+\sigma^{2}/R^{2})$
as in $\gc$. With $\Otilde(R^{2}/\sigma^{2})$ inner phases (for
doubling of $\sigma^{2}$) and sampling complexity $\Otilde(n^{2}\sigma^{2})$
per inner phase, this has total complexity $\Otilde(n^{2}R^{2})$.
At termination, $\psexp$ is run with $\mu_{R^{2},n}$ as the initial
distribution for target $\bar{\pi}\propto e^{-nt}|_{\bar{\K}}$, where
these two are close in $\eu R_{\infty}$. Using the LSI of $\bar{\pi}$
and the boosting scheme again, we can achieve an $\veps$-close sample
to $\pi$ (not $\bar{\pi}$) in the $\eu R_{\infty}$-divergence,
using $\Otilde(n^{2}(R^{2}\vee n)\polylog\frac{1}{\veps})$ queries
in total (Lemma~\ref{thm:warm-start}). 

\subsubsection{Rounding}

Rounding is the key to reducing the dependence on $R$ from $\poly R$
to $\polylog R$. 

\paragraph{Isotropic rounding for uniform distributions.}

The previous best rounding algorithm for uniform distributions, proposed
by \cite{jia2021reducing}, gradually isotropizes a sequence $\{\mu_{i}\}$
of distributions. For a convex body $\K$, they set $\mu_{i}=\text{Unif}\,(\K\cap B_{\delta^{i}}(0))$
for $\delta=1+n^{-1/2}$, increasing $i$ while $\delta^{i}\leq R$.
Their approach entails two important tasks: \underline{Outer loop}:
if $F(\K\cap B_{r})$ is near-isotropic for an affine map $F$, then
show that $F(\K\cap B_{\delta r})$ is well-rounded, and \underline{Inner loop}:
design an algorithm of $n^{3}$-complexity that isotropizes a given
well-rounded uniform distribution. The first task was accomplished
by \cite{jia2024reducingisotropyvolumekls} through Paouris' lemma
(i.e., exponential tail decay) and a universal property that the diameter
of an isotropic convex body is bounded by $n+1$.

The second task was addressed by repeating {[}draw samples $\to$
compute crude covariance estimation $\to$ upscale skewed directions
of the covariance estimation{]}. They first run $\gc$ to obtain a
warm start for a uniform distribution $\mu$ from a convex body. Then,
when the inner radius is $r$, the $\bw$ (or $\ino$) is used to
generate $r^{2}$ samples approximately distributed according to $\mu$.
These $r^{2}$ samples give a rough estimate of the covariance matrix
$\overline{\Sigma}$ of $\mu$ such that $\abs{\overline{\Sigma}-\Sigma}\precsim nI_{n}$,
where $\Sigma=\cov\mu$. Since the query complexity of these uniform
samplers is $n^{2}\norm{\Sigma}/r^{2}$, this procedure uses $n^{2}\norm{\Sigma}$
queries in total. Then, it computes the eigenvalue/vectors of $\overline{\Sigma}$
and scales up (by a factor of $2)$ the subspace spanned by eigenvectors
with eigenvalues less than $n$. One iteration of this process achieves
two key properties: the largest eigenvalue of the covariance $\Sigma$
increases by at most $n$ additively while $r=\inrad\mu$ \emph{almost}
doubles. Since the well-roundedness ensures that $\norm{\Sigma}=\O(n)$
initially, the complexity of one iteration remains as $\Otilde(n^{3})$
throughout. Since there are $\sqrt{n}\log R$ outer iterations, the
algorithm uses a total of $\Otilde(n^{3.5}\polylog R)$ queries.

\paragraph{Extension to general logconcave distributions.}

Our rounding algorithm essentially follows this approach, with several
technical refinements. First of all, we define a \emph{ground set}
of a general logconcave distribution, namely the level set $\msf L_{g}:=\{x:V(x)-\min V\leq10n\}$.
This ground set takes up most of measure due to the universal property
in terms of the potential value (Lemma~\ref{lem:LC-tail}). Focusing
on the ground set is the first step toward a streamlined extension
of the previous approach, so we consider the \emph{grounded distribution}
$\nu^{X}:=\pi^{X}|_{\msf L_{g}}$. Then, for $\nu_{r}^{X}:=\nu^{X}|_{B_{r}(0)}$,
we isotropize a sequence of distributions, given by $\nu_{1}^{X}\to\nu_{\delta}^{X}\to\nu_{\delta^{2}}^{X}\to\cdots\to\nu_{D}^{X}\to\nu^{X}\to\pi^{X}$
for $D=\Theta(R)$.

To analyze the outer loop, we show that for an affine map $F^{X}$
between $\Rn$, if $F_{\#}^{X}\nu_{r}^{X}$ is near-isotropic, then
$(4F^{X})_{\#}\nu_{\delta r}^{X}$ is well-rounded. Unfortunately,
a varying density of $F_{\#}^{X}\nu_{\delta r}^{X}$ poses a daunting
challenge in extending the previous proof in \cite{jia2024reducingisotropyvolumekls}
to general logconcave distributions. Nonetheless, we can resolve this
issue by working again with the exponential reduction (Lemma~\ref{lem:outer-loop}).
This extension asks for a universal property that an isotropic grounded
distribution has diameter of order $\O(n)$, similar to isotropic
uniform distributions. We show this in Lemma~\ref{lem:diameter-level-set}.
Transferring roundedness in the steps $\nu_{D}^{X}\to\nu^{X}$ and
$\nu^{X}\to\pi^{X}$ is relatively straightforward by combining the
change of measures and the reverse H\"older inequality for logconcave
distributions.

For the inner loop, we can still apply the algorithm from \cite{jia2021reducing}
(or its streamlined version in \cite{kook2024covariance}) with only
minor changes to constants. First, with the logconcave sampler $\psexp$
whose mixing rate depends on $\norm{\cov\pi^{X}}$ rather than $\tr(\cov\pi^{X})$,
the complexity analysis of the inner loop extends naturally to general
logconcave distributions. Next, we note that the proofs for controlling
$\norm{\Sigma_{i}}$ and $\tr\Sigma_{i}$ are identical to those for
uniform distributions. The proof for the doubling of the inner radius
at each iteration is nearly the same as in the uniform case, since
the existence of a large ball due to isotropy is also a \emph{universal}
property of logconcave distributions (see Lemma~\ref{lem:ball-in-isotropy}).

\subsubsection{Integration}

For integration, we use the stronger guarantees of our logconcave
sampler, along with a version of the $\tgc$ scheme to obtain a cubic
algorithm for \emph{well-rounded} logconcave functions. For general
logconcave functions, we use the rounding algorithm as a pre-processing
step, then apply the integration algorithm to the near-isotropic distribution
obtained after rounding. 

\paragraph{Volume computation through annealing.}

Similar to sampling, prior volume algorithms also follow a sequence
$\{f_{i}\}_{i\in[m]}$ of logconcave functions, moving across distributions
$\mu_{i}\propto f_{i}$ using a logconcave sampler. The annealing
scheme is designed in a way that $f_{1}$ is easy to integrate, $\mu_{i}$
is an $\O(1)$-warm start for $\mu_{i+1}$, the variance of the estimator
$E_{i}=f_{i+1}/f_{i}$ with respect to $\mu_{i}$ is bounded by 
$m^{-1}\E_{\mu_{i}}[E_{i}^{2}]$ (i.e., $\chi^{2}(\mu_{i+1}\mmid\mu_{i})=\O(m^{-1})$),
and $f_{m+1}=f$ is the target logconcave function. Since $\E_{\mu_{i}}E_{i}=\int f_{i+1}/\int f_{i}$,
accurate estimations of all $E_{i}$ guarantee that the product $\int f_{1}\times E_{1}\cdots E_{m}$
is a good estimator of $\int f_{1}\cdot\E E_{1}\cdots\E E_{m}=\int f$.

The best-known algorithm, $\gc$ \cite{cousins2018gaussian}, uses
unnormalized Gaussian densities $f_{i}(x)=\exp(-\nicefrac{\norm x^{2}}{2\sigma_{i}^{2}})\cdot\ind_{\K}(x)$
for a convex body $\K$, along with an update schedule for $\sigma_{i}^{2}$.
As described earlier in \S\ref{subsec:warm-start-technical}, its
uses the $\bw$ to sample from $\mu_{i}\propto f_{i}$ with $\tv$-guarantees. 

\paragraph{Extension to general logconcave functions.}

Our integration algorithm follows this approach, but once again in
the lifted space. With $\psann$ used for sampling, we follow a modified
version of $\tgc$ for ease of analysis, particularly for variance
control. We use $f_{i}(x,t)=\exp(-\nicefrac{\norm x^{2}}{2\sigma_{i}^{2}}-\rho_{i}t)|_{\K}$
as the intermediate annealing functions, and ensure that $\eu R_{\infty}(\mu_{i}\mmid\mu_{i+1})=\O(1)$
for efficient sampling from $\mu_{i+1}$, and $\eu R_{2}(\mu_{i+1}\mmid\mu_{i})=\exp(\O(m^{-1}))$
for $\var_{\mu_{i}}(E_{i}/\E_{\mu_{i}}E_{i})=\O(m^{-1})$. Since $\O(1)$-warmness
can be shown as in $\tgc$, we elaborate on technical tools for variance
control along with design of the algorithm. In Phase I, we go from
$\exp(-\nicefrac{n\norm x^{2}}{2})|_{\K}$ to $\exp(-\nicefrac{\norm x^{2}}{2})|_{\K}$
with the update $\sigma^{2}\gets\sigma^{2}(1+\nicefrac{1}{n})$, where
variance control is achieved by the logconcavity of $a\mapsto a^{n}\int h^{a}$
for a logconcave function $h$ \cite{kalai2006simulated}. In Phase
II, we move from $\exp(-\nicefrac{\norm x^{2}}{2}-t)|_{\K}$ to $\exp(-\nicefrac{\norm x^{2}}{2}-nt)|_{\K}$
with the update $\rho\gets\rho\,(1+\nicefrac{1}{n})$, where variance
control follows from the previous lemma and another lemma in \cite{cousins2018gaussian}.
Lastly in Phase III, we move from $\exp(-\nicefrac{\norm x^{2}}{2}-nt)|_{\K}$
to $\exp(-nt)|_{\K}$ with the update $\sigma^{2}\gets\sigma^{2}(1+\nicefrac{\sigma^{2}}{n})$,
and we use the lemma in \cite{cousins2018gaussian} again for variance
control.

\paragraph{Streamlined statistical analysis.}

If all samples used are \emph{independent,} and the estimators have
moderate variance, namely $\var_{\mu_{i}}(\nicefrac{E_{i}}{\E E_{i}})=\O(m^{-1})$,
then $\var(E_{1}\cdots E_{m})=(\prod(1+m^{-1})-1)\,(\E E_{i})^{2}\approx\O(1)\,(\E[E_{1}\cdots E_{m}])^{2}$,
which implies concentration of the estimator around $\E[E_{1}\cdots E_{m}]=\E E_{1}\cdots\E E_{m}$
through Chebyshev's inequality. However, samples given by (say) the
$\bw$ in $\gc$ are \emph{approximately} distributed according to
$\mu_{i}$. In fact, samples drawn from $\mu_{i}$ and $\mu_{i+1}$
are \emph{dependent}, since a sample from $\mu_{i}$ is used as a
warm-start for $\mu_{i+1}$. This means that the estimator $E_{1}\cdots E_{m}$
is a \emph{biased estimator} of $\E E_{1}\cdots\E E_{m}(\neq\E[E_{1}\cdots E_{m}])$.
These two issues complicate statistical analysis in prior work on
how close the estimator is to the integral of $f$, and require additional
technical tools to address them. For the first issue, previous work
used a coupling argument based on $\tv$-distance (referred to as
``divine intervention'') to account for the effects of approximate
distributions (rather than exactly $\mu_{i}$). For the second, they
used the notion of $\alpha$-mixing \cite{rosenblatt1956central}
(referred to as ``$\mu$-independence'' therein) to bound the bias
of the product estimator.

We simplify this statistical analysis substantially by using stronger
guarantees of our sampler $\psann$. For the first issue, when $\bar{\mu}_{i}$
denotes the actual law of a sample $X_{i}$ satisfying $\eu R_{\infty}(\bar{\mu}_{i}\mmid\mu_{i})\leq\veps$,
we can notice that the probability of a bad event (any event) with
respect to $\bar{\mu}_{i}$ (instead of $\mu_{i}$) only increases
by at most a multiplicative factor of $1+\veps$. Since the mixing
time of $\psann$ has a polylogarithmic dependence on $\nicefrac{1}{\veps}$,
we can set $\veps$ polynomially small (i.e., $\veps\gets\veps/\poly(n,R)$),
so we can enforce $\eu R_{\infty}(\otimes_{i}\bar{\mu}_{i}\mmid\otimes_{i}\mu_{i})\leq\veps$
 without a huge overhead in the query complexity. For the second
issue, we use notion of $\beta$-mixing (or the coefficient of absolute
regularity) \cite{kolmogorov1960strong}, which is stronger than $\alpha$-mixing.
This quantity basically measures the discrepancy between a joint distribution
and the product of marginal distributions by $\norm{\law(X_{i},X_{i+1})-\law X_{i}\otimes\law X_{i+1}}_{\tv}=\norm{\law(X_{i},X_{i+1})-\bar{\mu}_{i}\otimes\bar{\mu}_{i+1}}_{\tv}\leq\beta$.
Since the mixing of $\psann$ via \eqref{eq:lsi} ensures mixing from
\emph{any start} (Lemma~\ref{lem:ps-annealing-iter}), we can easily
bound $\beta$ by $\O(\veps)$ (see \eqref{eq:beta-bound} for details).
Thus, when analyzing the probability of a bad event, we can replace
$\otimes\bar{\mu}_{i}$ with $\law(X_{1},\dots,X_{m})$ at the additive
cost of $\O(m\veps)$ in probability (we can replace $\veps\gets m\veps$
once again).

\subsection{Preliminaries\label{sec:Preliminaries}}

Let $\mc P(\Rn)$ be the family of probability measures (distributions)
on $\Rn$ that are absolutely continuous with respect to the Lebesgue
measure. We use the same symbol for a distribution and density. For
a set $S$ and its indicator function $\ind_{S}(x)=[x\in S]$, we
use $\mu|_{S}$ to denote a distribution $\mu$ truncated to $S$
(i.e., $\mu|_{S}\propto\mu\cdot\ind_{S}$). For a measurable map $F:\Rn\to\Rn$
and $\mu\in\mc P(\Rn)$, the \emph{pushforward} measure $F_{\#}\mu$
is defined as $F_{\#}\mu(A)=(\mu\circ F^{-1})(A)$ for a measurable
set $A$. For $a,b\in\R$, we use $a\vee b$ and $a\wedge b$ to indicate
their maximum and minimum, respectively. We use $B_{r}^{n}(x_{0})$
to denote the $n$-dimensional ball of radius $r>0$ centered at $x_{0}\in\Rn$,
dropping the superscript $n$ if there is no confusion. We use $a=\Otilde(b)$
to denote $a=\mc O(b\polylog b)$. For two symmetric matrices $A,B$,
we use $\abs A\preceq B$ to denote $-B\preceq A\preceq B$. Unless
specified otherwise, for a vector $v$ and a PSD matrix $M$, $\norm v$
and $\norm{\Sigma}$ refer to the $\ell_{2}$-norm of $v$ and the
operator norm of $\Sigma$, respectively. We use $\overline{\R}$
to denote the extended real number system $\R\cup\{\pm\infty\}$.

We recall notions of common probability divergences/distances between
distributions. 
\begin{defn}
\label{def:p-dist} For $\mu,\nu\in\mc P(\Rn)$, the \emph{$f$-divergence}
of $\mu$ towards $\nu$ with $\mu\ll\nu$ is defined as, for a convex
function $f:\R_{+}\to\R$ with $f(1)=0$ and $f'(\infty)=\infty$,
\[
D_{f}(\mu\mmid\nu):=\int f\bpar{\frac{\D\mu}{\D\nu}}\,\D\nu\,.
\]
For $q\in(1,\infty)$, the \emph{$\KL$-divergence and $\chi^{q}$-divergence}
correspond to $f(x)=x\log x$ and $x^{q}-1$, respectively. The \emph{$q$-R\'enyi
divergence} is defined as 
\[
\eu R_{q}(\mu\mmid\nu):=\tfrac{1}{q-1}\log\bpar{\chi^{q}(\mu\mmid\nu)+1}\,.
\]
The \emph{R\'enyi-infinity divergence} is defined as
\[
\eu R_{\infty}(\mu\mmid\nu):=\log\esssup_{\mu}\frac{\D\mu}{\D\nu}\,.
\]
A distribution $\mu$ is said to be \emph{$M$-warm with respect to
a distribution} $\nu$ if $\frac{\mu(S)}{\nu(S)}\le M$ for any measurable
subset $S$ (i.e., $\mu$ is $\exp(\eu R_{\infty}(\mu\mmid\nu))$-warm\emph{
}with respect to $\nu$). The \emph{total variation} (TV) distance
for $\mu,\nu\in\mc P(\Rn)$ is defined by 
\[
\norm{\mu-\nu}_{\msf{TV}}:=\frac{1}{2}\int|\mu(x)-\nu(x)|\,\D x=\sup_{S\in\mc F}\abs{\mu(S)-\nu(S)}\,,
\]
where $\mc F$ is the collection of all measurable subsets of $\Rn$. 
\end{defn}

We recall $\KL=\lim_{q\downarrow1}\eu R_{q}\leq\eu R_{q}\leq\eu R_{q'}\leq\eu R_{\infty}$
for $1\leq q\leq q'$ and $2\,\norm{\cdot}_{\tv}^{2}\leq\KL\leq\eu R_{2}=\log(\chi^{2}+1)\leq\chi^{2}$.
We refer readers to \cite{van2014renyi} for basic properties of the
R\'enyi-divergence (e.g., continuity/monotonicity in $q$). The data
processing inequality is given in Lemma~\ref{lem:DPI} for reference.

\section{Sampling from logconcave distributions\label{sec:lc-sampling}}

In this section, we study the complexity of sampling from a logconcave
distribution under the setup stated in the introduction.  We remark
that the ``roundness'' (or regularity) of the distribution is specified
through the existence of a ball in a level set of \textbf{constant
measure}. We find it more convenient to work with the \textbf{values
of the potential} (see Lemma~\ref{lem:relaxed-regularity}). More
precisely, we will focus on a particular level set, called the \emph{ground
set} of $\pi$, defined as 
\begin{equation}
\msf L_{\pi,g}:=\{x\in\Rn:V(x)-\min V\leq10n\}\,.\label{eq:ground_set}
\end{equation}
Note that the ground set of $\pi$ is \emph{independent} of the normalization
constant of $\pi$, which justifies using $\pi$ in the subscript
to indicate its uniqueness. We use $\msf L_{V,g}$ to reveal a specific
choice of the potential $V$. We also define the notation: $V_{0}:=\min V$,
and a level set $\msf L_{\pi,l}:=\{x\in\Rn:V(x)-\min V\leq nl\}$.
\ul{Hereafter, in $\eval_{\mc P}(V)$, the requirement of inclusion
of $B_{1}(x_{0})$ in the level set of $\pi$ of measure $1/8$ is
replaced by $B_{1}(x_{0})\subset\msf L_{\pi^{X},g}$.} As we will
see in Lemma~\ref{lem:relaxed-regularity}, this is a weaker condition. 
\begin{problem}
[Zeroth-order logconcave sampling from a warm start] \label{prob:exp-sampling}
Assume access to a well-defined function oracle $\eval(V)$ for convex
$V:\Rn\to\overline{\R}$.  Given target accuracy $\veps>0$ and $M$-warm
initial distribution for $\pi^{X}$, how many evaluation queries do
we need to obtain a sample $X_{*}$ such that
\[
D(\law X_{*},\pi^{X})\leq\veps\quad\text{for }D=\tv,\KL,\eu R_{q},\text{ etc.}
\]
\end{problem}

Note that this \emph{does not} require the knowledge of $x_{0},R$,
and $\min V$. Also, if the inner radius is $r$, then we can simply
rescale the whole system by $x\mapsto r^{-1}x$. 

We prove \textbf{Result 1} in this section:
\begin{thm}
[Restatement of Theorem~\ref{thm:lc-warmstart-intro}] In the setting
of Problem~\ref{prob:exp-sampling}, there exists an algorithm that
for given $\veps,\eta\in(0,1)$, $q\in[2,\infty)$, and $M$-warm
initial distribution for $\pi^{X}$, with probability at least $1-\eta$
returns a sample $X_{*}$ satisfying $\eu R_{q}(\law X_{*}\mmid\pi^{X})\leq\veps$,
using $\Otilde(qMn^{2}(\norm{\cov\pi^{X}}\vee1)\log^{4}\nicefrac{1}{\eta\veps})$
evaluation queries in expectation.
\end{thm}

See Theorem~\ref{thm:exp-sampling} for details. When $\norm{\cov\pi^{X}}=\Omega(1)$
(the primary regime for algorithmic applications later), this complexity
bound generalizes the well-known complexity bound of $\Otilde(qMn^{2}\norm{\cov\pi^{X}}\polylog\nicefrac{1}{\eta\veps})$
for uniform sampling from a convex body \cite{kook2024inout}. Also,
if the inner radius is $r$, then the query complexity includes an
additional multiplicative factor of $r^{-2}$.

\subsection{Reduction to an exponential distribution\label{ssec:exp-reduction}}

The logconcave sampling problem can be reduced to the following exponential
sampling problem:
\begin{equation}
\D\pi(x,t)\propto\exp(-nt)\,\D x\D t\quad\text{subject to }\K:=\bbrace{(x,t)\in\Rn\times\R:V(x)\leq nt}\,.\tag{exp-red}\label{eq:exp-reduction}
\end{equation}
We note that convexity of the constraint $\K$ follows from that of
$V$, and that the $X$-marginal of $\pi$ exactly matches the desired
target $\pi^{X}\propto\exp(-V)$. 
\begin{prop}
\label{prop:x-marginal} The $X$-marginal $\pi^{X}$ of $\pi$ is
proportional to $\exp(-V)$.
\end{prop}

\begin{proof}
We just integrate out $t$ as follows:
\[
\pi^{X}(x)=\int\pi(x,t)\,\D t=\frac{\int\exp(-nt)|_{[V(x)\leq nt]}\,\D t}{\int\exp(-nt)|_{[V(x)\leq nt]}\,\D t\D x}=\frac{\exp\bpar{-V(x)}}{\int\exp\bpar{-V(x)}\,\D x}\,,
\]
where the last equality follows from 
\[
\int\exp(-nt)|_{[V(x)\leq nt]}\,\D t=\int_{V(x)/n}^{\infty}\exp(-nt)\,\D t=\frac{1}{n}\int_{V(x)}^{\infty}\exp(-u)\,\D u=\frac{1}{n}\exp\bpar{-V(x)}\,,
\]
which completes the proof.
\end{proof}
This implies that in order to sample from $\exp(-V)$, one can just
focus on sampling $(x,t)\sim\pi$ and take $x$-part only. Precisely,
once we obtain a sample $(X,T)\in\R^{n+1}$ such that $D\bigl(\law(X,T)\mmid\pi\bigr)\leq\veps$
for the $f$-divergence or R\'enyi divergence $D$, the data-processing
inequality (DPI) (Lemma~\ref{lem:DPI}) ensures that $D\bpar{\law X\mmid\pi^{X}}\leq\veps$.
Therefore, the general logconcave sampling problem can be reduced
to sampling from a particular form of the exponential distribution
above. While it allows us to move to the simpler problem, for its
mixing analysis, we will need to know the mean and variance along
the new dimension as well as the operator norm of the covariance matrix
of $\pi$.

\paragraph{Mean.}

For the mean $(\mu_{x},\mu_{t}):=\E_{\pi}[(X,T)]\in\Rn\times\R$,
we establish a useful bound on $\mu_{t}=\E_{\pi}T$.
\begin{lem}
\label{lem:exp-mean} In the setting of Problem~\ref{prob:exp-sampling},
it holds that $0\leq\mu_{t}-V_{0}/n\leq\sqrt{10}$.
\end{lem}

\begin{proof}
We directly compute $\mu_{t}=\E_{\pi^{T}}T=\E_{\pi}T$ as follows:
\begin{align*}
\mu_{t} & =\E_{\pi}T=\frac{\int te^{-nt}|_{[V(x)\leq nt]}\,\D t\D x}{\int e^{-nt}|_{[V(x)\leq nt]}\,\D t\D x}=\frac{\int\int_{V(x)/n}^{\infty}te^{-nt}\,\D t\D x}{n^{-1}\int\exp(-V)\,\D x}=\frac{\int\int_{V(x)}^{\infty}ue^{-u}\,\D u\D x}{n\int\exp(-V)\,\D x}\\
 & =\frac{\int\bpar{[-ue^{-u}]_{V(x)}^{\infty}+\int_{V(x)}^{\infty}e^{-u}\,\D u}\,\D x}{n\int\exp(-V)\,\D x}=\frac{\int(V+1)\,\exp(-V)\,\D x}{n\int\exp(-V)\,\D x}\\
 & =n^{-1}\E_{\pi^{X}}(V+1)=n^{-1}\E_{\pi}(V+1)\,.
\end{align*}

We now prove $\E_{\pi}(V-V_{0})\le\sqrt{10}n$. For $\pi^{X}=\exp(-W)/\int\exp(-W)$
with $W:=V-V_{0}\geq0$, Lemma~\ref{lem:LC-tail} ensures that for
$\beta\geq2$,
\[
\P_{\pi}\bpar{W(X)\geq\beta\,(n-1)}\leq(e^{1-\beta}\beta)^{n-1}\,.
\]
Therefore, for $n\geq2$ and the Gamma function $\Gamma$,
\begin{align*}
\E_{\pi}W & =\int_{0}^{\infty}\P_{\pi}\bpar{W(X)\geq u}\,\D u\\
 & =\int_{0}^{2(n-1)}\P_{\pi}\bpar{W(X)\geq u}\,\D u+\int_{2(n-1)}^{\infty}\P_{\pi}\bpar{W(X)\geq u}\,\D u\\
 & \leq2\,(n-1)+(n-1)\int_{2}^{\infty}\P_{\pi}\bpar{W(X)\geq\beta\,(n-1)}\,\D\beta\\
 & \leq2\,(n-1)+(n-1)\int_{0}^{\infty}(e^{1-\beta}\beta)^{n-1}\,\D\beta\,\\
 & \leq n\,\Bpar{2\,\frac{n-1}{n}+\bpar{\frac{e}{n-1}}^{n-1}\frac{\Gamma(n)}{n}}\,.
\end{align*}
When $n=1$, by using $\P_{\pi}(W(X)\geq t)\leq e^{-t}$ by \cite[Lemma 5.6(a)]{lovasz2007geometry},
we also have $\E_{\pi}W\leq1$. One can manually check that the final
bound is at most $\sqrt{10}n$ for $n\leq7$. When $n\geq8$, since
the digamma function $\psi(x):=\log\Gamma(x)$ satisfies $\psi(x)\leq\log x$,
the second term in the parenthesis is decreasing in $n$, and thus
$\E_{\pi}W\leq\sqrt{10}n$ as well. Substituting this bound back to
the expression of $\mu_{t}$, it follows that 
\[
\frac{V_{0}}{n}\leq\mu_{t}\leq\frac{V_{0}}{n}+\sqrt{10}\,,
\]
which completes the proof.
\end{proof}

\paragraph{Covariance.}

Let us denote the covariance of $\pi,\pi^{X},\pi^{T}$ by $\Sigma:=\E_{\pi}\bbrack{\bpar{(X,T)-(\mu_{x},\mu_{t})}^{\otimes2}}$,
$\Sigma_{x}:=\E_{\pi^{X}}[(X-\mu_{x})^{\otimes2}]$, $\Sigma_{t}:=\E_{\pi^{T}}[(T-\mu_{t})^{2}]$,
respectively. We study two important quantities here --- the largest
eigenvalue (i.e., operator norm of $\Sigma$) and trace (i.e., second
moment of $\pi$).
\begin{lem}
\label{lem:exp-cov} The variance in the $t$-direction is at most
$160$. Moreover, $\tr\Sigma\leq\tr\Sigma_{x}+160$, and $\norm{\Sigma}\leq2\,(\norm{\Sigma_{x}}+160)$.
\end{lem}

\begin{proof}
Let us bound $\E_{\pi}[\abs{T-\mu_{t}}^{2}]=\E_{\pi^{T}}[\abs{T-\mu_{t}}^{2}]$.
By the reverse H\"older (Lemma~\ref{lem:reverse-Holder}),
\[
\E_{\pi}[\abs{T-\mu_{t}}^{2}]\leq\E_{\pi}[\abs{T-V_{0}/n}^{2}]\leq16\,(\E_{\pi}\abs{T-V_{0}/n})^{2}\leq160\,.
\]
Hence,
\[
\tr\Sigma=\E_{\pi}[\norm{X-\mu_{x}}^{2}+\abs{T-\mu_{t}}^{2}]=\E_{\pi^{X}}[\norm{X-\mu_{x}}^{2}]+\E_{\pi^{T}}[\abs{T-\mu_{t}}^{2}]\leq R^{2}+160\,.
\]

As for the largest eigenvalue, it follows that for a unit vector $v=(v_{x},v_{t})\in\Rn\times\R$,
\[
\E_{\pi}\bbrack{\bpar{v_{x}\cdot(X-\mu_{x})+v_{t}\,(T-\mu_{t})}^{2}}\leq2\,\E_{\pi}\bbrack{\bpar{v_{x}\cdot(X-\mu_{x})}^{2}}+2\,\E_{\pi}\bbrack{\bpar{v_{t}\,(T-\mu_{t})}^{2}}\leq2\,(\norm{\Sigma_{x}}+160)\,.
\]
Hence, $\norm{\Sigma}\leq2\,(\norm{\Sigma_{x}}+160)$.
\end{proof}
We now show that the `measure-based' regularity implies the `value-based'
one.
\begin{lem}
\label{lem:relaxed-regularity} When a level set of $\pi^{X}\propto\exp(-V)$
of measure $\msf m<1$ contains $B_{r}(x_{0})$ for $r>0$ and $x_{0}\in\Rn$,
the level set $\msf L_{V,l}=\{V(x)-V_{0}\leq\ell n\}$ with $\ell=5+13\log\frac{e}{1-\msf m}$
contains $B_{r}(x_{0})$.
\end{lem}

\begin{proof}
Consider the exponential reduction $\pi(x,t)\propto\exp(-nt)|_{\K}$
for $\K=\{(x,t):V(x)\leq nt\}$. Since the variance in the $t$-direction
is at most $160$, Lemma~\ref{lem:LC-exponential-decay} leads to
\[
\P_{\pi}(\abs{T-\mu_{t}}\geq13c)=\P_{\pi^{T}}(\abs{T-\mu_{t}}\geq13c)\leq\exp(-c+1)\quad\text{for any }c>0\,,
\]
which implies $\P_{\pi}(T\geq5+V_{0}/n+13c)\leq\P_{\pi}(T\geq\mu_{t}+13c)\leq\exp(-c+1)$.
Hence, for $c=-\log\frac{e}{1-\msf m}$, we have $\P_{\pi}(T-V_{0}/n\leq\ell)\geq\msf m$.
Since $\K\cap[t\leq V_{0}/n+\ell]=\K\cap(\msf L_{V,l}\times\R)$,
it follows from $\P_{\pi}(T\leq V_{0}/n+\ell)=\pi^{X}(\msf L_{V,l})$
that $\msf L_{V,l}$ contains $B_{r}(x_{0})$.
\end{proof}

\subsection{Proximal sampler for the reduced exponential distribution\label{ssec:prox-exp}}

We now propose a sampling algorithm for this exponential distribution.
For $z:=(x,t)\in\Rn\times\R$ and parameter $h>0$, we use the $\ps$
($\PS$) to sample from
\[
\pi^{Z,Y}(z,y)\propto\exp\bpar{-\alpha^{\T}z-\frac{1}{2h}\,\norm{z-y}^{2}}|_{\K}\quad\text{for }\alpha=ne_{n+1}\,,
\]
where its $Z$-marginal $\pi^{Z}$ corresponds to our desired target
$\pi^{Z}\propto\exp(-\alpha^{\T}z)|_{\K}$.

The $\ps$ for this target, called $\psexp$ (Algorithm~\ref{alg:prox-exp}),
repeats the following two steps:
\begin{itemize}
\item {[}Forward{]} $y\sim\pi^{Y|Z=z}=\mc N(z,hI_{n+1})$.
\item {[}Backward{]} $z\sim\pi^{Z|Y=y}\propto\exp(-\alpha^{\T}z-\frac{1}{2h}\,\norm{y-z}^{2})|_{\K}=\mc N(y-h\alpha,hI_{n+1})|_{\K}$
for $\alpha=ne_{n+1}$,
\end{itemize}
and the backward step is implemented by rejection sampling with proposal
$\mc N(y-h\alpha,hI_{n+1})$.

As in previous work, the analysis can be divided into two parts: (1)
mixing towards the target (\S\ref{subsec:exp-mixing}), and (2) the
query complexity of implementing the backward step (\S\ref{subsec:exp-backward-complexity}).
The first part can be analyzed similarly to prior work on uniform
sampling \cite{kook2024inout} or constrained Gaussian sampling \cite{kook2024renyi},
while the second part introduces additional technical challenges.

\subsubsection{Convergence rate \label{subsec:exp-mixing}}

One iteration of the $\ps$ (one forward step $+$ one backward step)
for $\pi^{X}\propto\exp(-V)$ can be viewed as a composition of forward
and backward heat flow, where it can be written as follows: for $\pi(x,y)\propto\exp\bpar{-V(x)-\frac{1}{2h}\norm{x-y}^{2}}$,
repeat (i) $y_{i+1}\sim\pi^{Y|X=x_{i}}=\mc N(x_{i},hI_{n})$ and (ii)
$x_{i+1}\sim\pi^{X|Y=y_{i+1}}$. To make this perspective concrete,
we borrow expositions from \cite[p.4-5]{kook2024inout} and \cite[\S8.3]{chewi2023log}.
For the laws of $x_{i}$ and $y_{i}$ denoted by $\mu_{i}^{X}$ and
$\mu_{i}^{Y}$, the forward step can be described as $\mu_{i+1}^{Y}=\int\pi^{Y|X=x}\,\mu_{i}^{X}(\D x)=\mu_{i}^{X}*\mc N(0,hI_{n})$.
This step can be seen as simulating a Brownian motion with initial
measure $\mu_{i}^{X}$ for time $h$:
\[
\D Z_{t}=\D B_{t}\quad\text{with }Z_{0}\sim\mu_{i}^{X}\Longrightarrow Z_{h}\sim\mu_{i+1}^{Y}\,.
\]
Let us denote $\mu P_{t}:=\mu*\mc N(0,tI_{n})$ for $\mu\in\mc P(\Rn)$
(i.e., $(P_{t})_{t\geq0}$ is the heat semigroup). Next, the backward
step corresponds to $\mu_{i+1}^{X}=\int\pi^{X|Y=y}\,\mu_{i+1}^{Y}(\D y)$,
and it turns out that this process can also be represented by an SDE
(called the backward heat flow): for a new Brownian motion $B_{t}$,
\[
\D Z_{t}^{\leftarrow}=\nabla\log(\pi^{X}P_{h-t})(Z_{t}^{\leftarrow})\,\D t+\D B_{t}\quad\text{for }t\in[0,h]\,.
\]
Its construction ensures that if $Z_{0}^{\leftarrow}\sim\delta_{y}$
(a point mass at $y$), then $Z_{h}^{\leftarrow}\sim\pi^{X|Y=y}$.
Hence, if we initialize this SDE with $Z_{0}^{\leftarrow}\sim\mu_{i+1}^{Y}$,
then the law of $Z_{h}^{\leftarrow}$ corresponds to $\int\pi^{X|Y=y}\,\mu_{i+1}^{Y}(\D y)=\mu_{i+1}^{X}$.

This perspective establishes a natural connection between mixing guarantees
and constants of functional inequalities (e.g., \eqref{eq:pi} and
\eqref{eq:lsi}) for a target distribution. \cite{chen2022improved}
has already demonstrated this connection for the $\ps$ in the case
of smooth, unconstrained distributions.
\begin{prop}
[{\cite[Theorem 3 and 4]{chen2022improved}}] \label{prop:ps-smooth}
Assume that $\pi\in\mc P(\Rn)$ has a smooth density, and let $P$
denote the Markov kernel of the $\ps$. If $\pi$ satisfies \eqref{eq:lsi}
with constant $\clsi$, then for any $q\geq1$ and distribution $\mu$
with smooth density and $\mu\ll\pi$,
\begin{equation}
\eu R_{q}(\mu P\mmid\pi)\leq\frac{\eu R_{q}(\mu\mmid\pi)}{(1+h/C_{\msf{LSI}})^{2/q}}\,.\label{eq:LSI-contraction}
\end{equation}
If $\pi$ satisfies \eqref{eq:pi} with constant $\cpi$, then 
\begin{equation}
\chi^{2}(\mu P\mmid\pi)\leq\frac{\chi^{2}(\mu\mmid\pi)}{(1+h/C_{\msf{PI}})^{2}}\,.\label{eq:PI-contraction}
\end{equation}
Moreover, for all $q\geq2$,
\begin{equation}
\eu R_{q}(\mu P\mmid\pi)\leq\begin{cases}
\eu R_{q}(\mu\mmid\pi)-\frac{2\log(1+h/C_{\msf{PI}})}{q} & \text{if }\eu R_{q}(\mu\mmid\pi)\geq1\,,\\
\frac{\eu R_{q}(\mu\mmid\pi)}{(1+h/C_{\msf{PI}})^{2/q}} & \text{if }\eu R_{q}(\mu\mmid\pi)\leq1\,.
\end{cases}\label{eq:cpi-renyi}
\end{equation}
\end{prop}

Hence, its mixing times are roughly $qh^{-1}\clsi\log\frac{\log\eu R_{q}}{\veps}$
and $qh^{-1}\cpi\log\frac{\eu R_{q}}{\veps}$ under \eqref{eq:lsi}
and \eqref{eq:pi}, respectively. We can extend this mixing result
to constrained distributions under mild assumptions. 
\begin{lem}
[{\cite[Lemma 22]{kook2024inout}}] \label{lem:extension} Let
$\pi$ be a probability measure, absolutely continuous with respect
to the Lebesgue measure over $\K$. For $\{P_{t}\}_{t\geq0}$ the
heat semigroup, the forward and backward heat flow equations given
by
\begin{align*}
\partial_{t}\mu_{t} & =\frac{1}{2}\,\Delta\mu_{t}\,,\\
\partial_{t}\mu_{t}^{\leftarrow} & =-\Div\bigl(\mu_{t}^{\leftarrow}\nabla\log(\pi P_{h-t})\bigr)+\frac{1}{2}\,\Delta\mu_{t}^{\leftarrow}\quad\text{with }\mu_{0}^{\leftarrow}=\mu_{h}\,,
\end{align*}
admit solutions on $(0,h]$, and the weak limit $\lim_{t\to h}\mu_{t}^{\leftarrow}=\mu_{h}^{\leftarrow}$
exists for any initial measure $\mu_{0}$ with support\footnote{The original result in \cite{kook2024inout} is proven only for a
bounded support, but it can be extended to an unbounded support since
the almost sure pointwise convergence of $\pi*\mc N(0,\veps I_{n})$
to $\pi$ holds as $\veps\to0$ \cite[Theorem 8.15]{folland1999real}.} contained in $\supp\pi$. Moreover, for any $q$-R\'enyi divergence
with $q\in(1,\infty)$, 
\[
\eu R_{q}(\mu_{h}^{\leftarrow}\mmid\pi)\leq\lim_{t\downarrow0}\eu R_{q}(\mu_{h-t}^{\leftarrow}\mmid\pi P_{t})\,.
\]
\end{lem}

 Using this lemma, we now establish a mixing of the $\ps$ (for potentially
constrained distributions) in $\chi^{2}$ or $\eu R_{q}$ for any
$q\geq2$ under \eqref{eq:pi}, and in $\eu R_{q}$ for any $q\geq1$
under \eqref{eq:lsi}.
\begin{lem}
\label{lem:exp-mixing} Let $P$ be the Markov kernel of the $\ps$
and $\mu$ be any initial distribution supported on $\K\subset\Rn$.
Then, the contractions \eqref{eq:LSI-contraction} under \eqref{eq:lsi},
and \eqref{eq:PI-contraction}-\eqref{eq:cpi-renyi} under \eqref{eq:pi}
remain valid for $\clsi=\clsi(\pi^{X})$ and $\cpi=\cpi(\pi^{X})$.
\end{lem}

\begin{proof}
For any $\epsilon>0$, we have that $\mu_{\epsilon}=\mu*\mc N(0,\epsilon I_{n})$
and $\pi_{\epsilon}=\pi^{X}*\mc N(0,\epsilon I_{n})$ are $C^{\infty}$-smooth.
Let $D$ denote either $\chi^{2}$ or $\eu R_{q}$, and define $C_{\epsilon}:=1+\tfrac{h-\epsilon}{C_{*}(\pi^{X})+\epsilon}\leq1+\tfrac{h-\epsilon}{C_{*}(\pi_{\epsilon})}$
for $*\in\{\msf{PI},\msf{LSI}\}$, where the last inequality holds
due to \cite[Corollary 13]{chafai2004entropies}:
\[
\cpi\bigl(\pi^{X}*\mc N(0,\epsilon I_{n})\bigr)\leq\cpi(\pi^{X})+\epsilon\quad\&\quad\clsi\bigl(\pi^{X}*\mc N(0,\epsilon I_{n})\bigr)\leq\clsi(\pi^{X})+\epsilon\,.
\]
Hence, by Proposition~\ref{prop:ps-smooth} with step size $h-\epsilon$,
we obtain that if $\eu R_{q}(\mu\mmid\pi^{X})\leq1$, for $\mu_{0}=\mu$,
\[
D(\mu_{h-\epsilon}^{\leftarrow}\mmid\pi_{\epsilon})\leq C_{\epsilon}^{-1}D(\mu_{\epsilon}\mmid\pi_{\epsilon})\leq C_{\epsilon}^{-1}D(\mu\mmid\pi^{X})\,,
\]
where the last inequality follows from the DPI (Lemma~\ref{lem:DPI})
for the $f$-divergence ($\chi^{2}$) and R\'enyi divergence $\eu R_{q}$.
By the lower semicontinuity of $D$ (Lemma~\ref{lem:extension}),
sending $\epsilon\to0$ leads to 
\[
D(\mu P\mmid\pi^{X})=D(\mu_{h}^{\leftarrow}\mmid\pi^{X})\leq C_{0}^{-1}D(\mu\mmid\pi^{X})=\frac{\chi^{2}(\mu\mmid\pi^{X})}{\bpar{1+h/\cpi(\pi^{X})}^{2}}\ \text{ or }\ \frac{\eu R_{q}(\mu\mmid\pi^{X})}{\bpar{1+h/\cpi(\pi^{X})}^{2/q}}\,.
\]
When $\eu R_{q}(\mu\mmid\pi^{X})\geq1$, a similar limiting argument
shows that 
\[
\eu R_{q}(\mu P\mmid\pi^{X})\leq\eu R_{q}(\mu_{\epsilon}\mmid\pi_{\epsilon})-\frac{2\log C_{\epsilon}}{q}\leq\eu R_{q}(\mu\mmid\pi^{X})-\frac{2\log C_{0}}{q}\,,
\]
where in the first inequality we used the lower-semicontinuity of
$\eu R_{q}$, and in the second inequality used the DPI and then sent
$\epsilon\to0$. Contraction in $\eu R_{q}$ for any $q\geq1$ under
\eqref{eq:lsi} follows from the same argument above.
\end{proof}
We can simply set $X\gets Z=(X,T)$ and $Y\gets Y$ above to obtain
a mixing guarantee of the $\psexp$ in terms of $\cpi(\pi^{Z})$,
which satisfies $\cpi(\pi^{Z})\lesssim(\cpi(\pi^{X})+1)\log n$ due
to \cite{klartag2023logarithmic} and Lemma~\ref{lem:exp-cov}.
\begin{rem}
\label{rmk:LSI-mixing-exp} In general, one cannot expect a mixing
result of the $\psexp$ based on \eqref{eq:lsi} even if the original
distribution $\pi^{X}\propto\exp(-V)$ satisfies the LSI, as it does
not hold in general for an exponential distribution. Hence, before
sampling, an additional step such as truncation to a bounded convex
set may be required. We will go through this procedure in \S\ref{sec:warm-start}
in order to invoke a mixing result via the LSI (see, for instance,
Lemma~\ref{lem:ps-exp-Rinf}).
\end{rem}

\subsubsection{Complexity of implementing the backward step \label{subsec:exp-backward-complexity}}

As opposed to the forward step, it is unclear how to sample from the
backward distribution $\pi^{Z|Y=y}\propto\exp(-\alpha^{\T}z-\frac{1}{2h}\,\norm{y-z}^{2})|_{\K}$
for $\alpha=ne_{n+1}$. Due to 

\[
\alpha^{\T}z+\frac{1}{2h}\,\norm{z-y}^{2}=\frac{1}{2h}\,\norm{z-(y-h\alpha)}^{2}+\alpha^{\T}y-\frac{h}{2}\,\norm{\alpha}^{2}\,,
\]
the distribution can be written as 
\[
\pi^{Z|Y=y}\propto\exp\bpar{-\frac{1}{2h}\,\norm{z-(y-h\alpha)}^{2}}\big|_{\K}\,.
\]
Thus, we use rejection sampling with proposal $\mc N(y-h\alpha,hI_{n+1})$,
accepting a proposed sample only if it lies within $\K$ (which can
be checked by using the evaluation oracle). The expected number of
trials until the first acceptance is 
\[
\frac{1}{\ell(y)}:=\frac{(2\pi h)^{(n+1)/2}}{\int_{\mc K}\exp(-\frac{1}{2h}\,\norm{z-(y-h\alpha)}^{2})\,\D z}\,.
\]

We show in this section that the expected number of wasted proposals
is moderate under suitable choices of variance $h$ and threshold
$N$ in the $\psexp$. To this end, we first establish analytical
and geometrical properties of $\pi^{Y}=\pi^{Z}*\mc N(0,hI_{n+1})$.

\paragraph{Density after taking the forward heat flow.}

We deduce a density of $\pi^{Y}=\pi^{Z}*\mc N(0,hI_{n+1})$, using
the definition of convolution.
\begin{lem}
[Density of $\pi^Y$] \label{lem:exp-piY-density} For $\pi^{Z}\propto\exp(-\alpha^{\T}z)|_{\K}$,
the density of $\pi^{Y}=\pi^{Z}*\mc N(0,hI_{n+1})$ is
\[
\pi^{Y}(y)=\frac{\int_{\mc K}\exp(-\frac{1}{2h}\,\norm{z-(y-h\alpha)}^{2})\,\D z}{\int_{\mc K}\exp(-\alpha^{\T}z)\,\D z\cdot(2\pi h)^{(n+1)/2}}\,\exp\bpar{-\alpha^{\T}y+\half\,h\norm{\alpha}^{2}}=\frac{\ell(y)\exp(-\alpha^{\T}y+\half\,h\norm{\alpha}^{2})}{\int_{\mc K}\exp(-\alpha^{\T}z)\,\D z}\,.
\]
\end{lem}

\begin{proof}
By the definition of the convolution,
\begin{align*}
\pi^{Y}(y) & =\frac{\int_{\mc K}\exp(-\alpha^{\T}z-\frac{1}{2h}\,\norm{z-y}^{2})\,\D z}{\int_{\mc K}\exp(-\alpha^{\T}z)\,\D z\cdot(2\pi h)^{(n+1)/2}}=\frac{\int_{\mc K}\exp(-\frac{1}{2h}\,\norm{z-(y-h\alpha)}^{2})\,\D z}{\int_{\mc K}\exp(-\alpha^{\T}z)\,\D z\cdot(2\pi h)^{(n+1)/2}}\,\exp\bpar{-\alpha^{\T}y+\half\,h\norm{\alpha}^{2}}\\
 & =\frac{\ell(y)}{\int_{\mc K}\exp(-\alpha^{\T}z)\,\D z}\exp\bpar{-\alpha^{\T}y+\half\,h\norm{\alpha}^{2}}\,,
\end{align*}
which completes the proof.
\end{proof}
We now identify an \emph{effective domain} of $\pi^{Y}$, which takes
up most of measure of $\pi^{Y}$. Below, $\K_{\delta}:=\{x:d(x,\K)\leq\delta\}$
denotes the $\delta$-\emph{blow up} of $\K$.
\begin{lem}
\label{lem:exp-eff-domain} In the setting of Problem~\ref{prob:exp-sampling},
for $\widetilde{\K}:=\K_{\delta}+h\alpha$, if $\delta\geq13hn$,
then
\[
\pi^{Y}(\widetilde{\K}^{c})\leq\exp\bpar{-\frac{\delta^{2}}{2h}+13\delta n}\,.
\]
\end{lem}

\begin{proof}
Using the density formula of $\pi^{Y}$ in Lemma~\ref{lem:exp-piY-density},
\begin{align*}
\int_{\mc K}e^{-\alpha^{\T}z}\,\D z\cdot\pi^{Y}(\widetilde{\K}^{c}) & =\frac{1}{(2\pi h)^{(n+1)/2}}\int_{\widetilde{\K}^{c}}\int_{\mc K}e^{-\frac{1}{2h}\,\norm{z-(y-h\alpha)}^{2}}e^{-\alpha^{\T}y+\frac{1}{2}\,h\norm{\alpha}^{2}}\,\D z\D y\\
 & \underset{(i)}{=}\frac{1}{(2\pi h)^{(n+1)/2}}\int_{\mc K_{\delta}^{c}}\int_{\mc K}e^{-\frac{1}{2h}\,\norm{z-v}^{2}}e^{-\alpha^{\T}v-\half\,h\norm{\alpha}^{2}}\,\D z\D v\\
 & \underset{(ii)}{\leq}\frac{1}{(2\pi h)^{(n+1)/2}}\int_{\mc K_{\delta}^{c}}\int_{\mc H(v)}e^{-\frac{1}{2h}\,\norm{z-v}^{2}}e^{-\alpha^{\T}v-\half\,h\norm{\alpha}^{2}}\,\D z\D v\\
 & =\int_{\mc K_{\delta}^{c}}e^{-\alpha^{\T}v-\half\,h\norm{\alpha}^{2}}\Bpar{\int_{d(v,\mc K)}^{\infty}\frac{e^{-\frac{t^{2}}{2h}}}{\sqrt{2\pi h}}\,\D t}\,\D v=:(\#)\,,
\end{align*}
where $(i)$ follows from the change of variables by $v=y-h\alpha$,
and in $(ii)$ for $v\in\K_{\delta}^{c}$, $\mc H(v)$ denotes the
supporting half-space at $\msf{proj}_{\mc K}(v)$ containing $\mc K$,
where $\msf{proj}_{\mc K}(v)$ denotes the projection of $v$ to $\K$.
Denoting the CDF of the standard Gaussian by $\Phi$, we have
\[
\frac{1}{\sqrt{2\pi h}}\int_{s}^{\infty}e^{-\frac{t^{2}}{2h}}\,\D t=1-\Phi(h^{-1/2}s)=:\Phi^{c}(h^{-1/2}s)\,.
\]
By the co-area formula and integration by parts, for the $n$-dimensional
Hausdorff measure $\mc H^{n}$,
\begin{align}
(\#) & \underset{\text{coarea}}{=}e^{-\half h\norm{\alpha}^{2}}\int_{\delta}^{\infty}\Phi^{c}(h^{-1/2}s)\int_{\de\mc K_{s}}e^{-\alpha^{\T}v}\,\mc H^{n}(\D v)\D s\nonumber \\
 & \underset{\text{IBP}}{=}e^{-\half h\norm{\alpha}^{2}}\Bpar{\Bbrack{\Phi^{c}(h^{-1/2}s)\int_{0}^{s}\int_{\de\mc K_{t}}e^{-\alpha^{\T}v}\,\mc H^{n}(\D v)\D t}_{s=\delta}^{\infty}+\int_{\delta}^{\infty}\frac{e^{-\frac{s^{2}}{2h}}}{\sqrt{2\pi h}}\int_{0}^{s}\int_{\de\mc K_{t}}e^{-\alpha^{\T}v}\,\mc H^{n}(\D v)\D t\D s}\nonumber \\
 & =e^{-\half h\norm{\alpha}^{2}}\Bpar{\underbrace{\Bbrack{\Phi^{c}(h^{-1/2}s)\int_{\mc K_{s}\backslash\mc K}e^{-\alpha^{\T}v}\,\D v}_{s=\delta}^{\infty}}_{=:\msf I}+\int_{\delta}^{\infty}\frac{e^{-\frac{s^{2}}{2h}}}{\sqrt{2\pi h}}\int_{\mc K_{s}\backslash\mc K}e^{-\alpha^{\T}v}\,\D v\D s}\,.\label{eq:last1}
\end{align}
As shown later in Lemma~\ref{lem:exp-expansion}, this quantity
is bounded by $e^{13sn}\int_{\K}\exp(-\alpha^{\T}v)\,\D v$. From
a standard bound on $\Phi^{c}$ \cite[Theorem 20]{ruskai2000study}
given by
\begin{equation}
\Phi^{c}(h^{-1/2}s)\leq\frac{\exp(-\frac{s^{2}}{2h})}{\sqrt{2\pi}}\quad\text{for any }s>0\,,\label{eq:Gaussian-tail}
\end{equation}
it follows that $\Phi^{c}(h^{-1/2}s)\int_{\mc K_{s}\backslash\mc K}e^{-\alpha^{\T}v}\,\D v$
vanishes at $s=\infty$, and thus the term $\msf I$ can bounded by
$0$.

Putting the bound in Lemma~\ref{lem:exp-expansion} into \eqref{eq:last1},
\[
\int_{\mc K}e^{-\alpha^{\T}z}\,\D z\cdot\pi^{Y}(\widetilde{\K}^{c})\leq\int_{\delta}^{\infty}\frac{1}{\sqrt{2\pi h}}\,e^{-\frac{s^{2}}{2h}}e^{13sn}\,\D s\cdot\int_{\mc K}e^{-\alpha^{\T}z}\,\D z\,.
\]
Dividing both sides by $\int_{\mc K}e^{-\alpha^{\T}z}\,\D z$ and
invoking the standard bound on $\Phi^{c}$ again, we obtain 
\begin{align*}
\pi^{Y}(\widetilde{\K}^{c}) & \leq\int_{\delta}^{\infty}\frac{1}{\sqrt{2\pi h}}\,e^{-\frac{s^{2}}{2h}}e^{13sn}\D s=e^{52hn^{2}}\Bpar{1-\Phi\bpar{\frac{\delta-13hn}{h^{1/2}}}}\\
 & \leq e^{52hn^{2}}\exp\bpar{-\frac{(\delta-13hn)^{2}}{2h}}=\exp\bpar{-\frac{\delta^{2}}{2h}+13\delta n}\,,
\end{align*}
which completes the proof. 
\end{proof}
This analytical property suggests taking $\delta\asymp1/n$ and $h\asymp1/n^{2}$.
Thus, we set $h=\tfrac{c}{13^{2}n^{2}}$ and $\delta=\tfrac{t}{13n}$
for parameters $c,t>0$, under which $\delta\geq13hn$ turns into
$t\geq c$, and 
\[
\pi^{Y}(\widetilde{\K}^{c})\leq\exp\bpar{-\frac{t^{2}}{2c}+t}\,.
\]

We now prove an interesting result quantifying how fast $\int_{\mc K_{s}}\exp(-\alpha^{\T}v)\,\D v$
increases in $s$.
\begin{lem}
\label{lem:exp-expansion} In the setting of Problem~\ref{prob:exp-sampling},
\[
\frac{\int_{\mc K_{s}}\exp(-\alpha^{\T}v)\,\D v}{\int_{\mc K}\exp(-\alpha^{\T}v)\,\D v}\leq\exp(13sn)\,.
\]
\end{lem}

\begin{proof}
Due to the assumption on the ground set $\msf L_{V,g}=\{V(x)-V_{0}\leq10n\}$,
it holds that $\K\cap[t=V_{0}/n+10]$ contains $B_{1}^{n}(x_{0})$.
Then, the definition of $\K$ ensures that $\K$ contains $B_{1}^{n+1}(x_{0},V_{0}/n+11)$.

We now bound the quantity of interest: for $t_{0}=11+V_{0}/n$,
\begin{align*}
\frac{\int_{\K_{s}}\exp(-\alpha^{\T}v)\,\D v}{\int_{\K}\exp(-\alpha^{\T}v)\,\D v} & =\frac{\int_{\K_{s}}\exp(-nt)\,\D t\D x}{\int_{\K}\exp(-nt)\,\D t\D x}=\frac{\int_{\K_{s}-t_{0}e_{n+1}}\exp(-nt)\,\D t\D x}{\int_{\K-t_{0}e_{n+1}}\exp(-nt)\,\D t\D x}\\
 & \underset{(i)}{\leq}\frac{\int_{(1+s)(\K-t_{0}e_{n+1})}\exp(-nt)\,\D t\D x}{\int_{\K-t_{0}e_{n+1}}\exp(-nt)\,\D t\D x}\\
 & \underset{(ii)}{=}(1+s)^{n+1}\frac{\int_{\K-t_{0}e_{n+1}}\exp(-nt)\exp(-nst)\,\D t\D x}{\int_{\K-t_{0}e_{n+1}}\exp(-nt)\,\D t\D x}\\
 & \underset{(iii)}{\leq}e^{s(n+1)}e^{11sn}\frac{\int_{\K-t_{0}e_{n+1}}\exp(-nt)\,\D t\D x}{\int_{\K-t_{0}e_{n+1}}\exp(-nt)\,\D t\D x}\leq e^{13sn}\,,
\end{align*}
where $(i)$ follows from $\K_{s}-t_{0}e_{n+1}=(\K-t_{0}e_{n+1})_{s}\subset(1+s)(\K-t_{0}e_{n+1})$
due to $B_{1}(x_{0},0)\subset\K-t_{0}e_{n+1}$, and in $(ii)$ we
used the change of variables by scaling the whole system by $(1+s)$,
and $(iii)$ follows from $t\geq-11$ in $\K-t_{0}e_{n+1}$.
\end{proof}

\paragraph{Per-step guarantees.}

Using the properties established above, we provide concrete choices
of parameters, bounding a per-step query complexity for the backward
step under a \emph{warm start}. 
\begin{lem}
\label{lem:exp-perstep} In the setting of Problem~\ref{prob:exp-sampling},
let $\mu$ be an $M$-warm initial distribution for $\pi^{Z}$. Given
$k\in\mbb N$ and $\eta\in(0,1)$, set $S=\frac{16kM}{\eta}(\geq16)$,
$h=\frac{1}{n^{2}}\frac{(\log\log S)^{2}}{13^{2}\log S}$ and $N=S\log^{2}S$.
Then, per iteration, the failure probability is at most $\eta/k$,
and the expected number of queries is $\Otilde(M\log^{2}\tfrac{k}{\eta})$.
\end{lem}

\begin{proof}
We note that $\mu*\mc N(0,hI_{n+1})$ is $M$-warm with respect to
$\pi^{Y}$ due to $\D\mu/\D\pi^{X}\leq M$. Hence, the failure probability
per iteration is $\E_{\mu*\mc N(0,hI_{n+1})}[(1-\ell)^{N}]\leq M\E_{\pi^{Y}}[(1-\ell)^{N}]$.
For $\widetilde{\K}=\K_{\delta}+h\alpha$,
\begin{align*}
\E_{\pi^{Y}}[(1-\ell)^{N}] & =\int_{\widetilde{\K}^{c}}\cdot+\int_{\widetilde{\K}\cap[\ell\geq N^{-1}\log(3kM/\eta)]}\cdot+\int_{\widetilde{\K}\cap[\ell\leq N^{-1}\log(3kM/\eta)]}\cdot\\
 & \underset{(i)}{\leq}\pi^{Y}(\widetilde{\K}^{c})+\int_{[\ell\geq N^{-1}\log(3kM/\eta)]}\exp(-\ell N)\,\D\pi^{Y}\\
 & \qquad+\int_{\widetilde{\K}\cap[\ell\leq N^{-1}\log(3kM/\eta)]}\frac{\ell(y)\exp(-\alpha^{\T}y+\half\,h\norm{\alpha}^{2})}{\int_{\mc K}\exp(-\alpha^{\T}v)\,\D v}\,\D y\\
 & \underset{(ii)}{\leq}\exp\bpar{-\frac{t^{2}}{2c}+t}+\frac{\eta}{3kM}+\frac{\log(3kM/\eta)}{N}\,\frac{\int_{\widetilde{\K}}\exp(-\alpha^{\T}y+\half\,h\norm{\alpha}^{2})\,\D y}{\int_{\mc K}\exp(-\alpha^{\T}v)\,\D v}\\
 & \underset{(iii)}{=}\exp\bpar{-\frac{t^{2}}{2c}+t}+\frac{\eta}{3kM}+\frac{\log(3kM/\eta)}{N}\,\frac{\int_{\K_{\delta}}\exp(-\alpha^{\T}v-\half\,h\norm{\alpha}^{2})\,\D v}{\int_{\mc K}\exp(-\alpha^{\T}v)\,\D v}\\
 & \underset{(iv)}{\leq}\exp\bpar{-\frac{t^{2}}{2c}+t}+\frac{\eta}{3kM}+\frac{\log(3kM/\eta)}{N}\,e^{13\delta n}\\
 & \leq\frac{\eta}{kM}\,,
\end{align*}
where in $(i)$ we used the density formula of $\pi^{Y}$ (Lemma~\ref{lem:exp-piY-density}),
$(ii)$ follows from Lemma~\ref{lem:exp-eff-domain}, $(iii)$ is
due to the change of variables via $y=v+h\alpha$, in $(iv)$ we used
Lemma~\ref{lem:exp-expansion}, and the last line follows from the
choices of $c=\frac{(\log\log S)^{2}}{13^{2}\log S}$, $t=\log\log S$,
and $N=S\log^{2}S$. Therefore, $\E_{\mu*\mc N(0,hI_{n+1})}[(1-\ell)^{N}]\leq\eta/k$.
Note that under these choices of parameters, we have $t\geq c$ and
so $\delta\geq13hn$, as requested in Lemma~\ref{lem:exp-eff-domain}.

We now bound the expected number of trials until the first acceptance.
In a similar vein to above, the expected number is bounded by $M\E_{\pi^{Y}}[\ell^{-1}\wedge N]$,
where
\begin{align*}
\int_{\R^{n+1}}\frac{1}{\ell}\wedge N\,\D\pi^{Y} & \leq\int_{\widetilde{\K}}\frac{1}{\ell}\,\D\pi^{Y}+N\pi^{Y}(\widetilde{\K}^{c})\leq\frac{\int_{\K_{\delta}}\exp(-\alpha^{\T}v-\half\,h\norm{\alpha}^{2})\,\D v}{\int_{\mc K}\exp(-\alpha^{\T}v)\,\D v}+N\exp\bpar{-\frac{t^{2}}{2c}+t}\\
 & \leq e^{t}+N\exp\bpar{-\frac{t^{2}}{2c}+t}\leq\log S+\log^{2}S\\
 & \leq2\log^{2}\frac{16kM}{\eta}\,.
\end{align*}
Therefore, from an $M$-warm start, the expected number of queries
per iteration is $\Otilde(M\log^{2}\tfrac{k}{\eta})$. 
\end{proof}
Combining the results on the mixing and per-iteration complexity together,
we can conclude the query complexity of sampling from a logconcave
distribution, as claimed in Theorem~\ref{thm:lc-warmstart-intro}.
\begin{thm}
\label{thm:exp-sampling} In the setting of Problem~\ref{prob:exp-sampling},
consider $\pi(x,t)\propto\exp(-nt)|_{\K}$ given in \eqref{eq:exp-reduction}.
For any given $\eta,\veps\in(0,1)$, $q\geq2$, $k\in\mbb N$ defined
below, the $\psexp$ with $h=(13^{2}n^{2}\log\frac{16kM}{\eta})^{-1}$,
$N=\frac{16kM}{\eta}\log^{2}\frac{16kM}{\eta}$, and initial distribution
$\mu_{0}$ which is $M$-warm with respect to $\pi$ achieves $\eu R_{q}(\mu_{k}\mmid\pi)\leq\veps$
(so $\eu R_{q}(\mu_{k}^{X}\mmid\pi^{X})\leq\veps$) after $k=\Otilde(qn^{2}(\norm{\Sigma_{x}}\vee1)\log^{2}\frac{M}{\eta\veps})$
iterations, where $\mu_{k}$ is the law of $k$-th iterate, and $\Sigma_{x}$
is the covariance matrix of $\pi^{X}\propto\exp(-V)$. With probability
$1-\eta$, the algorithm iterates $k$ times successfully, using $\Otilde(qMn^{2}(\norm{\Sigma_{x}}\vee1)\log^{4}\frac{1}{\eta\veps})$
evaluation queries in expectation. Moreover, an $M$-warm start for
$\pi^{X}$ can be used to generate an $M$-warm start for $\pi$.
\end{thm}

\begin{proof}
By Lemma~\ref{lem:exp-mixing} and Lemma~\ref{lem:exp-cov}, the
$\psexp$ should iterate 
\[
k=\O\bpar{qh^{-1}\cpi(\pi)\log\frac{M}{\veps}}=\O\bpar{qh^{-1}(\norm{\Sigma_{x}}\vee1)\log\frac{M}{\veps}}\underset{\text{Remark \ref{rem:self-dependence}}}{=}\Otilde\bpar{qn^{2}(\norm{\Sigma_{x}}\vee1)\log^{2}\frac{M}{\eta\veps}}
\]
times to achieve $\veps$-distance in $\eu R_{q}$. Under the choice
of $h$ and $N$ as claimed, each iteration succeeds with probability
$1-\nicefrac{k}{\eta}$ and uses $\Otilde(M\log^{2}\nicefrac{k}{\eta})$
queries in expectation by Lemma~\ref{lem:exp-perstep}. Therefore,
throughout $k$ iterations, the algorithm succeeds with probability
$1-\eta$, using 
\[
\Otilde\bpar{qMn^{2}(\norm{\Sigma_{x}}\vee1)\log^{4}\frac{1}{\eta\veps}}
\]
 evaluation queries in expectation.

Next, an $M$-warm start $\mu^{X}$ for $\pi^{X}$ can be lifted to
an $M$-warm distribution on $\R^{n+1}$ for $\pi$ through the following
procedure: generate $x\sim\mu^{X}$ and draw $t\sim\pi^{T|X=x}\propto e^{-nt}\cdot\ind[t\geq\nicefrac{V(x)}{n}]$,
obtaining $(x,t)$. Note that $\pi^{T|X=x}$ can be sampled easily
by drawing $u\sim\Unif\,([0,1])$ and taking $F^{-1}(u)$, where $F$
is the CDF of $\pi^{T|X=x}$. Then, the density of the law of $(x,t)$
is $\mu^{X}\cdot\pi^{T|X}$, and it follows from $\pi(x,t)=\pi^{X}\cdot\pi^{T|X}$
that
\[
\frac{\mu^{X}\pi^{T|X}}{\pi}=\frac{\mu^{X}}{\pi^{X}}\leq M\,,
\]
which completes the proof.
\end{proof}
\begin{rem}
[Self-dependence of parameteres] \label{rem:self-dependence} Upon
substitution of the chosen $h$, the required number of iterations
would be 
\[
k\gtrsim qn^{2}(\norm{\Sigma_{x}}\vee1)\log\frac{M}{\veps}\log\frac{kM}{\eta}\,,
\]
where $k$ shows up in both sides. When $c\geq1$, since $x\geq c\log x$
holds for $x\geq2c\log c$, the above bound on $k$ can be turned
into 
\[
k\gtrsim qn^{2}(\norm{\Sigma_{x}}\vee1)\log\frac{M}{\veps}\log\bpar{qn^{2}(\norm{\Sigma_{x}}\vee1)\frac{M}{\eta}\log\frac{M}{\veps}}=\Otilde\bpar{qn^{2}(\norm{\Sigma_{x}}\vee1)\log\frac{M}{\veps}\log\frac{M}{\eta}}\,.
\]
\end{rem}

\section{Warm-start generation via tilted Gaussian cooling\label{sec:warm-start}}

In the previous section, we assumed access to an $M$-warm distribution
for $\pi(x,t)\propto\exp(-nt)|_{\K}$. Here we address the problem
of generating a warm start (i.e., $M=\O(1)$). 
\begin{problem}
[Warm-start generation] \label{prob:warm-start-generation} Given
access to a well-defined function oracle $\eval_{x_{0},R}(V)$ for
convex $V:\Rn\to\Rext$, what is the complexity of generating an
$\O(1)$-warm start $\mu^{X}$ with respect to $\pi^{X}$ (i.e., $\eu R_{\infty}(\mu^{X}\mmid\pi^{X})=\O(1)$)?
\end{problem}

Compared to Problem~\ref{prob:exp-sampling}, this \textbf{requires}
access to $x_{0}$ and $R$. We prove \textbf{Result 2} in this section:
\begin{thm}
[Restatement of Theorem~\ref{thm:tgc-intro}] In the setting of
Problem~\ref{prob:warm-start-generation}, there exists an algorithm
that for given $\eta,\veps\in(0,1)$, with probability at least $1-\eta$
returns a sample $X_{*}$ satisfying $\eu R_{\infty}(\law X_{*}\mmid\pi^{X})\leq\veps$,
using $\Otilde(n^{2}(R^{2}\vee n)\polylog\nicefrac{1}{\eta\veps})$
evaluation queries in expectation. In particular, if $\pi^{X}$ is
well-rounded (i.e., $R^{2}\asymp n$), then we need $\Otilde(n^{3}\polylog\nicefrac{1}{\eta\veps})$
queries in expectation.
\end{thm}

As in the previous section, we tackle this problem after reducing
it to the exponential distribution through \eqref{eq:exp-reduction},
and then truncate this distribution to a convex domain.

\paragraph{Convex truncation.}

Using Lemma~\ref{lem:LC-exponential-decay} with $\E_{\pi^{X}}[\norm{X-x_{0}}^{2}]\leq R^{2}$,
we have that $\text{for any }l\geq1$,
\[
\P_{\pi^{X}}(\norm{X-x_{0}}\geq Rl)\leq\exp(-l+1).
\]
In addition, since $\E_{\pi^{T}}[\abs{T-\mu_{t}}^{2}]=\E_{\pi}[\abs{T-\mu_{t}}^{2}]\leq160$
(Lemma~\ref{lem:exp-cov}), we also have 
\[
\P_{\pi^{T}}(\abs{T-\mu_{t}}\geq13l)\leq\exp(-l+1)\quad\text{for any }l\geq1\,.
\]
Putting these two together and using $\mu_{t}\leq V_{0}/n+5\leq V(x_{0})/n+5$
(Lemma~\ref{lem:exp-mean}), for $\veps>0$, $l=\log\frac{2e}{\veps}$,
and $B=\frac{V(x_{0})}{n}+5+13l$,
\begin{align*}
\P_{\pi}\bbrace{\bpar{B_{Rl}^{n}(x_{0})\times(-\infty,B]}^{c}} & \leq\P_{\pi}\bigl\{\bigl(B_{Rl}^{n}(x_{0})\times\R\bigr)^{c}\bigr\}+\P_{\pi}\bbrace{\bpar{\Rn\times[B,\infty)}^{c}}\\
 & =\P_{\pi^{X}}\bpar{\Rn\backslash B_{Rl}^{n}(x_{0})}+\P_{\pi^{T}}\bpar{\R\backslash[B,\infty)}\leq\frac{\veps}{2}+\frac{\veps}{2}=\veps\,.
\end{align*}
Thus, $\bar{\K}:=\K\cap\{B_{Rl}(x_{0})\times(-\infty,B]\}$ takes
up $(1-\veps)$ measure of $\pi$, and $B_{Rl}(x_{0})\times(V_{0}/n,B]$
could be thought of as the \emph{essential domain} of $\pi$. Moreover,
since $x_{0}\in\msf L_{\pi^{X},V}$ (i.e., $V(x_{0})-V_{0}\leq10n$),
\[
\frac{V(x_{0})}{n}+5+13l-\frac{V_{0}}{n}\leq15+13l\,,
\]
so the diameter of $\bar{\K}$ is $\O(lR)$. Going forward, we use
$\bar{\pi}$ to denote $\pi|_{\bar{\K}}\propto\exp(-nt)|_{\bar{\K}}$.

\subsection{Tilted Gaussian Cooling}

We would like to find a sequence $\{\mu_{i}\}_{0\leq i\leq m}$ of
distributions, where $\mu_{0}$ is easy to sample, $\mu_{m}$ is our
target, and $\mu_{i}$ is $\O(1)$-warm with respect to $\mu_{i+1}$.
Then, as $\mu_{i}$ is warm for $\mu_{i+1}$, we can run an efficient
sampler to go from $\mu_{i}$ to $\mu_{i+1}$.

$\gc$, introduced in \cite{cousins2018gaussian}, can generate an
$\O(1)$-warm start for the uniform distribution over a convex body.
However, extending this to general logconcave distributions has remained
elusive. A straightforward extension to the exponential distribution
$\bar{\pi}$ fails due to an exponential change in the value of $e^{-nt}$
in $\bar{\K}$; roughly speaking, while $t$ may vary by at most $\Theta(l)$,
the density can become exponentially small relative to its maximum.
This significant change in value prevents us from achieving $\O(1)$-warmness,
necessitating additional ideas to address this challenge.

We propose $\tgc$ (Algorithm~\ref{alg:tgc}), which handles the
$x$-direction through the $\gc$ algorithm and the $t$-direction
through an exponential tilt. Its annealing distributions have two
parameters $\sigma^{2}$ and $\rho$, given as
\[
\mu_{\sigma^{2},\rho}(x,t)\propto\exp\bpar{-\frac{1}{2\sigma^{2}}\,\norm{x-x_{0}}^{2}-\rho t}\big|_{\bar{\K}}\,.
\]
We can translate the whole system by $(x,t)\mapsto(x-x_{0},t-\frac{V(x_{0})}{n}-11)$
so that the shifted $\bar{\K}$ contains a unit ball centered at the
origin (i.e., $x_{0}=0$), and $t$ is within $[-21,13l-6]$ over
$\bar{\K}$.

We develop in \S\ref{ssec:prox-annealing} a sampling algorithm $\psann$
for these annealing distributions. Assuming this sampler is given,
we provide an outline of $\tgc$. For the annealing distribution $\mu_{i}:=\mu_{\sigma_{i}^{2},\rho_{i}}$,
we use $\bar{\mu}_{i}$ to denote an approximate distribution close
to $\mu_{i}$ produced by the sampler. $\tgc$ has two main annealing
phases:
\begin{itemize}
\item Initialization ($\sigma^{2}=n^{-1}$ and $\rho=l^{-1}$)
\begin{itemize}
\item Initial distribution: $\mu_{0},\bar{\mu}_{0}\propto\exp(-\frac{n}{2}\,\norm x)^{2}|_{\bar{\K}}$
(run rejection sampling with proposal $\mc N(0,n^{-1}I_{n})\otimes\text{Unif}\,([-21,13l-6])$)
\item Target distribution: $\mu_{1}\propto\exp(-\frac{n}{2}\,\norm x^{2}-t/l)|_{\bar{\K}}$. 
\end{itemize}
\item Phase I ($n^{-1}\leq\sigma^{2}\leq1$ and $l^{-1}\leq\rho\leq n$)
\begin{itemize}
\item Run the $\psann$ with initial $\bar{\mu}_{i}$, target $\mu_{i+1}$,
and accuracy $\log2$ in $\eu R_{\infty}$, where 
\[
\sigma_{i+1}^{2}=\sigma_{i}^{2}\bpar{1+\frac{1}{n}}\quad\&\quad\rho_{i+1}=\rho_{i}\bpar{1+\frac{1}{ln}}\,.
\]
\end{itemize}
\item Phase II ($1\leq\sigma^{2}\leq l^{2}R^{2}$)
\begin{itemize}
\item Run the $\psann$ with initial $\bar{\mu}_{i}$, target $\mu_{i+1}$,
and accuracy $\log2$ in $\eu R_{\infty}$, where 
\[
\sigma_{i+1}^{2}=\sigma_{i}^{2}\bpar{1+\frac{\sigma_{i}^{2}}{l^{2}R^{2}}}\,.
\]
\end{itemize}
\item Termination $(\sigma^{2}=l^{2}R^{2})$
\begin{itemize}
\item Run the $\psexp$ with initial $\mu_{l^{2}R^{2},n}(x,t)$ and target
$\bar{\pi}\propto\exp(-nt)|_{\bar{\K}}$.
\end{itemize}
\end{itemize}
\begin{algorithm}[t]
\hspace*{\algorithmicindent} \textbf{Input: }convex $\K=\{(x,t):V(x)\leq nt\}$,
accuracy $\veps>0$, target dist. $\pi(x,t)\propto\exp(-nt)|_{\K}$
s.t. $B_{1}^{n}(x_{0})\subset\msf L_{\pi^{X},g}$ and $\E_{\pi}[\norm{X-x_{0}}^{2}]\leq R^{2}$.

\hspace*{\algorithmicindent} \textbf{Output:} a sample $z$ with
law approximately close to $\pi$.

\begin{algorithmic}[1] \caption{$\protect\tgc$\label{alg:tgc}}
\STATE Let $\bar{\K}=\K\cap\bigl(B_{Rl}^{n}(x_{0})\times(-\infty,V(x_{0})/n+5+13l]\bigr)$
with $l=\log\frac{6}{\veps}$. 

\STATE Draw $z_{0}\sim\exp(-\frac{n}{2}\,\norm{x-x_{0}}^{2})|_{\bar{\K}}$
via rejection sampling with proposal $\mc N(x_{0},n^{-1}I_{n})\otimes\text{Unif}\,\bigl((V(x_{0})/n-10,V(x_{0})/n+5+13l]\bigr)$.

\STATE Denote $\mu_{\sigma^{2},\rho}:=\exp(-\frac{1}{2\sigma^{2}}\,\norm{x-x_{0}}^{2}-\rho t)|_{\bar{\K}}$
for $\sigma^{2}=n^{-1}$ and $\rho=l^{-1}$.

\STATE Get $z\sim\psann$ with initial $z_{0}$, target $\mu_{\sigma^{2},\rho}$,
accuracy $\log2$ (in $\eu R_{\infty}$), and success prob. $\frac{\eta}{\Otilde(l^{2}R^{2}\vee n)}$.

\WHILE{$\sigma^{2}\in[n^{-1},l^{2}R^{2}]$ and $\rho\in[l^{-1},n]$} 

\STATE Sample $z$ from $\psann$ with the same setup. Then update
\begin{align*}
\sigma^{2} & \gets\sigma^{2}(1+\frac{1}{n})\quad\&\quad\rho\gets\min\bigl(n,\rho\,(1+\frac{1}{ln})\bigr)\ \text{when }\sigma^{2}\in\bbrack{\frac{1}{n},1},\,\rho\in\bbrack{\frac{1}{l},n}\\
\sigma^{2} & \gets\sigma^{2}(1+\frac{\sigma^{2}}{l^{2}R^{2}})\ \text{when }\sigma^{2}\in[1,l^{2}R^{2}]\,.
\end{align*}
\ENDWHILE

\STATE Return $z\sim\psexp$ with initial $\mu_{l^{2}R^{2},n}$,
target $\pi|_{\bar{\K}}$, accuracy $\varepsilon/3$, and success
prob. $\frac{\eta}{\Otilde(l^{2}R^{2}\vee n)}$.

\end{algorithmic} 
\end{algorithm}

\subsection{Proximal sampler for annealing distributions\label{ssec:prox-annealing}}

We employ $\PS$ to sample from the annealing distributions $\{\mu_{\sigma^{2},\rho}\}$.
For $v:=(x,t)\in\Rn\times\R$ and $w:=(y,s)\in\Rn\times\R$, $\PS$
for this target and variance $h$, called $\psann$ (Algorithm~\ref{alg:prox-ann}),
can be implemented by repeating 
\begin{itemize}
\item {[}Forward{]} $w\sim\mu^{W|V=v}=\mc N(v,hI_{n+1})$.
\item {[}Backward{]} For $\tau:=\frac{\sigma^{2}}{h+\sigma^{2}}<1$, $y_{\tau}:=\tau y$,
and $h_{\tau}:=\tau h$, run
\[
v\sim\mu^{V|W=w}\propto\exp\bpar{-\frac{1}{2\sigma^{2}}\,\norm x^{2}-\rho t-\frac{1}{2h}\,\norm{w-v}^{2}}\big|_{\bar{\K}}\propto\bbrack{\mc N(y_{\tau},h_{\tau}I_{n})\times\mc N(s-\rho h,h)}\big|_{\bar{\mc K}}\,,
\]
\end{itemize}
where the backward step is implemented by rejection sampling with
the proposal $\mc N(y_{\tau},h_{\tau}I_{n})\times\mc N(s-\rho h,h)$.
Note that $\psexp$ may be thought of as a special case of $\psann$
with $\sigma^{2}=\infty$. As in \S\ref{ssec:prox-exp}, we analyze
its mixing and per-step complexity.
\begin{algorithm}[t]
\hspace*{\algorithmicindent} \textbf{Input:} initial pt. $v_{0}\sim\pi_{0}\in\mc P(\R^{n+1})$,
truncated $\bar{\K}$, $k\in\mathbb{N}$, threshold $N$, variances
$h$, parameters $\sigma^{2},\rho>0$. 

\hspace*{\algorithmicindent} \textbf{Output:} $v_{k+1}=(x_{k+1},t_{k+1})$.

\begin{algorithmic}[1] \caption{$\protect\psann$\label{alg:prox-ann}}

\STATE Let $\tau=\frac{\sigma^{2}}{h+\sigma^{2}}$. Denote $v=(x,t)$
and $w=(y,s)$ in $\Rn\times\R$.

\FOR{$i=0,\dotsc,k$}

\STATE Sample $w_{i+1}\sim\mc N(v_{i},hI_{n+1})$.

\STATE Sample $v_{i+1}\sim\msf G(w_{i+1})|_{\bar{\K}}$ for $\msf G(w):=\mc N(\tau y,\tau hI_{n})\times\mc N(s-\rho h,h)$\@.

\STATE $\quad(\uparrow)$ Repeat $v_{i+1}\sim\msf G(w_{i+1})$ until
$v_{i+1}\in\bar{\K}$. If $\#$attempts$_{i}$$\,\geq N$, declare
\textbf{Failure}.

\ENDFOR 

\end{algorithmic}
\end{algorithm}

\subsubsection{Convergence rate \label{subsec:annealing-mixing}}

For $\psann$, we use a mixing result based on \eqref{eq:lsi}, rather
than \eqref{eq:pi}. As opposed to the exponential distribution, the
annealing distributions have $\clsi(\mu)<\infty$, so we can now invoke
the \eqref{eq:lsi}-mixing result in Lemma~\ref{lem:exp-mixing}.
Therefore, for $q\geq1$, the kernel $P$ of  $\psann$, and initial
distribution $\nu$ with $\nu\ll\mu$, it holds that
\begin{equation}
\eu R_{q}(\nu P\mmid\mu)\leq\frac{\eu R_{q}(\nu\mmid\mu)}{\bigl(1+h/C_{\msf{LSI}}(\mu)\bigr)^{2/q}}\,.\label{eq:exp-contraction-LSI}
\end{equation}

We provide a bound on $\clsi(\mu)$, using the Bakry-\'Emery criterion
and bounded perturbation.
\begin{lem}
\label{lem:lsi-annealing}$\clsi(\mu)\lesssim\sigma^{2}\vee l^{2}$
for $\mu(x,t)\propto\exp(-\frac{1}{2\sigma^{2}}\,\norm x^{2}-\rho t)|_{\bar{\K}}$.
\end{lem}

\begin{proof}
Consider a probability measure defined by
\[
\nu(x,t)\propto\exp\bpar{-\frac{1}{2\sigma^{2}}\,\norm x^{2}-\frac{1}{2l^{2}}\,t^{2}-\rho t}\,.
\]
It is $\min(\sigma^{-2},l^{-2})$-strongly logconcave, so $\clsi(\nu)\leq\sigma^{2}\vee l^{2}$
by the Bakry-\'Emery criterion (Lemma~\ref{lem:bakry-emery}). It
is well known that a convex truncation of strongly logconcave measures
only helps in satisfying the criterion, so $\clsi(\nu|_{\bar{\K}})\leq\sigma^{2}\vee l^{2}$
\cite[Theorem 3.3.2]{wang2013analysis}. Since
\[
\frac{\D\mu}{\D\nu|_{\bar{\K}}}=\ind[(x,t)\in\bar{\K}]\exp\bpar{\frac{t^{2}}{2l^{2}}}\,\frac{\int_{\bar{\K}}\exp(-\frac{1}{2\sigma^{2}}\,\norm x^{2}-\frac{1}{2l^{2}}\,t^{2}-\rho t)}{\int_{\bar{\K}}\exp(-\frac{1}{2\sigma^{2}}\,\norm x^{2}-\rho t)}\,,
\]
and $t\in[-21,13l-6]$ over $\bar{\K}$, we have
\[
\exp\bigl(-\Theta(1)\bigr)\leq\inf_{\bar{\K}}\frac{\D\mu}{\D\nu|_{\bar{\K}}}\leq\sup_{\bar{\K}}\frac{\D\mu}{\D\nu|_{\bar{\K}}}\leq\exp\bpar{\Theta(1)}\,.
\]
By the Holley--Stroock perturbation principle (Lemma~\ref{lem:lsi-bdd-perturbation}),
\[
\clsi(\mu)\leq\exp\bpar{\Theta(1)}\,\clsi(\nu|_{\bar{\K}})\lesssim\sigma^{2}\vee l^{2}\,,
\]
which completes the proof.
\end{proof}
Combining the exponential contraction \eqref{eq:exp-contraction-LSI}
and the bound on the LSI constant, we conclude that if $\psann$ started
at an initial distribution $\nu$ iterates $\O(qh^{-1}\clsi(\mu)\log\frac{\eu R_{q}(\nu\mmid\mu)}{\veps})=\O(qh^{-1}(\sigma^{2}\vee l^{2})\log\frac{\eu R_{q}(\nu\mmid\mu)}{\veps})$
times, then it achieves $\veps$-distance within the target $\nu$
in $\eu R_{q}$-divergence.

We can strengthen this $\eu R_{q}$-divergence result to the stronger
$\eu R_{\infty}$-divergence by following the boosting scheme used
in \cite{kook2024renyi}, which stems from a classical result in \cite{del2003contraction}.
\begin{lem}
[{\cite[Theorem 14]{kook2024renyi}}] \label{lem:boosting} Consider
a Markov chain with initial distribution $\nu$ and kernel $P$ reversible
with respect to a stationary distribution $\mu$. Then,
\[
\bnorm{\frac{\D(\nu P^{k})}{\D\mu}-1}_{L^{\infty}}\leq2\,\bnorm{\frac{\D\nu}{\D\mu}-1}_{L^{\infty}}\esssup_{x}\norm{\delta_{x}P^{k}-\mu}_{\tv}\,,
\]
and in particular, $\eu R_{\infty}(\nu P^{k}\mmid\mu)\leq2\,\norm{\frac{\D\nu}{\D\mu}-1}_{L^{\infty}}\esssup_{x}\norm{\delta_{x}P^{k}-\mu}_{\tv}$.
\end{lem}

In words, a uniform bound on the $\tv$-distance from \emph{any start}
implies the R\'enyi-infinity bound. To use this result, we establish
the uniform ergodicity of $\psann$.
\begin{lem}
\label{lem:ps-annealing-iter} For given $\veps\in(0,1)$, let $P$
be the Markov kernel of $\psann$ with variance $h$ and target distribution
$\mu$. Then, $\sup_{v_{0}\in\mc{\bar{K}}}\norm{\delta_{v_{0}}P^{k}-\mu}_{\tv}\leq\veps$
for $k=\O(h^{-1}(\sigma^{2}\vee l^{2})\log\frac{n+h^{-1}R^{2}}{\veps})$.
If an initial distribution $\nu$ is $M$-warm with respect to $\mu$,
then $\eu R_{\infty}(\nu P^{k}\mmid\mu)\leq\veps$ after $k=\O(h^{-1}(\sigma^{2}\vee l^{2})\log\frac{M(n+h^{-1}R^{2})}{\veps})$
iterations.
\end{lem}

\begin{proof}
We first bound the warmness of $\delta_{v_{0}}P$ for $\mu$. By the
definition of $\psann$, the forward step brings $\delta_{v_{0}}$
to $\mc N(v_{0},hI_{n+1})$, and then the backward step $\pi^{V|W=(y,s)}(x,t)\propto\exp(-\frac{1}{2\sigma^{2}\,}\norm x^{2}-\rho t-\frac{1}{2h}\norm{w-v}^{2})|_{\bar{\K}}$
brings it to, for $v=(x,t)$ and $w=(y,s)$
\[
(\delta_{v_{0}}P)(x,t)=\ind_{\bar{\K}}(x,t)\,\int\frac{\exp(-\frac{1}{2\sigma^{2}}\,\norm x^{2}-\rho t-\frac{1}{2h}\,\norm{w-v}^{2})\exp(-\frac{1}{2h}\,\norm{w-v_{0}}^{2})}{(2\pi h)^{(n+1)/2}\int_{\bar{\K}}\exp(-\frac{1}{2\sigma^{2}}\,\norm z^{2}-\rho l-\frac{1}{2h}\,\norm{(y,s)-(z,l)}^{2})\,\D z\D l}\,\D y\D s\,.
\]
For $D:=\text{diam}\,\bar{\K}\lesssim Rl$, it follows from Young's
inequality ($\norm{a-b}^{2}\leq3\,\norm a^{2}+\frac{3}{2}\,\norm b^{2}$)
that
\[
\frac{1}{2h}\,\norm{(y,s)-(z,l)}^{2}\leq\frac{3}{4h}\,\norm{(y,s)}^{2}+\frac{3}{2h}\,\norm{(z,l)}^{2}\leq\frac{3}{4h}\,\norm{(y,s)}^{2}+\frac{3D^{2}}{2h}\,.
\]
Hence,
\[
\delta_{v_{0}}P\leq\frac{\exp(\frac{3D^{2}}{2h})\exp(-\frac{1}{2\sigma^{2}}\,\norm x^{2}-\rho t)}{(2\pi h)^{\frac{n+1}{2}}\int_{\bar{\K}}\exp(-\frac{1}{2\sigma^{2}}\,\norm z^{2}-\rho l)\,\D z\D l}\int\exp\bigl(-\frac{1}{2h}\,\norm{w-v}^{2}-\frac{1}{2h}\,\norm{w-v_{0}}^{2}+\frac{3}{4h}\,\norm w^{2}\bigr)\,\D y\D s\,.
\]
The exponent of the integrand can be bounded as follows:
\begin{align*}
 & -\frac{1}{2h}\,\norm{w-v}^{2}-\frac{1}{2h}\,\norm{w-v_{0}}^{2}+\frac{3}{4h}\,\norm w^{2}\\
= & -\frac{1}{4h}\,\norm{w-2(v+v_{0})}^{2}-\frac{1}{2h}\,\norm v^{2}-\frac{1}{2h}\,\norm{v_{0}}^{2}+\frac{1}{h}\,\norm{v+v_{0}}^{2}\\
\leq & -\frac{1}{4h}\,\norm{w-2(v+v_{0})}^{2}+\frac{4D^{2}}{h}\,,
\end{align*}
so it follows that
\[
\delta_{v_{0}}P\leq\frac{2^{(n+1)/2}\exp(\frac{6D^{2}}{h})\exp(-\frac{1}{2\sigma^{2}}\,\norm x^{2}-\rho t)}{\int_{\bar{\K}}\exp(-\frac{1}{2\sigma^{2}}\,\norm z^{2}-\rho l)\,\D z\D l}\,.
\]
Dividing this by $\mu(x,t)\propto\exp(-\frac{1}{2\sigma^{2}}\,\norm x^{2}-\rho t)|_{\bar{\K}}$,
we have
\[
\frac{\delta_{v_{0}}P}{\mu}\leq2^{\frac{n+1}{2}}\exp(6h^{-1}D^{2})\,,
\]
which implies $\eu R_{q}(\delta_{v_{0}}P\mmid\mu)\lesssim n+h^{-1}D^{2}$.
Thus, $\psann$ can achieve $\eu R_{q}(\delta_{v_{0}}P^{k}\mmid\mu)\leq\veps$
for $k=\O(qh^{-1}(\sigma^{2}\vee l^{2})\log\frac{n+h^{-1}D^{2}}{\veps})$,
and due to $2\,\norm{\cdot}_{\tv}^{2}\leq\lim_{q\downarrow1}\eu R_{q}$,
\[
\sup_{v_{0}\in\mc{\bar{K}}}\norm{\delta_{v_{0}}P^{k}-\mu}_{\tv}\leq\veps\,.
\]

For the second claim, replacing $\veps$ with $\veps/M$ and using
Lemma~\ref{lem:boosting} with $\bnorm{\frac{\D\nu}{\D\mu}-1}_{L_{\infty}}\leq M$,
we have $\bnorm{\frac{\D(\nu P^{k})}{\D\mu}-1}_{L_{\infty}}\leq\veps$
and thus $\eu R_{\infty}(\nu P^{k}\mmid\mu)\leq\veps$.
\end{proof}
Similarly, we can establish a $\eu R_{\infty}$-guarantee for the
exponential sampling considered in \S\ref{ssec:prox-exp}. 
\begin{lem}
\label{lem:ps-exp-Rinf} For given $\veps\in(0,1)$, let $P$ be the
Markov kernel of $\psexp$ with variance $h$ and target distribution
$\bar{\pi}\propto\exp(-nt)|_{\bar{\K}}$. Then, $\sup_{z_{0}\in\mc{\bar{K}}}\norm{\delta_{z_{0}}P^{k}-\bar{\pi}}_{\tv}\leq\veps$
for $k=\O(h^{-1}R^{2}l^{2}\log\frac{n+h^{-1}R^{2}}{\veps})$. If an
initial distribution $\nu$ is $M$-warm with respect to $\bar{\pi}$,
then $\eu R_{\infty}(\nu P^{k}\mmid\bar{\pi})\leq\veps$ after $k=\O(h^{-1}R^{2}l^{2}\log\frac{M(n+h^{-1}R^{2})}{\veps})$
iterations.
\end{lem}

\begin{proof}
As before, for $z_{0}\in\bar{\K}$, the composition of the forward
and backward step brings $\delta_{z_{0}}$ to
\[
(\delta_{z_{0}}P)(x,t)=\ind_{\bar{\K}}(x,t)\int\frac{\exp(-nt-\frac{1}{2h}\,\norm{y-z}^{2})\exp(-\frac{1}{2h}\,\norm{y-z_{0}}^{2})}{(2\pi h)^{(n+1)/2}\int_{\bar{\K}}\exp(-nt-\frac{1}{2h}\,\norm{y-z}^{2})\,\D z}\,\D y\,.
\]
Using $\frac{1}{2h}\,\norm{y-z}^{2}\leq\frac{3}{4h}\norm y^{2}+\frac{3}{2h}\norm z^{2}\leq\frac{3}{4h}\norm y^{2}+\frac{3D^{2}}{2h}$
(Young's inequality), for $D:=\text{diam}\,\bar{\K}\lesssim Rl$,
\begin{align*}
\delta_{z_{0}}P & \leq\frac{\exp(\frac{3D^{2}}{2h})\exp(-nt)}{(2\pi h)^{(n+1)/2}\int_{\bar{\K}}\exp(-nt)\,\D x\D t}\int\exp\bigl(-\frac{1}{2h}\,\norm{y-z}^{2}-\frac{1}{2h}\,\norm{y-z_{0}}^{2}+\frac{3}{4h}\,\norm y^{2}\bigr)\,\D y\\
 & \leq\frac{2^{(n+1)/2}\exp(\frac{6D^{2}}{h})\exp(-nt)}{\int_{\bar{\K}}\exp(-nt)\,\D x\D t}\,,
\end{align*}
where the last line follows from
\begin{align*}
 & -\frac{1}{2h}\norm{y-z}^{2}-\frac{1}{2h}\norm{y-z_{0}}^{2}+\frac{3}{4h}\norm y^{2}\\
= & -\frac{1}{4h}\norm{y-2(z+z_{0})}^{2}-\frac{1}{2h}\norm z^{2}-\frac{1}{2h}\norm{z_{0}}^{2}+\frac{1}{h}\norm{z+z_{0}}^{2}\\
\leq & -\frac{1}{4h}\norm{y-2(z+z_{0})}^{2}+\frac{4D^{2}}{h}\,.
\end{align*}
Dividing $\delta_{z_{0}}P$ by $\bar{\pi}(x,t)\propto\exp(-nt)|_{\bar{\K}}$,
\[
\frac{\delta_{z_{0}}P}{\bar{\pi}}\leq2^{\frac{n+1}{2}}\exp(6h^{-1}D^{2})\,,
\]
which implies $\eu R_{q}(\delta_{z_{0}}P\mmid\bar{\pi})\lesssim\frac{q}{q-1}\,(n+h^{-1}D^{2})$.
Noting that the truncated distribution $\bar{\pi}$ has a bounded
support, so $\clsi(\bar{\pi})=\O(D^{2})$, the $\psexp$ can achieve
$\eu R_{q}(\delta_{z_{0}}P^{k}\mmid\bar{\pi})\leq\frac{\veps}{M}$
after $k=\O(qh^{-1}D^{2}\log\frac{M(n+h^{-1}D^{2})}{\veps})$ iterations,
and
\[
\sup_{z_{0}\in\mc{\bar{K}}}\norm{\delta_{z_{0}}P^{k}-\bar{\pi}}_{\tv}\leq\frac{\veps}{M}\,.
\]
Hence, one can complete the proof using Lemma~\ref{lem:boosting}
with $\bnorm{\frac{\D\nu}{\D\bar{\pi}}-1}_{L_{\infty}}\leq M$.
\end{proof}

\subsubsection{Complexity of implementing the backward step \label{subsec:annealing-backward-complexity}}

We recall the backward step: for $\tau=\frac{\sigma^{2}}{h+\sigma^{2}}<1$,
$y_{\tau}=\tau y$, and $h_{\tau}=\tau h$, 
\[
v\sim\mu^{V|W=w}\propto\exp\bpar{-\frac{1}{2\sigma^{2}}\,\norm x^{2}-\rho t-\frac{1}{2h}\,\norm{w-v}^{2}}\big|_{\bar{\K}}\propto\bbrack{\mc N(y_{\tau},h_{\tau}I_{n})\times\mc N(s-\rho h,h)}\big|_{\bar{\mc K}}\,.
\]
Since we use the proposal $\mc N(y_{\tau},h_{\tau}I_{n})\times\mc N(s-\rho h,h)$
for rejection sampling, the success probability of each trial at $w=(y,v)$
is 
\[
\ell(w):=\frac{1}{(2\pi h_{\tau})^{n/2}(2\pi h)^{1/2}}\,\int_{\mc{\bar{K}}}\exp\bpar{-\frac{1}{2h_{\tau}}\,\bnorm{x-y_{\tau}}^{2}-\frac{1}{2h}\,\abs{t-(s-\rho h)}^{2}}\,\D x\D t\,.
\]

\paragraph{Density after taking the forward step.}

As in Lemma~\ref{lem:exp-piY-density}, we deduce the density of
$\mu^{W}=\mu^{V}*\mc N(0,hI_{n+1})$ by computing the convolution
directly.
\begin{lem}
[Density of $\mu^W$] \label{lem:anneal-muW-density} For $\mu^{V}\propto\exp(-\frac{\norm x^{2}}{2\sigma^{2}}-\rho t)|_{\bar{\K}}$,
\begin{align*}
\mu^{W}(y,s) & =\frac{\int_{\bar{\K}}\exp(-\frac{1}{2h_{\tau}}\,\norm{x-y_{\tau}}^{2}-\frac{1}{2h}\,\abs{t-(s-\rho h)}^{2})\,\D x\D t}{\int_{\bar{\K}}\exp(-\frac{1}{2\sigma^{2}}\,\norm x^{2}-\rho t)\,\D x\D t\cdot(2\pi h)^{(n+1)/2}}\,\exp\bpar{-\frac{1}{2\tau\sigma^{2}}\,\norm{y_{\tau}}^{2}}\exp\bpar{-\rho s+\half\,\rho^{2}h}\\
 & =\frac{\tau^{n/2}\ell(y,s)}{\int_{\bar{\K}}\exp(-\frac{1}{2\sigma^{2}}\norm x^{2}-\rho t)\,\D x\D t}\,\exp\bpar{-\frac{1}{2\tau\sigma^{2}}\,\norm{y_{\tau}}^{2}}\exp\bpar{-\rho s+\half\,\rho^{2}h}\,,
\end{align*}
\end{lem}

We now claim that the following is the \emph{effective domain} of
$\mu^{W}$:
\begin{equation}
\widetilde{\K}=\left[\begin{array}{cc}
\tau^{-1}I_{n}\\
 & 1
\end{array}\right]\bar{\mc K}_{\delta}+\left[\begin{array}{c}
0\\
\rho h
\end{array}\right]\,.\label{eq:annealing-eff-domain}
\end{equation}

\begin{lem}
\label{lem:annealing-eff-domain} In the setting of Problem~\ref{prob:warm-start-generation},
for $\widetilde{\K}$ in \eqref{eq:annealing-eff-domain}, if $\delta\geq24hn$,
then
\[
\mu^{W}(\widetilde{\K}^{c})\leq\exp\bpar{-\frac{\delta^{2}}{2h}+24\delta n+hn^{2}}\,.
\]
\end{lem}

\begin{proof}
Using the density formula of $\mu^{W}$ in Lemma~\ref{lem:anneal-muW-density},
for $v=(x,t)$, $w=(y,s)$, and $\bar{\K}_{\delta}^{c}:=(\bar{\K}_{\delta})^{c}$,
\begin{align*}
 & \int_{\bar{\mc K}}\exp\bpar{-\frac{1}{2\sigma^{2}}\,\norm x^{2}-\rho t}\,\D x\D t\cdot\mu^{W}(\widetilde{\K}^{c})\\
 & =\frac{1}{(2\pi h)^{\frac{n+1}{2}}}\int_{\widetilde{\K}^{c}}\int_{\bar{\K}}\exp\bpar{-\frac{1}{2h_{\tau}}\,\bnorm{x-y_{\tau}}^{2}-\frac{1}{2h}\,\abs{t-(s-\rho h)}^{2}}\exp\bpar{-\frac{\norm{y_{\tau}}^{2}}{2\tau\sigma^{2}}-\rho s+\frac{\rho^{2}h}{2}}\,\D v\D w\\
 & \underset{(i)}{=}\frac{\tau^{-n/2}}{(2\pi h_{\tau})^{n/2}(2\pi h)^{1/2}}\int_{\bar{\K}_{\delta}^{c}}\int_{\bar{\mc K}}\exp\bpar{-\frac{1}{2h_{\tau}}\,\bnorm{x-y}^{2}-\frac{1}{2h}\,\abs{t-s}^{2}}\exp\bpar{-\frac{\norm y^{2}}{2\tau\sigma^{2}}-\rho s-\frac{\rho^{2}h}{2}}\,\D v\D w\\
 & \underset{(ii)}{\leq}\tau^{-n/2}\int_{\mc{\bar{K}}_{\delta}^{c}}e^{-\frac{1}{2\tau\sigma^{2}}\,\norm y^{2}-\rho s}\frac{1}{(2\pi h_{\tau})^{n/2}(2\pi h)^{1/2}}\int_{\mc H(w)}\exp\bpar{-\frac{1}{2h_{\tau}}\,\bnorm{x-y}^{2}-\frac{1}{2h}\,\abs{t-s}^{2}}\,\D v\D w\\
 & \leq\tau^{-n/2}\int_{\bar{\mc K}_{\delta}^{c}}e^{-\frac{1}{2\tau\sigma^{2}}\,\norm y^{2}-\rho s}\Bpar{\frac{1}{\sqrt{2\pi h}}\int_{d(w,\bar{\K})}^{\infty}\exp\bpar{-\frac{u^{2}}{2h}}\,\D u}\,\D w=:(\#)
\end{align*}
where in $(i)$ we used the change of variables via $y\gets\tau y$
and $s\gets s-\rho h$, and in $(ii)$ $\mc H(w)$ denotes the supporting
half-space at $\msf{proj}_{\bar{\mc K}}(w)$ containing $\bar{\mc K}$
for given $w\in\bar{\K}_{\delta}^{c}$. By the co-area formula and
integration by parts, for the $n$-dimensional Hausdorff measure $\mc H^{n}$,
\begin{align*}
(\#) & \underset{\text{coarea}}{=}\tau^{-n/2}\int_{\delta}^{\infty}\Phi^{c}(h^{-1/2}u)\int_{\de\bar{\K}_{u}}e^{-\frac{1}{2\tau\sigma^{2}}\,\norm y^{2}-\rho s}\,\mc H^{n}(\D w)\D u\\
 & \underset{\text{IBP}}{=}\tau^{-n/2}\Bbrack{\Phi^{c}(h^{-1/2}u)\int_{0}^{u}\int_{\de\bar{\K}_{l}}e^{-\frac{1}{2\tau\sigma^{2}}\,\norm y^{2}-\rho s}\,\mc H^{n}(\D w)\D l}_{u=\delta}^{\infty}\\
 & \qquad+\tau^{-n/2}\int_{\delta}^{\infty}\frac{1}{\sqrt{2\pi h}}\exp\bpar{-\frac{u^{2}}{2h}}\int_{0}^{u}\int_{\de\bar{\K}_{l}}e^{-\frac{1}{2\tau\sigma^{2}}\,\norm y^{2}-\rho s}\,\mc H^{n}(\D w)\D l\D u\,,
\end{align*}
and the double integral term is bounded as
\begin{align*}
\tau^{-n/2}\int_{0}^{u}\int_{\de\bar{\K}_{l}}e^{-\frac{1}{2\tau\sigma^{2}}\,\norm y^{2}-\rho s}\,\mc H^{n}(\D w)\D l & =\tau^{-n/2}\int_{\bar{\K}_{u}\backslash\bar{\K}}e^{-\frac{1}{2\tau\sigma^{2}}\,\norm y^{2}-\rho s}\,\D w\\
 & \underset{\text{Lemma }\ref{lem:annealing-expansion}}{\leq}2\exp(hn^{2}+24un)\int_{\bar{\K}}e^{-\frac{1}{2\sigma^{2}}\,\norm y^{2}-\rho s}\,\D y\D s\,.
\end{align*}
Using the tail bound on $\Phi^{c}$ in \eqref{eq:Gaussian-tail},
$\tau^{-n/2}\Phi^{c}(h^{-1/2}u)\int_{\bar{\K}_{u}\backslash\bar{\K}}\exp(-\frac{1}{2\tau\sigma^{2}}\,\norm y^{2}-\rho s)\,\D w$
vanishes at $u=\infty$, and thus the first term of $(\#)$ can bounded
by $0$. Hence,
\[
\int_{\bar{\K}}e^{-\frac{1}{2\sigma^{2}}\,\norm x^{2}-\rho t}\,\D x\D t\cdot\mu^{W}(\widetilde{\K}^{c})\leq2\int_{\delta}^{\infty}\frac{1}{\sqrt{2\pi h}}\exp\bpar{-\frac{u^{2}}{2h}}\exp(hn^{2}+24un)\,\D u\int_{\bar{\K}}e^{-\frac{1}{2\sigma^{2}}\,\norm x^{2}-\rho t}\,\D x\D t\,.
\]
Dividing both sides by $\int_{\bar{\K}}\exp(-\frac{1}{2\sigma^{2}}\,\norm x^{2}-\rho t)$
and using the standard bound on $\Phi^{c}$ when $\delta\geq23hn$,
\begin{align*}
\mu^{W}(\widetilde{\K}^{c}) & \leq2\int_{\delta}^{\infty}\frac{1}{\sqrt{2\pi h}}e^{-\frac{s^{2}}{2h}}e^{hn^{2}+24sn}\,\D s=2e^{hn^{2}}\Bpar{1-\Phi\bpar{\frac{\delta-24hn}{h^{1/2}}}}\\
 & \leq\exp\bpar{-\frac{\delta^{2}}{2h}+24\delta n+hn^{2}}\,,
\end{align*}
which completes the proof. 
\end{proof}
This suggests taking $\delta\asymp1/n$ and $h\asymp1/n^{2}$ as in
the exponential sampling. Precisely, under the choice of $h=\tfrac{c}{24^{2}n^{2}}$
and $\delta=\tfrac{c/24+t}{n}$ for parameters $c,t>0$, it follows
that $\delta\geq24hn$ and
\[
\pi^{Y}(\widetilde{\K}^{c})\leq\exp\bpar{-\frac{t^{2}}{2c}+c}\,.
\]

We control how fast $\int_{\bar{\K}_{u}}\exp(-\frac{1}{2\sigma^{2}}\,\norm x^{2}-\rho t)$
increases in $u$.
\begin{lem}
\label{lem:annealing-expansion} In the setting of Lemma~\ref{lem:annealing-eff-domain},
for $\tau=\frac{\sigma^{2}}{h+\sigma^{2}}<1$, $s>0$, and $\rho\in(0,n]$,
\[
\tau^{-n/2}\int_{\bar{\K}_{s}}e^{-\frac{1}{2\tau\sigma^{2}}\,\norm z^{2}-\rho l}\,\D z\D l\leq2\exp(hn^{2}+24sn)\int_{\bar{\K}}e^{-\frac{1}{2\sigma^{2}}\,\norm z^{2}-\rho l}\,\D z\D l\,.
\]
\end{lem}

\begin{proof}
Consider when $\sigma^{2}\leq\nicefrac{1}{2n}$. For $F:(z,l)\mapsto(\tau^{-1/2}z,l)$
between $\R^{n+1}$, and $\hat{\K}:=\bar{\K}+\frac{V(x_{0})}{n}-\frac{V_{0}}{n}$,
\begin{align*}
 & \tau^{-n/2}\int_{\mc{\hat{\K}}_{s}}\exp\bpar{-\frac{1}{2\tau\sigma^{2}}\,\norm z^{2}-\rho l}\\
 & \underset{(i)}{\le}\tau^{-n/2}\int_{(1+s)\hat{\K}}\exp\bpar{-\frac{1}{2\tau\sigma^{2}}\,\norm z^{2}-\rho l}=(1+s)^{n}\int_{F\hat{\K}}\exp\bpar{-\frac{(1+s)^{2}}{2\sigma^{2}}\,\norm z^{2}-\rho l\,(1+s)}\\
 & \underset{(ii)}{\leq}e^{13sn}\int_{F\hat{\K}}\exp\bpar{-\frac{1}{2\sigma^{2}}\,\norm z^{2}-\rho l}\\
 & \underset{(iii)}{=}e^{13sn}\int_{\hat{\K}}\exp\bpar{-\frac{1}{2\sigma^{2}}\,\norm z^{2}-\rho l}\times\frac{\P_{\mc N(0,\sigma^{2}I_{n})\times\msf{Exp}}\bpar{(X,T)\in F\hat{\K})}}{\P_{\mc N(0,\sigma^{2}I_{n})\times\msf{Exp}}\bpar{(X,T)\in\hat{\K})}}\\
 & \underset{(iv)}{\leq}e^{13sn}\int_{\hat{\K}}\exp\bpar{-\frac{1}{2\sigma^{2}}\,\norm z^{2}-\rho l}\times\bpar{\P_{\mc N(0,\sigma^{2}I_{n})\times\msf{Exp}}\bpar{\norm X\leq1,T\geq0}}^{-1}\\
 & \underset{(v)}{\leq}\frac{e^{13sn}e^{11sn}}{1-d\sigma^{2}}\int_{\hat{\K}}\exp\bpar{-\frac{1}{2\sigma^{2}}\,\norm z^{2}-\rho l}\leq2e^{24sn}\int_{\hat{\K}}\exp\bpar{-\frac{1}{2\sigma^{2}}\,\norm z^{2}-\rho l}\,.
\end{align*}
In $(i)$, we used $\hat{\K}_{s}\subset(1+s)\,\hat{\K}$ due to $B_{1}(0)\subset\hat{\K}$.
$(ii)$ follows from $l\geq-11$ in $F\hat{\K}$. In $(iii)$, $\msf{Exp}$
denotes the $1$-dim. exponential distribution with density proportional
to $\exp(-\rho t)|_{[-t_{0},\infty)}$ for $t_{0}=\min_{\hat{\K}}t\geq-11$.
$(iv)$ and $(v)$ follow from $B_{1}(0)\subset\hat{\K}$ and Chebyshev's
inequality with $\rho\leq n$, respectively.

When $\sigma^{2}\geq\nicefrac{1}{2n}$, starting from the $(i)$-line
and using $\tau^{-n/2}\leq\exp(\nicefrac{hn}{2\sigma^{2}})$,
\begin{align*}
\tau^{-n/2}\int_{\hat{\K}_{s}}\exp\bpar{-\frac{1}{2\tau\sigma^{2}}\,\norm z^{2}-\rho l} & \leq\exp\bpar{\frac{hn}{2\sigma^{2}}+sn}\int_{\hat{\K}}\exp\bpar{-\frac{1}{2\tau\sigma^{2}}\,\norm z^{2}-\rho l\,(1+s)}\\
 & \leq\exp(hn^{2}+12sn)\int_{\hat{\K}}\exp\bpar{-\frac{1}{2\sigma^{2}}\,\norm z^{2}-\rho l}\,.
\end{align*}
We complete the proof by combining the two bounds in both cases.
\end{proof}

\paragraph{Per-step guarantees.}

We provide the per-step query complexity and failure probability of
implementing the backward step from a warm start.
\begin{lem}
\label{lem:annealing-perstep} In the setting of Problem~\ref{prob:warm-start-generation},
let $\nu$ be an $M$-warm initial distribution for $\mu^{V}=\mu_{\sigma^{2},\rho}$
with $\sigma^{2}>0$ and $\rho\in(0,n]$. Given $k\in\mbb N$ and
$\eta\in(0,1)$, set $Z=\frac{16kM}{\eta}(\geq16)$, $h=\frac{1}{n^{2}}\frac{(\log\log Z)^{2}}{1200^{2}\log Z}$
and $N=2Z\log^{2}Z$. Then, per iteration, the failure probability
is at most $\eta/k$, and the expected number of queries is $\Otilde(M\log^{2}\tfrac{k}{\eta})$.
\end{lem}

\begin{proof}
Since $\nu*\mc N(0,hI_{n+1})$ is $M$-warm with respect to $\mu^{W}$
due to $\D\nu/\D\mu^{V}\leq M$, the probability for the bad event
per iteration is $\E_{\nu*\mc N(0,hI_{n+1})}[(1-\ell)^{N}]\leq M\E_{\mu^{W}}[(1-\ell)^{N}]$.
For $\widetilde{\K}$ given in \eqref{eq:annealing-eff-domain},
\begin{align*}
\E_{\mu^{W}}[(1-\ell)^{N}] & =\int_{\widetilde{\K}^{c}}\cdot+\int_{\widetilde{\K}\cap[\ell\geq N^{-1}\log(3kM/\eta)]}\cdot+\int_{\widetilde{\K}\cap[\ell\leq N^{-1}\log(3kM/\eta)]}\cdot\\
 & \underset{(i)}{\leq}\mu^{W}(\widetilde{\K}^{c})+\int_{[\ell\geq N^{-1}\log(3kM/\eta)]}\exp(-\ell N)\,\D\mu^{W}\\
 & \qquad+\int_{\widetilde{\K}\cap[\ell\leq N^{-1}\log(3kM/\eta)]}\frac{\tau^{n/2}\ell(y,s)\exp(-\frac{1}{2\tau\sigma^{2}}\,\norm{y_{\tau}}^{2})\exp(-\rho s+\half\,\rho^{2}h)}{\int_{\bar{\K}}\exp(-\frac{1}{2\sigma^{2}}\,\norm x^{2}-\rho t)}\\
 & \underset{(ii)}{\leq}\exp\bpar{-\frac{t^{2}}{2c}+c}+\frac{\eta}{3kM}+\frac{\log(3kM/\eta)}{N}\frac{\int_{\widetilde{\K}}\tau^{n/2}\exp(-\frac{1}{2\tau\sigma^{2}}\,\norm{y_{\tau}}^{2}-\rho s+\half\,\rho^{2}h)}{\int_{\bar{\K}}\exp(-\frac{1}{2\sigma^{2}}\,\norm x^{2}-\rho t)}\\
 & \underset{(iii)}{=}\exp\bpar{-\frac{t^{2}}{2c}+c}+\frac{\eta}{3kM}+\frac{\log(3kM/\eta)}{N}\frac{\tau^{-n/2}\int_{\bar{\K}_{\delta}}\exp(-\frac{1}{2\tau\sigma^{2}}\,\norm y^{2}-\rho s-\half\,\rho^{2}h)}{\int_{\bar{\K}}\exp(-\frac{1}{2\sigma^{2}}\,\norm x^{2}-\rho t)}\\
 & \underset{(iv)}{\leq}\exp\bpar{-\frac{t^{2}}{2c}+c}+\frac{\eta}{3kM}+\frac{\log(3kM/\eta)}{N}\,2\exp(hn^{2}+24\delta n)\\
 & \underset{(v)}{\leq}\exp\bpar{-\frac{t^{2}}{4c}}+\frac{\eta}{3kM}+\frac{\log(3kM/\eta)}{N}\,2\exp(25t)\leq\frac{\eta}{kM}\,,
\end{align*}
where in $(i)$ we used the density formula of $\mu^{W}$ in Lemma~\ref{lem:anneal-muW-density},
$(ii)$ follows from Lemma~\ref{lem:annealing-eff-domain}, $(iii)$
is due to the change of variables, in $(iv)$ we used Lemma~\ref{lem:annealing-expansion},
and $(v)$ follows from the choices of $c=\frac{(\log\log Z)^{2}}{4\cdot24^{2}\log Z}$,
$t=\tfrac{1}{24}\log\log Z$, and $N=2Z\log^{2}Z$. Therefore, $\E_{\nu*\mc N(0,hI_{n+1})}[(1-\ell)^{N}]\leq\tfrac{\eta}{k}$.

The expected number of trials until the first acceptance is bounded
by $M\E_{\mu^{W}}[\ell^{-1}\wedge N]$, where
\begin{align*}
\int_{\R^{n+1}}\frac{1}{\ell}\wedge N\,\D\mu^{W} & \leq\int_{\widetilde{\K}}\frac{1}{\ell}\,\D\mu^{W}+N\mu^{W}(\widetilde{\K}^{c})\\
 & \underset{(i)}{\leq}\frac{\int_{\widetilde{\K}}\tau^{n/2}\exp(-\frac{1}{2\tau\sigma^{2}}\,\norm{y_{\tau}}^{2}-\rho s+\half\rho^{2}h)\,\D y\D s}{\int_{\bar{\K}}\exp(-\frac{1}{2\sigma^{2}}\,\norm x^{2}-\rho t)\,\D x\D t}+N\exp\bpar{-\frac{t^{2}}{2c}+c}\\
 & \underset{(ii)}{\leq}2\exp(hn^{2}+24\delta n)+N\exp\bpar{-\frac{t^{2}}{2c}+c}\\
 & \leq4\log^{2}Z\,.
\end{align*}
where $(i)$ follows from Lemma~\ref{lem:anneal-muW-density} and
Lemma~\ref{lem:annealing-eff-domain}, and in $(ii)$ we repeat the
arguments around $(iii),(iv)$ above. Therefore, the expected number
of queries per iteration is $\Otilde(M\log^{2}\tfrac{k}{\eta})$.
\end{proof}
Combining the convergence rate and per-step guarantees of the $\psann$,
we establish the query complexity of sampling from the annealing distribution
from an $M$-warm start:
\begin{thm}
\label{thm:annealing-query-comp} In the setting of Problem~\ref{prob:warm-start-generation}
without knowledge of $x_{0}$ and $V_{0}$, for any given $\eta,\veps\in(0,1)$,
$k\in\mbb N$ defined below, $\psann$ with $h=(1200^{2}n^{2}\log\frac{16kM}{\eta})^{-1}$,
$N=\frac{32kM}{\eta}\log^{2}\frac{16kM}{\eta}$, and initial distribution
$\mu_{0}$ which is $M$-warm with respect to $\mu:=\mu_{\sigma^{2},\rho}$
(where $\sigma>0$ and $\rho\in(0,n]$) achieves $\eu R_{\infty}(\mu_{k}\mmid\mu)\leq\veps$
after $k=\Otilde(n^{2}(\sigma^{2}\vee l^{2})\log^{2}\frac{MR}{\eta\veps})$
iterations, where $\mu_{k}$ is the law of $k$-th iterate. With probability
at least $1-\eta$, the algorithm iterates $k$ times successfully,
using $\Otilde(Mn^{2}(\sigma^{2}\vee l^{2})\log^{4}\frac{R}{\eta\veps})$
expected number of evaluation queries in total. 
\end{thm}

\begin{proof}
By Lemma~\ref{lem:ps-annealing-iter}, $\psann$ with given $h$
should iterate $k=\Otilde(n^{2}(\sigma^{2}\vee l^{2})\log^{2}\frac{MR}{\eta\veps})$
times to achieve $\veps$-distance in $\eu R_{\infty}$. Under the
choice of $h$ and $N$ as claimed, each iteration succeeds with probability
$1-\nicefrac{\eta}{k}$ and uses $\Otilde(M\log^{2}\nicefrac{k}{\eta})$
queries in expectation by Lemma~\ref{lem:annealing-perstep}. Therefore,
throughout $k$ iterations, the algorithm succeeds with probability
$1-\eta$, using $\Otilde(Mn^{2}(\sigma^{2}\vee l^{2})\log^{4}\frac{R}{\eta\veps})$
queries.
\end{proof}

\subsection{R\'enyi-infinity sampling for logconcave distributions}

Now that we have a sampler for the annealing distribution with $\eu R_{\infty}$-guarantee,
we analyze a query complexity of each phase in $\tgc$.

\paragraph{Initialization.}

In the initialization where $\sigma^{2}=1/n$ and $\rho=1/l$, we
run $\psann$ with initial distribution $\mu_{0}=\exp(-\frac{n}{2}\,\norm x^{2})|_{\bar{\K}}$
and target $\mu_{1}\propto\exp(-\frac{n}{2}\,\norm x^{2}-t/l)|_{\bar{\K}}$.
Recall that $l=\log\frac{2e}{\veps}$. 
\begin{lem}
[Initialization] \label{lem:init} Rejection sampling for the initial
$\mu_{0}\propto\exp(-\frac{n}{2}\,\norm x^{2})|_{\bar{\K}}$ takes
$\O(l)$ queries in expectation. Given $\eta\in(0,1)$, $\psann$\emph{
}with\emph{ }initial $\mu_{0}$, target $\mu_{1}\propto\exp(-\frac{n}{2}\,\norm x^{2}-t/l)|_{\bar{\K}}$,
$h^{-1}=\widetilde{\Theta}(n^{2}\log\frac{k}{\eta})$, $N=\Thetilde(k/\eta)$,
and $k=\Thetilde(n^{2}l^{2}\log^{4}\frac{MR}{\eta})$ iterates $k$
times with probability at least $1-\eta$, returning a sample with
law $\bar{\mu}_{1}$ satisfying 
\[
\esssup_{\mu_{1}}\,\bigl|\frac{\D\bar{\mu}_{1}}{\D\mu_{1}}-1\bigr|\leq1\,,
\]
where the expected number of evaluation queries used is $\Otilde(n^{2}l^{2}\log^{4}\frac{R}{\eta})$.
\end{lem}

\begin{proof}
In the initialization through rejection sampling, the expected number
of trials until the first acceptance is bounded by 
\begin{align*}
\frac{\int_{\R^{n}}\exp(-\frac{n}{2}\,\norm x^{2})\,\D x\cdot(15+13l)}{\int_{\bar{\K}}\exp(-\frac{n}{2}\,\norm x^{2})\,\D x\D t} & \lesssim\frac{l\int_{\Rn}\exp(-\frac{n}{2}\norm x^{2})\,\D x}{\int_{B_{1}(0)}\exp(-\frac{n}{2}\,\norm x^{2})\,\D x}=l\,\bigl(\P_{\mc N(0,n^{-1}I_{n})}(\norm X\leq1)\bigr)^{-1}\\
 & \underset{(i)}{\leq}l\,(1-e^{-c\sqrt{n}})^{-1}=\O(l)\,,
\end{align*}
where $(i)$ follows from Paouris' lemma (Lemma~\ref{lem:version-Paouris})
for some universal constant $c$.

The warmness of $\mu_{0}$ with respect to $\mu_{1}$ is bounded as
\begin{align*}
\frac{\D\mu_{0}}{\D\mu_{1}} & =\sup_{(x,t)\in\bar{\K}}\frac{\exp(-\frac{n}{2}\,\norm x^{2})}{\int_{\bar{\K}}\exp(-\frac{n}{2}\,\norm x^{2})\,\D x\D t}\frac{\int_{\bar{\K}}\exp(-\frac{n}{2}\norm x^{2}-t/l)\,\D x\D t}{\exp(-\frac{n}{2}\norm x^{2}-t/l)}\\
 & \lesssim\sup_{\bar{\K}}\exp(t/l)\,\frac{\int_{\bar{\K}}\exp(-\frac{n}{2}\norm x^{2}-t/l)\,\D x\D t}{\int_{\bar{\K}}\exp(-\frac{n}{2}\,\norm x^{2})\,\D x}\lesssim1\,.
\end{align*}
Thus, $\psann$ with initial $\mu_{0}$ and target $\mu_{1}$ can
be run with target error of $\log2$ in the $\eu R_{\infty}$-divergence,
under a suitable choice of parameters given in Theorem~\ref{thm:annealing-query-comp}.
\end{proof}

\paragraph{Phase I.}

In the first phase ($\sigma^{2}\in[\frac{1}{n},1]$ and $\rho\in[\frac{1}{l},n]$),
these parameters are updated according to 
\[
\sigma_{i+1}^{2}=\sigma_{i}^{2}\bpar{1+\frac{1}{n}}\quad\&\quad\rho_{i+1}=\rho_{i}\bpar{1+\frac{1}{ln}}\,.
\]
Note that $\psann$ is initialized at $\bar{\mu}_{i}$ (not $\mu_{i}$)
for the target $\mu_{i+1}$.
\begin{lem}
[Phase I] \label{lem:tgc-phase1} Given $\eta\in(0,1)$, Phase I
of $\tgc$ started at $\bar{\mu}_{1}$ succeeds with probability at
least $1-\eta$, returning a sample with law $\bar{\mu}_{\msf I}$
satisfying 
\[
\esssup_{\bar{\K}}\,\bigl|\frac{\D\bar{\mu}_{\msf I}}{\D\mu_{\msf I}}-1\bigr|\leq1\quad\text{for }\mu_{\msf I}\propto\exp(-\frac{1}{2}\,\norm x^{2}-nt)|_{\bar{\K}}\,,
\]
using $\psann$ with suitable parameters. Phase I uses $\Otilde(n^{3}l^{3}\log^{5}\frac{R}{\eta})$
evaluation queries in expectation. 
\end{lem}

\begin{proof}
Since it takes at most $ln$ many inner phases to double $\sigma^{2}$
and $\rho$, the total number of inner phases within Phase I is $m=\O(nl\log nl)=\Otilde(nl)$.

The warmness between consecutive annealing distributions is bounded
as
\begin{align*}
\frac{\D\mu_{i}}{\D\mu_{i+1}} & \leq\sup_{(x,t)\in\bar{\K}}\frac{\exp(-\frac{1}{2\sigma_{i}^{2}}\,\norm x^{2}-\rho_{i}t)}{\exp(-\frac{1}{2\sigma_{i+1}^{2}}\,\norm x^{2}-\rho_{i+1}t)}\frac{\int_{\bar{\K}}\exp(-\frac{1}{2\sigma_{i+1}^{2}}\,\norm x^{2}-\rho_{i+1}t)\,\D x\D t}{\int_{\bar{\K}}\exp(-\frac{1}{2\sigma_{i}^{2}}\,\norm x^{2}-\rho_{i}t)\,\D x\D t}\\
 & \leq\sup_{(x,t)\in\bar{\K}}\exp\bpar{\frac{\rho_{i}\,(t-t_{0})}{ln}}\frac{\int_{\bar{\K}}\exp\bigl(-\frac{1}{2\sigma_{i+1}^{2}}\,\norm x^{2}-\rho_{i}\,(t-t_{0})\bigr)\,\D x\D t}{\int_{\bar{\K}}\exp\bigl(-\frac{1}{2\sigma_{i}^{2}}\,\norm x^{2}-\rho_{i}\,(t-t_{0})\bigr)\,\D x\D t}\\
 & \underset{(i)}{\lesssim}\frac{(1+\frac{1}{n})^{n/2}\int_{T\bar{\K}}\exp\bigl(-\frac{1}{2\sigma_{i}^{2}}\,\norm x^{2}-\rho_{i}\,(t-t_{0})\bigr)\,\D x\D t}{\int_{\bar{\K}}\exp\bigl(-\frac{1}{2\sigma_{i}^{2}}\,\norm x^{2}-\rho_{i}\,(t-t_{0})\bigr)\,\D x\D t}\leq\sqrt{e}\,,
\end{align*}
where in $(i)$ we used $\rho_{i}\leq n$ and $t-t_{0}\lesssim l$,
and the change of variables via $T:(x,t)\mapsto((1+\frac{1}{n})^{-1/2}x,t)$. 

Due to the design of the algorithm, $\psann$ always ensures that
$\esssup_{\mu_{i}}|\frac{\D\bar{\mu}_{i}}{\D\mu_{i}}-1|\leq1$, and
thus $\frac{\D\bar{\mu}_{i}}{\D\mu_{i+1}}\leq2\sqrt{e}$. Hence, for
each inner phase, $\psann$ with target accuracy $\log2$ in $\eu R_{\infty}$,
target failure probability $\eta/m$, warmness $M=2\sqrt{e}$, and
suitable parameters succeeds with probability $1-\eta/m$, returning
a sample with law $\bar{\mu}_{i+1}$ such that
\[
\esssup_{\bar{\K}}\,\bigl|\frac{\D\bar{\mu}_{i+1}}{\D\mu_{i+1}}-1\bigr|\leq1\,,
\]
and using $\Otilde(n^{2}l^{2}\log^{4}\frac{R}{\eta})$ many evaluation
queries in expectation. Repeating this argument $m$ many times, we
complete the proof.
\end{proof}

\paragraph{Phase II.}

In this second phase ($\sigma^{2}\in[1,l^{2}R^{2}]$), only $\sigma^{2}$
is updated according to
\[
\sigma_{i+1}^{2}=\sigma_{i}^{2}\bpar{1+\frac{\sigma_{i}^{2}}{l^{2}R^{2}}}\,.
\]

\begin{lem}
[Phase II] \label{lem:tgc-phase2} Given $\eta\in(0,1)$, Phase II
of $\tgc$ started at $\bar{\mu}_{\msf I}$ succeeds with probability
at least $1-\eta$, returning a sample with law $\bar{\mu}_{\msf{II}}$
satisfying 
\[
\esssup_{\bar{\K}}\,\bigl|\frac{\D\bar{\mu}_{\msf{II}}}{\D\mu_{\msf{II}}}-1\bigr|\leq1\quad\text{for }\mu_{\msf{II}}\propto\exp(-\frac{1}{2l^{2}R^{2}}\,\norm x^{2}-nt)|_{\bar{\K}}\,,
\]
using $\psann$ with suitable parameters. Phase II uses $\Otilde(n^{2}R^{2}l^{4}\log^{4}\frac{1}{\eta})$
 queries in expectation. 
\end{lem}

\begin{proof}
We note that for given $\sigma^{2}$, it takes at most $l^{2}R^{2}/\sigma^{2}$
inner phases to double $\sigma^{2}$, so the total number of inner
phases within Phase II is $m=\Otilde(l^{2}R^{2})$.

The warmness between consecutive annealing distributions is bounded
as
\begin{align*}
\frac{\D\mu_{i}}{\D\mu_{i+1}} & \leq\sup_{(x,t)\in\bar{\K}}\frac{\exp(-\frac{1}{2\sigma_{i}^{2}}\,\norm x^{2}-nt)}{\exp(-\frac{1}{2\sigma_{i+1}^{2}}\,\norm x^{2}-nt)}\frac{\int_{\bar{\K}}\exp(-\frac{1}{2\sigma_{i+1}^{2}}\,\norm x^{2}-nt)\,\D x\D t}{\int_{\bar{\K}}\exp(-\frac{1}{2\sigma_{i}^{2}}\,\norm x^{2}-nt)\,\D x\D t}\\
 & \leq\frac{\int_{\bar{\K}}\exp(-\frac{1}{2\sigma_{i+1}^{2}}\,\norm x^{2}-nt)\,\D x\D t}{\int_{\bar{\K}}\exp(-\frac{1}{2\sigma_{i}^{2}}\,\norm x^{2}-nt)\,\D x\D t}\leq\sqrt{e}\,,
\end{align*}
where the last line holds, since $\norm x^{2}\leq l^{2}R^{2}$ on
$\bar{\K}$, 
\[
\exp\bpar{-\frac{1}{2\sigma_{i+1}^{2}}\norm x^{2}}=\exp\bpar{-\frac{1}{2\sigma_{i}^{2}}\norm x^{2}}\exp\bpar{\frac{1}{2(l^{2}R^{2}+\sigma_{i}^{2})}\norm x^{2}}\le\sqrt{e}\exp\bpar{-\frac{1}{2\sigma_{i}^{2}}\norm x^{2}}\,.
\]
As in the analysis for Phase I, this implies that due to the design
of the algorithm,
\[
\frac{\D\bar{\mu}_{i}}{\D\mu_{i+1}}\leq2\sqrt{e}\,.
\]
Hence, for each inner phase, $\psann$ with target accuracy $\log2$
in $\eu R_{\infty}$, target failure probability $\eta/m$, warmness
$M=2\sqrt{e}$, and suitable parameters succeeds with probability
$1-\eta/m$, returning a sample with law $\bar{\mu}_{i+1}$ such that
\[
\esssup_{\bar{\K}}\,\bigl|\frac{\D\bar{\mu}_{i+1}}{\D\mu_{i+1}}-1\bigr|\leq1\,,
\]
and using $\Otilde(n^{2}(\sigma^{2}\vee l^{2})\log^{4}\frac{R}{\eta})$
many evaluation queries in expectation. Each doubling part requires
\[
\Otilde\bpar{\frac{l^{2}R^{2}}{\sigma^{2}}\times n^{2}(\sigma^{2}\vee l^{2})\log^{4}\frac{R}{\eta}}=\Otilde\bpar{n^{2}R^{2}l^{4}\log^{4}\frac{1}{\eta}}
\]
many queries in expectation. Therefore, the claim follows from multiplying
this by $\log_{2}(lR)$.
\end{proof}

\paragraph{Termination.}

In the last phase, $\psexp$ is run with initial distribution $\bar{\mu}_{\msf{II}}$
and target $\bar{\pi}\propto e^{-nt}|_{\bar{\K}}$.
\begin{lem}
[Termination] Given $\eta,\veps\in(0,1)$, $\psexp$\emph{ with }initial
distribution $\bar{\mu}_{\msf{II}}$, $h^{-1}=\widetilde{\Theta}(n^{2}\log\frac{k}{\eta})$,
$N=\Thetilde(k/\eta)$, and $k=\Thetilde(n^{2}R^{2}l^{2}\log^{4}\frac{M}{\eta\veps})$
iterates $k$ times with probability at least $1-\eta$, returning
a sample with law $\nu$ satisfying 
\[
\esssup_{\bar{\pi}}\,\bigl|\frac{\D\nu}{\D\bar{\pi}}-1\bigr|\leq\veps\,,
\]
where the expected number of evaluation queries is $\Otilde(n^{2}R^{2}l^{2}\log^{4}\frac{1}{\eta\veps})$.
\end{lem}

\begin{proof}
As $\norm x^{2}\leq l^{2}R^{2}$ on $\bar{\K}$, the warmness between
$\bar{\mu}_{\msf{II}}$ and $\bar{\pi}$ is bounded as
\[
\frac{\D\mu_{\msf{II}}}{\D\bar{\pi}}\leq\sup_{(x,t)\in\bar{\K}}\frac{\exp(-\frac{1}{2l^{2}R^{2}}\,\norm x^{2}-nt)}{\exp(-nt)}\frac{\int_{\bar{\K}}\exp(-nt)\,\D x\D t}{\int_{\bar{\K}}\exp(-\frac{1}{2l^{2}R^{2}}\,\norm x^{2}-nt)\,\D x\D t}\leq\sqrt{e}\,,
\]
and thus $\frac{\D\bar{\mu}_{\msf{II}}}{\D\bar{\pi}}\leq2\sqrt{e}$.
Therefore, by Lemma~\ref{lem:ps-exp-Rinf}, $\psexp$ with initial
$\bar{\mu}_{\msf{II}}$ and suitable parameters above outputs a sample
with law $\nu$ such that 
\[
\esssup_{\bar{\K}}\,\bigl|\frac{\D\nu}{\D\bar{\pi}}-1\bigr|\leq\veps\,,
\]
using $\Otilde(n^{2}R^{2}l^{2}\log^{4}\frac{1}{\eta\veps})$ evaluation
queries in expectation.
\end{proof}

\paragraph{Final guarantee.}

As claimed in Theorem~\ref{thm:lc-warmstart-intro}, we prove that
$\tgc$ provides an $\O(1)$-warm start for $\pi(x,t)\propto\exp(-nt)|_{\K}$,
using $\Otilde(n^{2}(R^{2}\vee n))$ evaluation queries in expectation.
\begin{thm}
[Warm-start generation] \label{thm:warm-start} In the setting of
Problem~\ref{prob:warm-start-generation}, consider $\pi\propto\exp(-nt)|_{\K}$
given in \eqref{eq:exp-reduction}. For given $\eta,\veps\in(0,1)$,
$\tgc$ succeeds with probability at least $1-\eta$, returning a
sample with law $\nu$ such that $\eu R_{\infty}(\nu\mmid\pi)\leq\veps$,
using $\Otilde(n^{2}(R^{2}\vee n)\log^{8}\nicefrac{1}{\eta\veps})$
evaluation queries in expectation. In particular, if $\pi^{X}$ is
well-rounded (i.e., $R^{2}=\O(n)$), then we need $\Otilde(n^{3}\log^{8}\nicefrac{1}{\eta\veps})$
queries in expectation.
\end{thm}

\begin{proof}
Combining the four lemmas of this section with $\eta/4$ and $\veps/3$
in place of $\eta$ and $\veps$, we conclude that by the union bound,
the failure probability is at most $\eta$. Next, the choice of $l=\O(\log\frac{1}{\veps})$
leads to 
\[
\sup_{\K}\frac{\bar{\pi}}{\pi}\leq\frac{1}{\P_{\pi}(Z\in\bar{\K})}\leq1+\frac{2\veps}{3}\,.
\]
Therefore, the claim follow from 
\[
\eu R_{\infty}(\nu\mmid\pi)\leq\eu R_{\infty}(\nu\mmid\bar{\pi})+\eu R_{\infty}(\bar{\pi}\mmid\pi)\leq\frac{\veps}{3}+\log\bpar{1+\frac{2\veps}{3}}\leq\veps\,.\qedhere
\]
\end{proof}

\section{Application 1: Rounding logconcave distributions\label{sec:rounding}}

We have shown that the query complexity of generating a warm start
(in fact, sampling from $\pi^{X}$ with $\eu R_{\infty}$-guarantee)
is $\Otilde(n^{2}(R^{2}\vee n))$. However, $R$ could be arbitrarily
large so this query complexity is not actually polynomial in the problem
parameters. Hence, prior work combines sampling algorithms with a
\emph{rounding} procedure, which makes $\pi^{X}$ less skewed (e.g.,
$R^{2}=\poly n$). For instance, near-isotropic rounding makes $\cov\pi^{X}$
nearly the identity, under which $\E_{\pi^{X}}[\norm{X-\E X}^{2}]\approx n$. 
\begin{problem}
[Isotropic rouding] \label{prob:isotropic-rounding} Assume access
to a well-defined function oracle $\eval_{x_{0},R}(V)$ for convex
$V:\Rn\to\overline{\R}$.  Given $\msf s\geq1$, what is the query
complexity of finding an affine transformation $F:\Rn\to\Rn$ such
that $F_{\#}\pi^{X}$ satisfies 
\[
\msf s^{-1}I_{n}\preceq\cov F_{\#}\pi^{X}=F[\cov\pi^{X}]F^{\T}\preceq\msf sI_{n}\,.
\]
\end{problem}

We may assume that $\min V=0$, since it can be found using the evaluation
oracle when $x_{0}$ is given. This problem is equivalent to estimating
the covariance of $\pi^{X}$ accurately, since setting $F=[\cov\pi^{X}]^{-1/2}$
leads to $\cov F_{\#}\pi^{X}=I_{n}$. The prior best complexity of
isotropic rounding is $\Otilde(n^{4})$ by \cite{lovasz2006fast},
and in this section we improve this to $\Otilde(n^{3.5})$, as claimed
in Theorem~\ref{thm:rounding-intro}.
\begin{thm}
\label{thm:iso-rounding} In the setting of Problem~\ref{prob:isotropic-rounding},
there exists a randomized algorithm with query complexity of $\Otilde(n^{3.5}\polylog R)$
that finds an affine map $F^{X}:\Rn\to\Rn$ such that $F^{X}\pi^{X}$
is $1.01$-isotropic with probability at least $1-\O(n^{-1/2})$.
\end{thm}

\paragraph{Setup.}

By translation, we assume throughout this section that $x_{0}=0$
and that $\K=\{(x,t)\in\R^{n+1}:V(x)-11n\leq nt\}$ contains $B_{1}^{n+1}(0)$.
We now introduce several notation for a concise description. In the
first part of the algorithm, we mainly work with $\D\nu^{X}\propto\exp(-V)|_{\msf L_{\pi^{X},g}}\,\D x$
truncated to the ground set $\msf L_{\pi^{X},g}$, called the \emph{grounded
distribution}. One can think of it as a core of $\pi^{X}$ in that
(1) $\nu^{X}$ is $1.01$-warm with respect to $\pi^{X}$ (i.e., $\eu R_{\infty}(\nu^{X}\mmid\pi^{X})\leq\log1.01$),
and (2) the values of $V$ over the support of $\nu^{X}$ is within
$10n$ from $\min V=0$ (i.e., $\max_{\supp\nu^{X}}V\leq10n$). The
first property immediately follows from 
\begin{equation}
\sup_{\msf L_{\pi^{X},g}}\frac{\nu^{X}}{\pi^{X}}\leq\frac{\int\exp(-V)}{\int_{\msf L_{\pi^{X},g}}\exp(-V)}=\bpar{\P_{\pi^{X}}(X\in\msf L_{\pi^{X},g})}^{-1}\leq1.01\,,\label{eq:nu-pi-warmness}
\end{equation}
where the last inequality is due to Lemma~\ref{lem:LC-tail} (with
$\beta=10$). We also denote by $\nu_{r}^{X}\propto\exp(-V)|_{\msf L_{\pi^{X},g}\cap B_{r}(0)}$
the grounded distribution truncated to the ball $B_{r}(0)$. Hereafter,
we use $\msf L_{g}$ to indicate $\msf L_{\pi^{X},g}$ if there is
no confusion. 

For each of $\pi^{X},\nu^{X},\nu_{r}^{X}$, its exponential counterpart
can be written as, for the cylinders $C_{g}:=\msf L_{g}\times\R$
and $C_{r}:=B_{r}(0)\times\R$,
\[
\pi(x,t)\propto\exp(-nt)|_{\K}\,,\qquad\nu(x,t)\propto\exp(-nt)|_{\K\cap C_{g}}\,,\qquad\nu_{r}(x,t)\propto\exp(-nt)|_{\K\cap C_{g}\cap C_{r}}\,.
\]
We use $F^{X}:\Rn\to\Rn$ and $F:\R^{n+1}\to\R^{n+1}$ to denote an
affine transformation and its embedding into $\R^{n+1}$ defined by
$F(x,t):=(F^{X}(x),t)$, respectively. Then, for $W:=V\circ(F^{X})^{-1}:\Rn\to\R$,
\begin{align*}
F_{\#}\pi(x,t) & \propto\exp(-nt)|_{F\K}\text{ where }F\K=\{(x,t)\in\Rn\times\R:W(x)\leq nt\}\,,\\
F_{\#}\nu(x,t) & \propto\exp(-nt)|_{F\K\cap FC_{g}}\text{ where }FC_{g}=F^{X}\msf L_{g}\times\R\,,\\
F_{\#}\nu_{r}(x,t) & \propto\exp(-nt)|_{F\K\cap FC_{g}\cap FC_{r}}\text{ where }FC_{r}=F^{X}B_{r}(0)\times\R\,.
\end{align*}
Note that $(F_{\#}\pi)^{X}=F_{\#}^{X}\pi^{X}\propto\exp(-W)$, $(F_{\#}\nu)^{X}=F_{\#}^{X}\nu^{X}\propto\exp(-W)|_{F^{X}\msf L_{g}}$,
and $(F_{\#}\nu_{r})^{X}=F_{\#}^{X}\nu_{r}^{X}\propto\exp(-W)|_{F^{X}(\msf L_{g}\cap B_{r})}$.
Lastly, it clearly holds that for $\mu\in\{\pi,\nu,\nu_{r}\}$,
\[
\cov F_{\#}\mu=\left[\begin{array}{cc}
\cov F_{\#}^{X}\mu^{X} & *\\*
* & \cov\mu^{T}
\end{array}\right]\,.
\]

\paragraph{Roadmap to isotropy.}

At the high level, we gradually make $\pi^{X}$ less skewed through
an iterative process of sampling, estimating an approximate covariance
matrix, and applying a suitable affine transformation. Assuming that
Algorithm~\ref{alg:iterative_rounding} ($\msf{Isotropize}$) can
solve Problem~\ref{prob:isotropic-rounding} when $R^{2}=\O(n)$
by using $\Otilde(n^{3})$ many queries, we provide a brief overview
of the entire rounding algorithm. Similar to our warm-start generation
algorithm, for $\delta:=1+n^{-1/2}>1$ and $D:=2R\log400$, this algorithm
gradually isotropizes a sequence $\{\mu_{i}\}$ of distributions,
$\nu_{1}^{X}\to\nu_{\delta}^{X}\to\nu_{\delta^{2}}^{X}\to\cdots\to\nu_{D}^{X}\to\nu^{X}\to\pi^{X}$,
passing along an affine map $F^{X}$ such that if $F_{\#}^{X}\mu_{i}$
is near-isotropic, then $H_{\#}^{X}\mu_{i+1}$ (where $H^{X}$ is
a slight modification of $F^{X}$) satisfies regularity conditions:
\begin{align*}
(\msf{LR})\text{ Lower regularity:} & \quad\text{the ground set of \ensuremath{H_{\#}^{X}\mu_{i+1}} contains }B_{1}^{n}(x_{0})\ \text{for some }x_{0}\,.\\
(\msf{UR})\text{ Upper regularity:} & \quad\E_{H_{\#}^{X}\mu_{i+1}}[\norm{X-x_{0}}^{2}]\leq R^{2}=\O(n)\,.
\end{align*}
This ensures that sampling from $H_{\#}^{X}\mu_{i+1}$ can be achieved
with moderate complexity. We now provide a summary of the rounding
algorithm:
\begin{enumerate}
\item \underline{Step 1}: Make $\nu_{D}^{X}$ near-isotropic.
\begin{enumerate}
\item Run $\msf{Isotropize}$ on $\nu_{1}^{X}$ (which clearly satisfies
$\msf{(LR)}$ and $\msf{(UR)}$ with $R^{2}=\O(1)$) to obtain an
affine map $F^{X}$ such that $F_{\#}^{X}\nu_{1}^{X}$ is near-isotropic.
\item Repeat while $r\leq D=\O(R)$: if $F_{\#}^{X}\nu_{r}^{X}$ is near-isotropic,
then run $\msf{Isotropize}$ on $(4F^{X})_{\#}\nu_{\delta r}^{X}$,
and set $F^{X}\gets4F^{X}$ and $r\gets\delta r$.
\begin{enumerate}
\item Justified by Lemma~\ref{lem:outer-loop}: If $F_{\#}^{X}\nu_{r}^{X}$
is near-isotropic, then $(4F^{X})_{\#}\nu_{\delta r}^{X}$ satisfies
$\msf{(LR)}$ and $\msf{(UR)}$ with $R^{2}=\O(n)$.
\end{enumerate}
\end{enumerate}
\item \underline{Step 2}: Given a near-isotropic $F_{\#}^{X}\nu_{D}^{X}$,
make $\nu^{X}$ near-isotropic.
\begin{enumerate}
\item Run $\msf{Isotropize}$ on $(4F^{X})_{\#}\nu^{X}$ to obtain a new
map $F^{X}$ that makes $F_{\#}^{X}\nu^{X}$ near-isotropic.
\begin{enumerate}
\item Justified by Lemma~\ref{lem:nuD-nu}: If $F_{\#}^{X}\nu_{D}^{X}$
is near-isotropic, then $(4F^{X})_{\#}\nu^{X}$ satisfies $\msf{(LR)}$
and $\msf{(UR})$ with $R^{2}=\O(n)$.
\end{enumerate}
\end{enumerate}
\item \underline{Step 3}: Given a near-isotropic $F_{\#}^{X}\nu^{X}$,
make $\pi^{X}$ near-isotropic.
\begin{enumerate}
\item Draw $\Otilde(n)$ samples from $(4F^{X})_{\#}\pi^{X}$ to estimate
its covariance accurately.
\begin{enumerate}
\item Justified by Lemma~\ref{lem:nu-pi}: If $F_{\#}^{X}\nu^{X}$ is near-isotropic,
then $(4F^{X})_{\#}\pi^{X}$ satisfies $\msf{(LR)}$ and $\norm{\cov F_{\#}^{X}\pi^{X}}=\O(1)$.
\end{enumerate}
\end{enumerate}
\end{enumerate}
 Step 1 and 2 would use $\Otilde(n^{3}\times\sqrt{n}\log D)=\Otilde(n^{3.5}\log D)$
queries, as $\msf{Isotropize}$ is called $\O(\sqrt{n}\log D)$ times
in Step 1 and 2. In Step 3, as the $\psexp$ with an $\O(1)$-warm
start requires $\Otilde(n^{2})$ queries per sample, Step 3 uses $\Otilde(n^{3})$
queries in total. Therefore, the final complexity of this algorithm
would be $\Otilde(n^{3.5}\polylog\nicefrac{R}{\eta})$ for a target
failure probability $\eta$.
\begin{rem}
[Failure of samplers] Recall that the mixing time of $\psexp$ has
a poly-logarithmic dependence on the target failure probability $\eta>0$.
Since the number $k$ of iterations (of $\psexp$) throughout the
algorithm is polynomial in the problem parameters (e.g., $n,R$),
we can, for instance, simply set $\eta=1/(kn^{100})$ without affecting
the final query complexity by more than polylogarithmic factors. 
\end{rem}

\subsection{Isotropic rounding of a well-rounded distribution}

In this section, we analyze the aforementioned algorithm $\msf{Isotropize}$,
which finds an affine map $F^{X}$ making a well-rounded distribution
near-isotropic. We state an exact setup for $\msf{Isotropize}$.
\begin{assumption}
\label{ass:iso-well-rdd} Given access to an evaluation oracle for
convex $U:\Rn\to\overline{\R}$ with $\min U\geq0$, and distribution
$\D\mu^{X}\propto\exp(-U)\,\D x$, assume that $B_{1}^{n}(x_{0})\subset\msf L_{\mu^{X},g}=\{x\in\Rn:U(x)-\min U\leq10n\}$
for known $x_{0}$ and $\E_{\mu^{X}}[\norm{X-x_{0}}^{2}]\leq R^{2}=C^{2}n$
with a known constant $C\geq1$, and that $U(x)\leq10n$ for any $x\in\supp\mu^{X}$
(i.e., $\supp\mu^{X}=\msf L_{\mu^{X},g}$).
\end{assumption}

We note that $\mu(x,t)\propto\exp(-nt)|_{\mc S}$ for $\mc S:=\{(x,t):U(x)-11n\leq nt\}$
has an important property that the slice of $\supp\mu$ at $t=-1$
recovers $\supp\mu^{X}$, which simply follows from $\min U\geq\min V=0$.
Also, this setup can represent $F_{\#}^{X}\nu_{r}^{X}$ and $F_{\#}^{X}\nu^{X}$.
For instance, $\mu^{X}=F_{\#}^{X}\nu_{r}^{X}$ can be obtained by
setting $U=V\circ(F^{X})^{-1}/\ind[x\in F^{X}(\msf L_{g}\cap B_{r}(0))]$.

\paragraph{Rounding algorithm, $\protect\msf{Isotropize}$.}

This algorithm generalizes the approach from \cite{jia2021reducing},
a rounding algorithm for well-rounded uniform distributions over a
convex body, following a modified version of their algorithm studied
in \cite{kook2024covariance} for streamlined analysis. 

Algorithm~\ref{alg:iterative_rounding} begins with $\tgc$ to generate
an $\O(1)$-warm start for $\mu$, using $n^{3}$ queries (Theorem~\ref{thm:warm-start}).
The while-loops then proceeds as follows: for $r:=\inrad\mu^{X}$,
the largest radius of a ball contained in $\supp\mu^{X}=\msf L_{\mu^{X},g}$,
we \textbf{can} run the sampler $\psexp$ (with the $\O(1)$-warm
start and target failure probability $n^{-\O(1)}$) to generate $r^{2}$
samples distributed according to $\mu$, since under Assumption~\ref{ass:iso-well-rdd},
$\mu^{X}$ meets all conditions required by Problem~\ref{prob:exp-sampling}.
We then estimate an approximate covariance matrix $\overline{\Sigma}$
of $\mu^{X}$ (using these $r^{2}$ samples) such that $\abs{\overline{\Sigma}-\Sigma}\precsim nI_{n}$,
where $\Sigma=\cov\mu^{X}$. This procedure requires $r^{2}\times n^{2}\norm{\Sigma}/r^{2}=n^{2}\norm{\Sigma}$
many queries in expectation. Next, we compute the eigenvalue/vectors
of $\overline{\Sigma}$ and double the subspace spanned by the eigenvectors
with eigenvalues less than $n$. One iteration of this inner loop
achieves two key properties: ($\msf{P1}$) the largest eigenvalue
of the covariance $\Sigma$ increases by at most $n$, and ($\msf{P2}$)
$r=\inrad\mu^{X}$ \emph{almost} doubles.

This algorithm repeats this inner loop (i.e., sampling--estimation--scaling)
until $r^{2}$ reaches $n$. As $r$ almost doubles every inner loop,
the inner loop iterates at most $\O(\log n)$ times. Since the operator
norm of the initial covariance matrix is $\O(n)$ due to Assumption~\ref{ass:iso-well-rdd},
the operator norm will remain less than $\Otilde(n)$ throughout the
algorithm, and thus the query complexity of the inner loop is bounded
by $\Otilde(n^{3})$.

\begin{algorithm}[t]
\hspace*{\algorithmicindent} \textbf{Input:} distribution $\mu\propto\exp(-nt)|_{\{(x,t):U(x)-11n\leq nt\}}$
over $\R^{n+1}$ such that $B_{1}(0)\subset\supp\mu^{X}=\msf L_{\mu^{X},g}$
and $\ensuremath{\E_{\mu^{X}}[\norm X^{2}]\leq R^{2}=\O(n)}$.

\hspace*{\algorithmicindent} \textbf{Output:} near-isotropic distribution
$F_{\#}\mu_{\text{last}}$.

\begin{algorithmic}[1]

\STATE Run $\tgc$ to obtain $Z_{0}=(X_{0},T_{0})\in\R^{n+1}$ with
$\eu R_{\infty}(\law Z_{0}\mmid\mu)\leq\log2$.\label{line:GC-warm}

\STATE Let $r_{1}=1$, $\mu_{1}:=\mu$, and $i=1$.

\WHILE{$r_{i}^{2}\lesssim n$}

\STATE For $k_{i}\asymp r_{i}^{2}\polylog n$, draw $\{Z_{j}=(X_{j},T_{j})\}_{j\in[k_{i}]}\gets\PS_{N_{i}}\bpar{\mu_{i},\delta_{Z_{0}},h_{i}}$
with $N_{i}\asymp r_{i}^{-2}n^{3}\log n$ and $h_{i}\asymp r_{i}^{2}(n^{2}\log(r_{i}n))^{-1}$.
\label{line:INO-setting}

\STATE Estimate $\overline{\msf m}_{i}=\frac{1}{k_{i}}\sum_{j=1}^{k_{i}}X_{j}$
and $\overline{\Sigma}_{i}=\frac{1}{k_{i}}\sum_{j=1}^{k_{i}}(X_{j}-\overline{\msf m}_{i})^{\otimes2}$.
\label{line:mean-cov-estimate}

\STATE Compute $M_{i}=I_{n}+P_{i}$, where $P_{i}$ is the orthogonal
projection to the subspace spanned by eigenvectors of $\overline{\Sigma}_{i}$
with eigenvalue at most $n$. Define $M_{i}^{*}(x,t):=(M_{i}x,t)$.

\STATE Set $\mu_{i+1}=(M_{i}^{*})_{\#}\mu_{i}$, $X_{0}\gets M_{i}X_{0}$,
$r_{i+1}=2r_{i}(1-\nicefrac{1}{\log n})$, and $i\gets i+1$. \label{line:r-doubling}

\ENDWHILE

\STATE Draw $\Theta(n\polylog n)$ samples via $\PS_{N}(\mu_{\text{last}},\delta_{Z_{0}},h)$
with $N\asymp n^{2}\polylog n$ and $h\asymp(n\log n)^{-1}$, and
use their $X$-part to estimate the mean $\overline{\msf m}$ and
covariance $\overline{\Sigma}$. For $F(x,t):=(\overline{\Sigma}^{-1/2}(x-\overline{\msf m}),t)$,
return $F_{\#}\mu_{\text{last}}$. \label{line:last-estimation}

\end{algorithmic}\caption{$\protect\msf{Isotropize}$ \label{alg:iterative_rounding}}
\end{algorithm}

\paragraph{Covariance estimation using dependent samples.}

We now specify how we can accurately estimate the target covariance
$\Sigma$ using $r^{2}$ many \emph{dependent} samples drawn by the
$\psexp$. Letting $P$ be the Markov kernel of the $\psexp$ with
target $\mu(x,t)\propto\exp(-nt)|_{\mc S}$ and step size $h\asymp r^{2}/n^{2}$,
we showed in Lemma~\ref{lem:exp-mixing} that $\PS$ has contraction
in $\chi^{2}$-divergence, i.e., for any distribution $\nu\ll\mu$,
\[
\chi^{2}(\nu P\mmid\mu)\le\frac{\chi^{2}(\nu\mmid\mu)}{\bigl(1+h/\cpi(\mu)\bigr)^{2}}\,,
\]
which implies that the spectral gap of this Markov chain is at least
$h\cpi^{-1}=\Omega(r^{2}n^{-2}\cpi^{-1})$. Hence, a new Markov chain
with kernel $P^{N}$ for $N\asymp r^{-2}n^{2}\cpi$ (i.e., composition
of the $\psexp$ $N$-times) has a spectral gap at least $0.99$ (arbitrarily
close to $1$). Hereafter, $\PS_{N}(\mu,\nu,h)$ denotes a Markov
chain with kernel $P^{N}$, initial distribution $\nu$, and target
distribution $\mu$.

When the $\psexp$ samples from $\mu$ in Line~\ref{line:INO-setting},
it takes a sample every $N$ iteration; in other words, it draws $\Theta(r_{i}^{2}\polylog n)$
consecutive samples from the Markov chain with kernel $P^{N}$. Then,
the following result allows us to estimate a covariance matrix of
$\mu$ with a provable guarantee in Line~\ref{line:mean-cov-estimate}
when using dependent samples under a spectral-gap condition:
\begin{lem}
[{\cite[Theorem 8]{kook2024covariance}}] \label{lem:rough-cov-estimation}
Let $P$ be a reversible Markov chain on $\Rn$ with stationary distribution
$\mu$ and spectral gap $\lda$. Let $(Z_{i})_{i\in[k]}$ be a sequence
of outputs given by $P$ with initial distribution $\mu$. If $\mu$
is logconcave and $\lda\geq0.99$, then the covariance estimator $\overline{\Sigma}=k^{-1}\sum_{i=1}^{k}(Z_{i}-\overline{\msf m})^{\otimes2}$
for $\overline{\msf m}=k^{-1}\sum_{i=1}^{k}Z_{i}$ satisfies that
for any $\veps>0$ and $\delta\in(0,n]$, with probability at least
$1-\O(n^{-1})$, 
\[
\abs{\overline{\Sigma}-\Sigma}\preceq\veps\Sigma+\delta I_{n}\quad\text{for }\Sigma=\cov\mu\,,
\]
so long as $k\asymp\frac{\tr\Sigma}{\veps\delta}\log^{2}\frac{n\max(n,\tr\Sigma)}{\delta}\log^{2}n$.
\end{lem}

We also rely on another result to justify the covariance estimation
in Line~\ref{line:last-estimation}, when computing the mean $\overline{\msf m}$
and covariance $\overline{\Sigma}$ of $\mu_{\text{last}}$.
\begin{lem}
[{\cite[Corollary 10 and 33]{kook2024covariance}}] \label{lem:exact-cov-estimation}
In the setting of Lemma~\ref{lem:rough-cov-estimation}, $\overline{\Sigma}$
satisfies that for any $\veps\in(0,1)$, with probability at least
$1-\O(n^{-1})$, 
\[
\abs{\overline{\Sigma}-\Sigma}\preceq\veps\Sigma\quad\text{and}\quad\norm{\overline{\msf m}-\E_{\mu}Z}^{2}\leq\veps^{2}\norm{\Sigma}\,,
\]
so long as $k\asymp\frac{n}{\veps^{2}}\log^{2}\frac{n}{\veps}\log^{2}n$.
\end{lem}

We note that even when the chain starts from a non-stationary distribution
$\nu$, these statistical results still hold, with the bad probability
increasing by a multiplicative factor of $\exp(\eu R_{2}(\nu\mmid\mu))=1+\chi^{2}(\nu\mmid\mu)$
(see \cite[Lemma 29]{kook2024covariance}). As $\psexp$ starts from
an $\O(1)$-warm start and $\eu R_{2}\leq\eu R_{\infty}$, the bad
probability only increases by a factor of $\O(1)$.

\paragraph{Analysis. }

We now examine the two key properties: ($\msf{P1}$) the largest eigenvalue
of the covariance of $\mu_{i}^{X}$ increases by at most $\Otilde(n)$,
and ($\msf{P2}$) $r_{i}=\inrad\mu_{i}^{X}$ nearly doubles every
inner loop. \cite{jia2021reducing} established these results regarding
a rounding algorithm for the \emph{uniform distribution} over a convex
body. Notably, the proofs therein are independent of target distributions,
allowing these results to hold for logconcave distributions as well.
Here, we present an outline of these proofs, along with pointers to
detailed arguments in previous works, following a streamlined analysis
in \cite[\S4.1]{kook2024covariance}.

We first state that $\psexp$ employed in Algorithm~\ref{alg:iterative_rounding}
has a spectral gap at least $0.99$. Its proof is obvious from the
discussion in the paragraph above.
\begin{prop}
[{\cite[Lemma 14]{kook2024covariance}}] \label{prop:spectralGap-PS}
In Line~\ref{line:INO-setting} and~\ref{line:last-estimation},
$\PS_{N_{i}}(\mu_{i},\cdot,h_{i})$ with $N_{i}\asymp r_{i}^{-2}n^{3}\log n$
and $h_{i}\asymp r_{i}^{2}(n^{2}\log(r_{i}n))^{-1}$ has a spectral
gap at least $0.99$.
\end{prop}

As for ($\msf{P1}$), we need the following quantitative bounds on
the operator norm and trace of $\Sigma_{i}^{X}=\cov\mu_{i}^{X}$ in
each while loop of Algorithm~\ref{alg:iterative_rounding}. The original
result was established by \cite[Lemma 3.2]{jia2021reducing}, with
a simpler proof shown in \cite[Lemma 16]{kook2024covariance}. 
\begin{lem}
\label{lem:control-cov} In Algorithm~\ref{alg:iterative_rounding},
the while-loop iterates at most $\O(\log n)$ times, during which
$\tr\Sigma_{i}^{X}\lesssim r_{i}^{2}n$ and $\norm{\Sigma_{i}^{X}}\lesssim n(1+i)$.
\end{lem}

The base case $i=1$ holds since $\norm{\Sigma_{1}^{X}}\leq\tr\Sigma_{1}^{X}\leq R^{2}=\O(n)$.
Assuming the results hold for $i$ as the induction hypothesis, one
must need a guarantee on how close the covariance estimation in Line~\ref{line:mean-cov-estimate}
is to the true covariance, in order to establish the claim for $i+1$.
The following lemma directly follows from Lemma~\ref{lem:rough-cov-estimation}
and Proposition~\ref{prop:spectralGap-PS}.
\begin{prop}
[{\cite[Lemma 15]{kook2024covariance}}] \label{prop:cov-estimation-app}
Each while-loop ensures that with probability at least $1-\O(n^{-1})$,
for $\Sigma_{i}^{X}=\cov\mu_{i}^{X}$ and $\Sigma_{i}=\cov\mu_{i}$,
\[
\frac{9}{10}\,\Sigma_{i}^{X}-\frac{n+1}{100}\,I_{n+1}\preceq\overline{\Sigma}_{i}\preceq\frac{11}{10}\,\Sigma_{i}^{X}+\frac{n+1}{100}\,I_{n+1}\quad\text{ if}\ k_{i}\gtrsim\frac{\tr\Sigma_{i}}{n}\polylog n\,.
\]
\end{prop}

Since $\tr\Sigma_{i}\leq\tr\Sigma_{i}^{X}+\O(1)$ by Lemma~\ref{lem:exp-cov}
and $\tr\Sigma_{i}^{X}\lesssim r_{i}^{2}n$ from the induction hypothesis,
the result above holds when $k_{i}\gtrsim r_{i}^{2}\polylog n$, thereby
justifying the choice of $k_{i}$ in Line~\ref{line:INO-setting}.
With this statistical estimation, the proof of Lemma~\ref{lem:control-cov}
for $i+1$ is exactly the same with \cite[Lemma 16]{kook2024covariance}.

Regarding ($\msf{P2}$), we need minor adjustments to some constants
in the proof of \cite[Lemma 17]{kook2024covariance}, whose original
version (for the uniform distribution) appeared in \cite[Lemma 3.2]{jia2021reducing}.
\begin{lem}
In Algorithm~\ref{alg:iterative_rounding}, under $r_{i+1}=2(1-\nicefrac{1}{\log n})r_{i}$,
each while-loop ensures $r_{i+1}\leq\inrad\mu_{i+1}^{X}$.
\end{lem}

To demonstrate the existence of a ball of radius $r_{i+1}$ centered
at some point $c_{i+1}\in\supp\mu_{i+1}^{X}$, the proof relies on
two ellipsoids within $\supp\mu_{i}^{X}$: a ball $B_{r_{i}}(c_{i})$
from the induction hypothesis, and $\{\norm{x-\msf m_{i}}_{(\Sigma_{i}^{X})^{-1}}\leq e^{-2}\}$
for $\msf m_{i}=\E_{\mu_{i}^{X}}X$, whose existence follows from
Lemma~\ref{lem:ball-in-isotropy}. After multiplying $x\mapsto(I_{n}+P_{i})x=M_{i}x$,
for $c_{i}':=M_{i}c_{i}$ and $\msf m_{i}':=M_{i}\msf m_{i}$, the
support of the pushforward $(M_{i})_{\#}\mu_{i}^{X}=\mu_{i+1}^{X}$
contains two ellipsoids:
\[
\mc A:\{(x-c_{i}')^{\T}M_{i}^{-2}(x-c_{i}')\leq r_{i}^{2}\}\,,\qquad\text{and}\qquad\mc B:\{(x-\msf m_{i}')^{\T}(M_{i}\Sigma_{i}^{X}M_{i})^{-1}(x-\msf m_{i}')\leq e^{-2}\}\,.
\]
Let $\mc S$ be the subspace spanned by the eigenvectors of $\overline{\Sigma}_{i}$
with corresponding eigenvalues less than $n$. The transformation
$M_{i}$ scales up $\mc S$ by a factor of two, ensuring that the
$\mc A$-ellipsoid contains a $2r_{i}$-ball centered at $c_{i}'$
along $\mc S$. Moreover, as eigenvalues of $\overline{\Sigma}_{i}$
along $S^{\perp}$ are at least $n$, we can show that the $\mc B$-ellipsoid
contains an $(n^{1/2}/10)$-ball at $\msf m_{i}'$ along $\mc S^{\perp}$.
Using the convexity of $\supp\mu_{i+1}^{X}$, one can show that the
convex combination of these two ellipsoids leads to the existence
of a ball of radius $2(1-(\log n)^{-1})r_{i}$ at some point. We refer
readers to \cite[Lemma 17]{kook2024covariance} for details.

It is worth noting that $\mu_{i+1}^{X}\propto\exp(-W)$ for a new
potential $W$ satisfies that $W(x)-\min W=W(x)-\min U\leq W(x)\leq10n$
for any $x\in\supp\mu_{i+1}^{X}$ (i.e., $\supp\mu_{i+1}^{X}=\msf L_{\mu_{i+1}^{X},g}$),
and this ground set contains a ball of radius $r_{i+1}$. Given that
this aligns with the setup of Problem~\ref{prob:exp-sampling}, we
can proceed with the $\psexp$ to sample from $\mu_{i+1}^{X}$, referring
to the query complexity result in Theorem~\ref{thm:exp-sampling}.

After the while-loop, in Line~\ref{line:last-estimation}, we use
$k\asymp n\polylog n$ samples to estimate $\overline{\Sigma}$ and
$\overline{\msf m}$, which by Lemma~\ref{lem:exact-cov-estimation}
satisfies that with probability at least $1-\O(n^{-1})$, 
\[
(1-10^{-3})\cov\mu_{\text{last}}^{X}\preceq\overline{\Sigma}\preceq(1+10^{-3})\cov\mu_{\text{last}}^{X}\,.
\]
Thus, for $F^{X}:x\mapsto\overline{\Sigma}^{-1/2}(x-\overline{\msf m})$,
it holds that $F_{\#}^{X}\mu_{\text{last}}^{X}$ is $1.01$-isotropic.
Moreover, Lemma~\ref{lem:exact-cov-estimation} implies that the
mean of $F_{\#}^{X}\mu_{\text{last}}^{X}$ lies within a ball of radius
$10^{-2}$, i.e.,
\[
\norm{\overline{\Sigma}^{-1/2}(\msf m_{\text{last}}-\overline{\msf m})}\leq0.01\,.
\]
Combining these results, we record the final guarantee of Algorithm~\ref{alg:iterative_rounding}.
A similar result for the rounding of uniform distributions can be
found in \cite[Lemma 18]{kook2024covariance}.
\begin{lem}
[Guarantee of $\msf{Isotropize}$] Assume that a distribution $\mu^{X}$
satisfies Assumption~\ref{ass:iso-well-rdd}. Then, with probability
at least $1-\O(n^{-1/2})$, Algorithm~\ref{alg:iterative_rounding}
on the input $\mu$ returns a distribution $F_{\#}\mu_{\textup{last}}$
using $\Otilde(n^{3})$ evaluation queries such that the $X$-marginal
$F_{\#}^{X}\mu_{\textup{last}}^{X}=(F_{\#}\mu_{\textup{last}})^{X}$
is $1.01$-isotropic, with its mean lying within a ball of radius
$0.01$.
\end{lem}

\subsection{Maintaining well-roundedness}

Recall that Step 1-(b) repeats the following while $r=\O(R)$: if
$F_{\#}^{X}\nu_{r}^{X}$ is near-isotropic, then run $\msf{Isotropize}$
on $(4F^{X})_{\#}\nu_{\delta r}^{X}$ for $\delta=1+n^{-1/2}$, and
update $F^{X}\gets4F^{X}$ and $r\gets\delta r$. To ensure this approach
is valid, it must be the case that the input distribution $(4F^{X})_{\#}\nu_{\delta r}^{X}$
satisfies Assumption~\ref{ass:iso-well-rdd} as well. We show that
\cite[Lemma 3.4]{jia2024reducingisotropyvolumekls}, proven for the
uniform distribution over a convex body, remains valid for any logconcave
distribution.
\begin{lem}
\label{lem:outer-loop} Let $\delta=1+n^{-1/2}$ and $F^{X}:\Rn\to\Rn$
an invertible affine map. If $F_{\#}^{X}\nu_{r}^{X}$ is $1.01$-isotropic
with its mean lying within a ball of radius $0.01$, then $(4F^{X})_{\#}\nu_{\delta r}^{X}$
satisfies Assumption~\ref{ass:iso-well-rdd} (i.e., evaluation oracle,
inclusion of $B_{1}(0)$, well-roundedness, and ground set equal to
its support).
\end{lem}

\begin{proof}
By setting $U=V\circ(H^{X})^{-1}/\ind[x\in H^{X}(\msf L_{g}\cap B_{\delta r}(0))]$
with $H^{X}:=4F^{X}$, one can implement the evaluation oracle for
the convex $U$ (using the oracle for $V$) and recover $H_{\#}^{X}\nu_{\delta r}^{X}\propto\exp(-U)$.
It also follows from the definition of $U$ and $\msf L_{g}$ that
$U(x)\leq10n$ on $\supp H_{\#}^{X}\nu_{\delta r}^{X}$.

Let us now examine the inclusion of $B_{1}(0)$ in its ground set.
When $F^{X}:x\mapsto Ax-\mu$ between $\Rn$, let us define an affine
map $F_{0}^{X}:x\mapsto Ax-b$ from $\Rn$ to $\Rn$ which makes $\nu_{r}^{X}$
centered, noting that $\norm{b-\mu}\leq0.01$. As for the lower regularity,
since $(F_{0}^{X})_{\#}\nu_{r}^{X}$ is $1.01$-isotropic with zero
mean, by Lemma~\ref{lem:ball-in-isotropy}, its support contains
a ball of radius $(e\sqrt{1.01})^{-1}\geq0.35$ centered at the origin.
Hence, as $\norm{b-\mu}\leq0.01$, $\supp F_{\#}^{X}\nu_{r}^{X}$
contains a ball of radius $0.3$ centered at the origin, so does $\supp F_{\#}^{X}\nu_{\delta r}^{X}$.
Therefore, by scaling via $x\mapsto4x$, we can ensure that $\supp H_{\#}^{X}\nu_{\delta r}^{X}$
contains a unit ball $B_{1}(0)$. Due to $\min U\geq0$ and $U(x)\leq10n$
on $\supp H_{\#}^{X}\nu_{\delta r}^{X}$, we clearly have $B_{1}(0)\subset\msf L_{H_{\#}^{X}\nu_{\delta r}^{X},g}$.

We define an affine map $G_{0}^{X}:x\mapsto Ax-\delta b$ between
$\Rn$. Observe that for a constant $c'\geq2\delta$,
\begin{align}
\E_{(G_{0}^{X})_{\#}\nu_{\delta r}^{X}}\norm X & =\int_{0}^{\infty}\P_{(G_{0}^{X})_{\#}\nu_{\delta r}^{X}}(\norm X\geq u)\,\D u\leq c'\sqrt{n}+\int_{c'\sqrt{n}}^{\infty}\P(\norm X\geq u)\,\D u\nonumber \\
 & \leq c'\sqrt{n}+\delta\int_{c'\delta^{-1}\sqrt{n}}^{\infty}\P(\norm X\geq\delta u)\,\D u\,.\label{eq:exp-bound}
\end{align}
Let us bound the integrand in terms of $\P_{(F_{0}^{X})_{\#}\nu_{r}^{X}}(\norm X\geq u)$.
For the embedded $G:\R^{n+1}\to\R^{n+1}$ defined by $G(x,t)=(G_{0}^{X}(x),t)$,
we have
\[
\supp G_{\#}\nu_{\delta r}=G(\K\cap C_{g}\cap C_{\delta r})=\left[\begin{array}{cc}
A\\
 & 1
\end{array}\right](\K\cap C_{g}\cap C_{\delta r})-\left[\begin{array}{c}
\delta b\\
0
\end{array}\right]\,,
\]
where $C_{g}=\msf L_{g}\times\R$ and $C_{\delta r}=B_{\delta r}(0)\times\R$.
As $\K$ and $C_{g}$ contain the origin, and $\delta>1$, we clearly
have $\K\cap C_{g}\cap C_{\delta r}\subset\delta(\K\cap C_{g}\cap C_{r})$,
and thus
\begin{align*}
\supp G_{\#}\nu_{\delta r} & \subseteq\delta\left[\begin{array}{cc}
A\\
 & 1
\end{array}\right](\K\cap C_{g}\cap C_{r})-\left[\begin{array}{c}
\delta b\\
0
\end{array}\right]=\delta\Bpar{\left[\begin{array}{cc}
A\\
 & 1
\end{array}\right](\K\cap C_{g}\cap C_{r})-\left[\begin{array}{c}
b\\
0
\end{array}\right]}\\
 & =\delta\supp F_{\#}\nu_{r}\,,
\end{align*}
where $F:(x,t)\mapsto(F_{0}^{X}(x),t)$ is an affine map between $\R^{n+1}$.
Moreover, it follows from $\delta C_{l\sqrt{n}}=C_{l\delta\sqrt{n}}$
for any $l>0$ that 
\[
(\supp G_{\#}\nu_{\delta r})\backslash C_{l\delta\sqrt{n}}\subset[\delta\supp F_{\#}\nu_{r}]\backslash[\delta C_{l\sqrt{n}}]=\delta\,[(\supp F_{\#}\nu_{r})\backslash C_{l\sqrt{n}}]\,.
\]
Integrating $\exp(-nt)$ over these regions and using $\inf_{\K}t\geq-11$
in $(i)$ below,
\begin{align}
\int_{(\supp G_{\#}\nu_{\delta r})\backslash C_{l\delta\sqrt{n}}}\exp(-nt)\,\D x\D t & \leq\int_{\delta\,[(\supp F_{\#}\nu_{r})\backslash C_{l\sqrt{n}}]}\exp(-nt)\,\D x\D t\nonumber \\
 & =\delta^{n+1}\int_{(\supp F_{\#}\nu_{r})\backslash C_{l\sqrt{n}}}\exp(-nt-t\sqrt{n})\,\D x\D t\nonumber \\
 & \underset{(i)}{\leq}e^{11\sqrt{n}}\delta^{n+1}\int_{(\supp F_{\#}\nu_{r})\backslash C_{l\sqrt{n}}}\exp(-nt)\,\D x\D t\nonumber \\
 & =e^{11\sqrt{n}}\delta^{n+1}\cdot\P_{F_{\#}\nu_{r}}(\norm X\geq l\sqrt{n})\int_{\supp F_{\#}\nu_{r}}\exp(-nt)\,\D x\D t\,.\label{eq:int-ineq}
\end{align}
Since $(F_{\#}\nu_{r})^{X}=(F_{0}^{X})_{\#}\nu_{r}^{X}$ is an $1.01$-isotropic
centered logconcave distribution, a version of Paouris' theorem (Lemma~\ref{lem:version-Paouris})
implies that for some universal constant $c>0$ and any $l\geq2$,
\[
\P_{F_{\#}\nu_{r}}(\norm X\geq l\sqrt{n})=\P_{(F_{0}^{X})_{\#}\nu_{r}^{X}}(\norm X\geq l\sqrt{n})\leq\exp(-cl\sqrt{n})\,.
\]
Dividing both sides of \eqref{eq:int-ineq} by $\int_{\supp G_{\#}\nu_{\delta r}}\exp(-nt)\,\D x\D t$,
and substituting this Paouris bound,
\begin{align*}
\P_{(G_{0}^{X})_{\#}\nu_{\delta r}^{X}}(\norm X\geq l\delta\sqrt{n}) & =\P_{G_{\#}\nu_{\delta r}}(\norm X\geq l\delta\sqrt{n})\\
 & \leq e^{11\sqrt{n}}\delta^{n+1}\exp(-cl\sqrt{n})\,\frac{\int_{\supp F_{\#}\nu_{r}}\exp(-nt)\,\D x\D t}{\int_{\supp G_{\#}\nu_{\delta r}}\exp(-nt)\,\D x\D t}\\
 & \leq e^{1+(12-cl)\sqrt{n}}\,.
\end{align*}
Putting this bound back to \eqref{eq:exp-bound},
\begin{align*}
\E_{(G_{0}^{X})_{\#}\nu_{\delta r}^{X}}\norm X & \leq c'\sqrt{n}+\delta\int_{c'\delta^{-1}\sqrt{n}}^{\infty}\P_{(G_{0}^{X})_{\#}\nu_{\delta r}^{X}}(\norm X\geq\delta u)\,\D u\\
 & \leq c'\sqrt{n}+e^{2+12\sqrt{n}}\int_{c'\delta^{-1}\sqrt{n}}^{\infty}e^{-cu}\,\D u\\
 & =c'\sqrt{n}+c^{-1}e^{2+12\sqrt{n}}e^{-cc'\sqrt{n}/\delta}\,.
\end{align*}
By taking $c'$ large enough, the second term can be made exponentially
small. It then follows from the reverse H\"older (Lemma~\ref{lem:reverse-Holder})
that 
\[
\E_{(G_{0}^{X})_{\#}\nu_{\delta r}^{X}}[\norm X^{2}]\lesssim(\E_{(G_{0}^{X})_{\#}\nu_{\delta r}^{X}}\norm X)^{2}\lesssim n\,.
\]
Thus, for $F^{X}(x)=Ax-\mu=G_{0}^{X}+\delta b-\mu$, we have
\begin{align*}
\E_{F_{\#}^{X}\nu_{\delta r}^{X}}[\norm X^{2}] & \leq2\,\E_{F_{\#}^{X}\nu_{\delta r}^{X}}[\norm{X-\delta b+\mu}^{2}]+2\,\norm{\delta b-\mu}^{2}=2\,\E_{(G_{0}^{X})_{\#}\nu_{\delta r}^{X}}[\norm X^{2}]+2\,\norm{\delta b-\mu}^{2}\\
 & \lesssim n+\frac{\norm b^{2}}{n}\lesssim n\,,
\end{align*}
where the last inequality follows from Lemma~\ref{lem:diameter-level-set}
below. Therefore, $\E_{H_{\#}^{X}\nu_{\delta r}^{X}}[\norm X^{2}]=\O(n)$.
\end{proof}
We now present a key geometric property concerning the diameter of
the level set of logconcave distributions. We remark that this properly
generalizes \cite[Lemma 5.18]{lovasz2007geometry}, which says that
if $z$ is a maximum point of an isotropic logconcave density, then
$\norm z\leq n+1$.
\begin{lem}
[Diameter of level sets]\label{lem:diameter-level-set} Suppose an
isotropic logconcave density in $\R^{n}$ is proportional to $e^{-V}$.
Then, for $V_{0}=\min V$ and any $c>0$,
\[
\max_{x:V(x)-V_{0}\leq cn}\norm x\leq8e^{c+1}\,n\,.
\]
\end{lem}

\begin{proof}
Pick any point $x$ with $V(x)\leq V_{0}+cn$. We may assume by rotation
that $x=De_{1}$ for $D=\norm x>0$. Let us consider the marginal
of the density along the direction $e_{1}$, which is also isotropic
with zero mean. We also consider the $(n-1)$-dimensional slices perpendicular
to $e_{1}$ at $0$ and at $D/n$, calling them $\mc S(0)$ and $\mc S(D/n)$
respectively. We can define a bijective mapping $f:\mc S(D/n)\to\mc S(0)$
by $y\mapsto z=x+\frac{y-x}{1-n^{-1}}$; in words, $f(y)$ is the
point $z\in\mc S(0)$ such that $y\in\overline{zx}$. Using the convexity
of $V$ with endpoints $z$ and $x$,
\[
V(y)=V\bpar{\frac{n-1}{n}\,z+\frac{1}{n}\,x}\leq\frac{n-1}{n}\,V(z)+\frac{1}{n}\,V(x)\leq V(z)+c\,.
\]
It then follows that
\begin{align*}
\int_{\mc S(D/n)}e^{-V(y)}\,\D y & \geq e^{-c}\int_{\mc S(D/n)}e^{-V(f(y))}\,\D y=\frac{e^{-c}}{\abs{Df}}\int_{\mc S(0)}e^{-V(z)}\,\D z\\
 & =\bpar{1-\frac{1}{n}}^{n-1}e^{-c}\int_{\mc S(0)}e^{-V(z)}\,\D z\geq e^{-c-1}\int_{\mc S(0)}e^{-V(z)}\,\D z\,.
\end{align*}
By Lemma~\ref{lem:iso-onedim}, the density of the marginal at $0$
is at least $1/8$, so the density at $D/n$ is at least $e^{-c-1}/8$.
Since the value of the marginal density function is at least $\min(1/8,e^{-c-1}/8)=e^{-c-1}/8$
due to the convexity of $V$, the measure of the marginal in the range
$[0,D/n]$ is at least
\[
\frac{D}{8n}e^{-c-1}.
\]
As this must be at most $1$ (actually at most $1-(1/e)$), we obtain
that $D\le8e^{c+1}\,n$ as claimed. 
\end{proof}

\subsection{Rounding the grounded distribution}

We have justified Step 1 of the overall approach thus far. We now
examine Step 2 and 3.

\paragraph{Step 2 ($\nu_{D}^{X}\to\nu^{X}$). }

Given an affine map $F^{X}$ such that $F_{\#}^{X}\nu_{D}^{X}$ is
$1.01$-isotropic, and we show that $(4F^{X})_{\#}\nu^{X}$ satisfies
Assumption~\ref{ass:iso-well-rdd}. The following result allows us
to bound the complexity of $\msf{Isotropize}$ on $(4F^{X})_{\#}\nu^{X}$
by $\Otilde(n^{3})$ as well.
\begin{lem}
\label{lem:nuD-nu} Let $F^{X}:\Rn\to\Rn$ be an invertible affine
map, and $D=2R\log400$. If $F_{\#}^{X}\nu_{D}^{X}$ is $1.01$-isotropic
with its mean located within $0.01$ from the origin, then $(4F^{X})_{\#}\nu_{\delta r}^{X}$
satisfies Assumption~\ref{ass:iso-well-rdd}.
\end{lem}

\begin{proof}
Since $\nu^{X}$ is $1.01$-warm with respect to $\pi^{X}$ due to
\eqref{eq:nu-pi-warmness}, $\E_{\nu^{X}}[\norm X^{2}]\leq1.01\,\E_{\pi^{X}}[\norm X^{2}]\leq(2R)^{2}$,
so Lemma~\ref{lem:LC-exponential-decay} leads to 
\[
\P_{\nu^{X}}(\norm X\geq2tR)\leq e^{-t+1}\quad\text{for any }t\geq1\,.
\]
Thus, $\P_{\nu^{X}}(\norm X\leq D)\geq1-\veps$ for $D=2R\log\frac{e}{\veps}$,
which implies
\[
\frac{\nu_{D}^{X}}{\nu^{X}}=\bpar{\P_{\nu^{X}}(\supp\nu_{D}^{X})}^{-1}=\bpar{\P_{\nu^{X}}(\norm X\leq D)}^{-1}\leq(1-\veps)^{-1}\,.
\]
It means that $\mu_{D}:=F_{\#}^{X}\nu_{D}^{X}$ is also $(1-\veps)^{-1}$-warm
with respect to $\mu:=F_{\#}^{X}\nu^{X}$.

We first address the upper regularity $(\msf{UR})$. For any unit
vector $v\in\Rn$,
\begin{align*}
\bigl|\E_{\mu_{D}}\abs{X\cdot v}-\E_{\mu}\abs{X\cdot v}\bigr| & =\Big|\int\abs{X\cdot v}\,(\D\mu_{D}-\D\mu)\Big|\\
 & \leq\sqrt{\E_{\mu}[\abs{X\cdot v}^{2}]}\sqrt{\chi^{2}(\mu_{D}\mmid\mu)}\\
 & \leq4\,\E_{\mu}\abs{X\cdot v}\,\bpar{\frac{\mu_{D}}{\mu}-1}\,,
\end{align*}
where the last line follows from the reverse H\"older (Lemma~\ref{lem:reverse-Holder}).
Rearranging terms, we obtain that
\begin{align*}
\E_{\mu}\abs{X\cdot v} & \leq\frac{1-\veps}{1-5\veps}\,\E_{\mu_{D}}\abs{X\cdot v}\leq\frac{1-\veps}{1-5\veps}\,\bpar{\E_{\mu_{D}}\abs{(X-\E_{\mu_{D}}X)\cdot v}+0.01}\\
 & \leq\frac{1-\veps}{1-5\veps}\times1.02\,.
\end{align*}
Hence, by taking $\veps$ small enough, we can ensure that $\E_{\mu}\abs{X\cdot v}\leq1.05$
for any $v\in\Rn$, and this immediately implies that 
\[
\E_{\mu}[\norm X^{2}]=\E_{F_{\#}^{X}\nu^{X}}[\norm X^{2}]=\O(n)\,.
\]

As for the lower regularity $(\msf{LR})$, we already showed in the
proof of Lemma~\ref{lem:outer-loop} that due to the isotropy, the
support of $F_{\#}^{X}\nu_{D}^{X}$ contains a ball of radius $0.3$
centered at the origin, so does that of $F_{\#}^{X}\nu^{X}$. Therefore,
scaling by $x\mapsto4x$, we can conclude that $B_{1}(0)\subseteq\supp(4F^{X})_{\#}\nu^{X}$
and that $\E_{(4F^{X})_{\#}\nu^{X}}[\norm X^{2}]=\O(n)$. The implementability
of the evaluation oracle for the potential, as well as $\supp(4F^{X})_{\#}\nu^{X}=\msf L_{(4F^{X})_{\#}\nu^{X},g}$,
can be proven similarly to Lemma~\ref{lem:outer-loop}.
\end{proof}

\paragraph{Step 3 ($\nu^{X}\to\pi^{X}$). }

We show that if $F_{\#}^{X}\nu^{X}$ is $1.01$-isotropic, then the
largest eigenvalue of the covariance of $(4F^{X})_{\#}\pi^{X}$ is
$\O(1)$, and its ground set contains $B_{1}(0)$.
\begin{lem}
\label{lem:nu-pi} Let $F^{X}:\Rn\to\Rn$ be an invertible affine
map such that $F_{\#}^{X}\nu^{X}$ is $1.01$-isotropic with its mean
located within $0.01$ from the origin. Then, for the affine map $H^{X}=4F^{X}$,
the ground set $H^{X}\msf L_{g}$ of $H_{\#}^{X}\pi^{X}$ contains
a unit ball $B_{1}(0)$, and $\norm{\cov H_{\#}^{X}\pi^{X}}=\O(1)$.
\end{lem}

\begin{proof}
As in the proof of Lemma~\ref{lem:nuD-nu}, for any unit vector $v\in\Rn$,
\[
\bigl|\E_{F_{\#}^{X}\nu^{X}}\abs{X\cdot v}-\E_{F_{\#}^{X}\pi^{X}}\abs{X\cdot v}\bigr|\leq4\,\E_{F_{\#}^{X}\pi^{X}}\abs{X\cdot v}\,\bpar{\frac{\nu^{X}}{\pi^{X}}-1}\leq0.04\,\E_{F_{\#}^{X}\pi^{X}}\abs{X\cdot v}\,,
\]
and thus $\E_{F_{\#}^{X}\pi^{X}}\abs{X\cdot v}\leq1.05\,\E_{F_{\#}^{X}\nu^{X}}\abs{X\cdot v}$.
Using the reverse H\"older inequality in $(i)$ below,
\begin{align*}
\E_{F_{\#}^{X}\pi^{X}}[\abs{(X-\E_{F_{\#}^{X}\pi^{X}}X)\cdot v}^{2}] & \leq\E_{F_{\#}^{X}\pi^{X}}[\abs{X\cdot v}^{2}]\underset{(i)}{\leq}16\,(\E_{F_{\#}^{X}\pi^{X}}\abs{X\cdot v})^{2}\leq20\,(\E_{F_{\#}^{X}\nu^{X}}\abs{X\cdot v})^{2}\\
 & \leq20\,\bpar{\E_{F_{\#}^{X}\nu^{X}}\abs{(X-\E_{F_{\#}^{X}\nu^{X}}X)\cdot v}+0.01}^{2}\leq50\,.
\end{align*}
Thus, the largest eigenvalue of $\cov H_{\#}^{X}\pi^{X}$ is $\O(1)$,
and the ground set of $H_{\#}^{X}\pi^{X}$ contains $B_{1}(0)$.
\end{proof}
We note that $\E_{H_{\#}^{X}\pi^{X}}[\norm X^{2}]=\O(n)$, so we can
first run $\tgc$ to generate an $\O(1)$-warm start $Z_{0}=(X_{0},T_{0})$
for $H_{\#}\pi$, using $\Otilde(n^{3})$ queries in expectation.
Then, as in $\msf{Isotropize}$, we run $\PS_{N}(H_{\#}\pi,\delta_{Z_{0}},h)$
with $N\asymp n^{2}\polylog n$ and $h\asymp(n^{2}\polylog n)^{-1}$,
obtaining $k\asymp n\polylog n$ many samples to compute the empirical
covariance $\overline{\Sigma}$ for $\Sigma=\cov H_{\#}^{X}\pi^{X}$.
It then follows from Lemma~\ref{lem:exact-cov-estimation} that with
probability at least $1-\O(n^{-1})$,
\[
0.99\,\overline{\Sigma}\preceq\Sigma\preceq1.01\,\overline{\Sigma}\,.
\]
Thus, $(\overline{\Sigma}^{-1/2}\circ H^{X})_{\#}\pi^{X}$ becomes
$1.01$-isotropic. Since one iteration of $\PS_{N}$ uses $\Otilde(N)$
queries, the total complexity of this procedure (in addition to $\tgc$)
would be $\Otilde(n^{3})$. Therefore, Step 3 also uses $\Otilde(n^{3})$
queries in expectation.

\section{Application 2: Integrating logconcave functions\label{sec:integration}}

In this section, we propose a faster algorithm for integrating logconcave
functions, leveraging our improved sampling guarantees and streamlining
previous analysis in \cite{kannan1997random,lovasz2006simulated,cousins2018gaussian}.
\begin{problem}
[Integration] \label{prob:integration} Assume access to a well-defined
function oracle $\eval_{x_{0},R}(V)$ for convex $V:\Rn\to\overline{\R}$.
Given $\veps>0$, what is the query complexity of computing a $(1+\varepsilon)$-multiplicative
approximation $F$ to the integral of $f=e^{-V}$ with probability
at least $3/4$, i.e., 
\[
(1-\veps)\int f\leq F\leq(1+\veps)\int f\,.
\]
\end{problem}

We can assume that $\min V=0$ and that $\pi^{X}$ is well-rounded
due to the rounding algorithm in the previous section. Hereafter,
we assume that $x_{0}=0$, $R^{2}=C^{2}n$ for a known constant $C>0$,
$\K=\{(x,t):\Rn\times\R:V(x)-11n\leq nt\}$, and that for $l=C^{2}\vee\log\frac{8e}{\veps}$,
\[
\bar{\K}=\K\cap\{B_{Rl}(0)\times[-11,4+13l]\}\,.
\]
We recall that $\bar{\K}$ takes up $(1-\frac{\veps}{4})$ measure
of $\pi$ and denote $\bar{\pi}:=\pi|_{\bar{\K}}$. In this section,
we establish an improved complexity of integrating a logconcave function.
\begin{thm}
[Restatement of  Theorem~\ref{thm:integration-intro}] For any integrable
logconcave $f:\R^{n}\rightarrow\R$ given by a well-rounded function
oracle and any $\varepsilon>0$, there exists an algorithm that with
probability at least $3/4$, returns a $(1+\varepsilon)$-multiplicative
approximation to the integral of $f$ using $\Otilde(n^{3}/\varepsilon^{2})$
queries. For an arbitrary logconcave $f$, the total query complexity
is bounded by $\Otilde(n^{3.5}\polylog R+n^{3}/\varepsilon^{2})$.
\end{thm}

\subsection{Algorithm}

The main idea for integration is also annealing, and will parallel
the warm-start generation but in a more sophisticated way. We consider
a sequence of distributions in $\R^{n+1}$ starting with a simple
logconcave function that is easy to integrate and ending with the
target logconcave function, where two consecutive distributions $\mu_{i}$
and $\mu_{i+1}$ satisfy two conditions: (i) $\eu R_{\infty}(\mu_{i}\mmid\mu_{i+1})=\O(1)$
for fast sampling from $\mu_{i+1}$ started at $\mu_{i}$, and (ii)
$\eu R_{2}(\mu_{i+1}\mmid\mu_{i})=\O(n^{-1})$ or $\O(\nicefrac{\sigma^{2}}{C^{2}n})$
for small variance of estimators used in the algorithm.

This annealing scheme is used in \cite{cousins2018gaussian} for the
special case of Gaussians restricted to convex bodies. It needs an
efficient sampler alongside with good dependence on the warm start
and output guarantees. These were available at the time for Gaussians
restricted to convex bodies. Our new sampler allows us to get similar
strong guarantees for arbitrary logconcave distributions, so we will
be able to extend the guarantees to logconcave integration by addressing
additional technical difficulties. 

\paragraph{Setup.}

Let $m$ be the number of samples used for estimation, and $p$ the
number of total inner phases. Then, the integration algorithm (Algorithm~\ref{alg:integral})
initializes $m$ independent starts, generating a \emph{thread} of
samples denoted by $(Z_{i}^{s}=(X_{i}^{s},T_{i}^{s}))_{0\leq i\leq p}$
for $s\in[m]$. Here, the initial start $Z_{0}^{s}$ is drawn from
$\exp(-\frac{n}{2}\,\norm x^{2})|_{\bar{\K}}$, and subsequent samples
$(Z_{i}^{s})_{i\in[p]}$ drawn through $\psann$ are appended to the
thread. Clearly, for each $i\in[p]$, the $i$-th samples $Z_{i}^{s}$
across the threads are independent and have the same law (i.e., $Z_{i}^{s_{1}}\perp Z_{i}^{s_{2}}$
and $\law Z_{i}^{s_{1}}=\law Z_{i}^{s_{2}}$ for any $s_{1}\neq s_{2}\in[m]$).
$\psann$ with suitable parameters can ensure that the law of a sample
$Z_{i}^{s}$, denoted by $\bar{\mu}_{i}=(\bar{\mu}_{i}^{X},\bar{\mu}_{i}^{T})$,
is close to
\begin{align*}
\mu_{0}(x,t) & \propto\exp\bpar{-\frac{n}{2}\,\norm x^{2}}\big|_{\bar{\K}}=:f_{0}(x,t)\,,\\
\mu_{i}(x,t) & \propto\exp\bpar{-\frac{1}{2\sigma_{i}^{2}}\,\norm x^{2}-\rho_{i}t}\big|_{\bar{\K}}=:f_{i}(x,t)\quad\text{for }i\in[p]\,,
\end{align*}
and we denote its integral (i.e., the normalization constant) by $F_{i}:=\int f_{i}$.

\paragraph{Estimation using drawn samples.}

Algorithm~\ref{alg:integral} calls Algorithm~\ref{alg:sample-estimate}
for estimating $E_{i}:=\frac{1}{m}\sum_{s=1}^{m}Y_{i}^{s}$, just
the empirical average of $Y_{i}^{s}$ defined as
\[
Y_{i}^{s}:=\frac{f_{i+1}(Z_{i}^{s})}{f_{i}(Z_{i}^{s})}\quad\text{for }i=0,\dotsc,p\,,\ \text{with convention }f_{p+1}(z):=\exp(-nt)|_{\bar{\K}}\,.
\]
\begin{algorithm}[t]
\hspace*{\algorithmicindent} \textbf{Input: }initial point $Z_{i-1}$,
target density $f_{i}$, upcoming density $f_{i+1}$, target accuracy
$\epsilon$.

\hspace*{\algorithmicindent} \textbf{Output:} estimation $E_{i}$
and samples $Z_{i}$ approximately close to $\mu_{i+1}\propto f_{i+1}$.

\begin{algorithmic}[1] \caption{$\protect\msf{Samp\text{-}Est}\,(Z_{i-1},f_{i},f_{i+1},\epsilon)$\label{alg:sample-estimate}}

\STATE For each $s\in[m]$, draw $Z_{i}^{s}$ through $\psann$ started
at $Z_{i-1}^{s}$ for target $\mu_{i}\propto f_{i}$ with accuracy
$\epsilon$.

\STATE Estimate
\[
E_{i}=\frac{1}{m}\sum_{s=1}^{m}\frac{f_{i+1}(Z_{i}^{s})}{f_{i}(Z_{i}^{s})}\,.
\]

\STATE Return $E_{i}$ and $\{Z_{i}^{s}\}_{s\in[m]}$.

\end{algorithmic}
\end{algorithm}
If each sample $Z_{i}$ were to follow the \emph{exact distribution}
$\mu_{i}$ rather than the approximate one $\bar{\mu}_{i}$, then
\[
\E E_{i}=\E Y_{i}=\int\frac{f_{i+1}(z)}{f_{i}(z)}\,\frac{f_{i}(z)}{\int f_{i}}\,\D z=\frac{F_{i+1}}{F_{i}}\,,
\]
and thus
\[
\underbrace{\int_{\bar{\K}}\exp\bpar{-\frac{n}{2}\,\norm x^{2}}\,\D x\D t}_{=:I_{0}}\times\E_{\mu_{0}}Y_{0}\cdots\E_{\mu_{p}}Y_{p}=F_{0}\,\frac{F_{1}}{F_{0}}\cdots\frac{F_{p+1}}{F_{p}}=\int_{\bar{\K}}\exp(-nt)\,\D x\D t\,.
\]
Since $\P_{\pi}(X\in\bar{\K})\geq1-\veps$ by construction, and $\int_{\K}\exp(-nt)\,\D x\D t=\frac{1}{n}\int e^{-V}$,
it implies that 
\begin{equation}
(1-\veps)\int e^{-V}\leq nI_{0}\times\E Y_{0}\cdots\E Y_{p}\leq\int e^{-V}\,.\label{eq:final-integral-apprx}
\end{equation}
While estimating $\E Y_{i}$ is straightforward since $\psann$ yields
as many samples from $\bar{\mu}_{i}\approx\mu_{i}$ as needed, we
should also be able to estimate the initial integral $I_{0}=\int_{\bar{\K}}\exp(-\frac{n}{2}\,\norm x^{2})\,\D x\D t$.

\paragraph{The initial integral $I_{0}$.}

For $I:=[-11,4+13l]$, we define $f_{-1}(x,t):=\exp(-\frac{n}{2}\,\norm x^{2})\cdot\ind_{I}(t)$,
and $\mu_{-1}(x,t)=\mc N(0,n^{-1}I_{n})\otimes\text{Unif}\,I\propto f_{-1}(x,t)$.
Using the estimate 
\[
Y_{-1}(x,t):=\frac{f_{0}(x,t)}{f_{-1}(x,t)}=\frac{\exp(-\frac{n}{2}\,\norm x^{2})\cdot\ind_{\bar{\K}}(x,t)}{\exp(-\frac{n}{2}\,\norm x^{2})\cdot\ind_{I}(t)}=\ind_{\bar{\K}}(x,t)\quad\text{for }(x,t)\in\supp\mu_{-1}\,,
\]
and its empirical average $E_{-1}:=\frac{1}{m}\sum_{s=1}^{m}Y_{-1}^{s}$,
we can approximate $I_{0}$ through $E_{-1}\int f_{-1}$ due to 
\[
E_{-1}\int f_{-1}\approx\E_{\mu_{-1}}Y_{-1}\times\int f_{-1}=\int\ind_{\bar{\K}}(x,t)\,f_{-1}(x,t)\,\D x\D t=\int_{\bar{\K}}\exp\bpar{-\frac{n}{2}\,\norm x^{2}}\,\D x\D t=I_{0}\,.
\]
Here, $\int f_{-1}$ can be computed by a formula 
\[
\int f_{-1}=\int\exp(-\frac{n}{2}\,\norm x^{2})\,\D x\times(15+13l)=(2\pi n^{-1})^{n/2}(15+13l)\,.
\]

\paragraph{Overview of algorithm.}

We propose a new annealing scheme (Algorithm~\ref{alg:integral})
to fulfill the two conditions on $\eu R_{2}$ and $\eu R_{\infty}$. 
\begin{itemize}
\item \textbf{{[}Phase I{]}} While $\sigma^{2}\in[n^{-1},1]$:
\begin{itemize}
\item End-to-end: $\mu_{\text{start}}\propto\exp(-\frac{n}{2}\,\norm x^{2})|_{\bar{\K}}$
to $\mu_{\text{end}}\propto\exp(-\half\,\norm x^{2})_{\bar{\K}}$.
\item Update: $\sigma^{2}\gets\sigma^{2}\,(1+\frac{1}{n})$ and $\rho=0$.
\item Guarantee: $\eu R_{2}(\mu_{i+1}\mmid\mu_{i})=\O(n^{-1})$ \& $\Otilde(n)$
phases (Lemma~\ref{lem:int-phase1}).
\end{itemize}
\item \textbf{{[}Phase II{]}} While $\rho\leq n$: 
\begin{itemize}
\item End-to-end: $\mu_{\text{start}}\propto\exp(-\frac{1}{2}\,\norm x^{2}-\frac{1}{28l}\,t)|_{\bar{\K}}$
to $\mu_{\text{end}}\propto\exp(-\half\,\norm x^{2}-nt)_{\bar{\K}}$.
\item Update: {[}$\sigma^{2}\gets\sigma^{2}/(1+\frac{1}{28nl})$ and $\rho\gets\rho\,(1+\frac{1}{28nl})${]}
and {[}$\sigma^{2}\gets\sigma^{2}\,(1+\frac{1}{28nl})$ and $\rho\gets\rho${]}.\footnote{In fact, we can combine two updates, but we separate them for ease
of analysis.}
\item Guarantee: $\eu R_{2}(\mu_{i+1}\mmid\mu_{i})=\O(n^{-1})$ \& $\Otilde(nl)$
phases (Lemma~\ref{lem:int-phase2}).
\end{itemize}
\item \textbf{{[}Phase III{]}} While $\sigma^{2}\in[1,C^{2}n]$:
\begin{itemize}
\item End-to-end: $\mu_{\text{start}}\propto\exp(-\frac{1}{2}\,\norm x^{2}-nt)|_{\bar{\K}}$
to $\mu_{\text{end}}\propto\exp(-\frac{1}{2C^{2}n}\,\norm x^{2}-nt)_{\bar{\K}}$.
\item Update $\sigma^{2}\gets\sigma^{2}\,(1+\frac{\sigma^{2}}{C^{2}n})$
and $\rho=n$.
\item Guarantee: $\eu R_{2}(\mu_{i+1}\mmid\mu_{i})=\O(\frac{\sigma^{2}}{C^{2}n})$
\& $\Otilde(\frac{C^{2}n}{\sigma^{2}})$ phases for doubling of $\sigma^{2}$
(see Lemma~\ref{lem:int-phase3}).
\end{itemize}
\end{itemize}
In summary, $\psann$ started at $\bar{\mu}_{i-1}$ generates samples
$Z\sim\bar{\mu}_{i}$, and estimate $f_{i+1}(Z)/f_{i}(Z)$. Meanwhile,
we would like to ensure that (1) $\eu R_{\infty}(\mu_{i-1}\mmid\mu_{i})=\O(1)$
when moving from $\bar{\mu}_{i-1}$ to $\bar{\mu}_{i}$, and (2) $\eu R_{2}(\mu_{i+1}\mmid\mu_{i})=\O(n^{-1})$
in Phase I and II, and $\O(\frac{\sigma^{2}}{C^{2}n})$ in Phase III
for variance control of $f_{i+1}/f_{i}$. 

\begin{algorithm}[t]
\hspace*{\algorithmicindent} \textbf{Input: }a logconcave function
$f$ in Problem~\ref{prob:integration}, and target accuracy $\veps>0$.

\hspace*{\algorithmicindent} \textbf{Output:} an estimation $E$
for $\int f$.

\begin{algorithmic}[1] \caption{\textbf{$\protect\msf{Integral}\,$}($f,\varepsilon$)\label{alg:integral}}
\STATE $\diamond$ $i=0,l=C^{2}\vee\log\frac{8e}{\veps},(\sigma_{0}^{2},\rho_{0})=(\frac{1}{n},0),m=\frac{10^{11}l\log nl\log C^{2}n}{\veps^{2}},\epsilon=\frac{\veps^{2}}{10^{15}C^{2}nl^{2}\log^{3}(C^{2}nl)}$.
Denote
\[
f_{j}(x,t):=\exp\bpar{-\frac{1}{2\sigma_{j}^{2}}\,\norm x^{2}-\rho_{j}t}\big|_{\bar{\K}}\,.
\]

\STATE Get $\{Z_{-1}^{s}\}_{s\in[m]}\sim\mu_{-1}=\mc N(0,\frac{1}{n}I_{n})\otimes\text{Unif}\,([-11,4+13l])$,
and compute $E_{-1}=\frac{1}{m}\sum_{s=1}^{m}\ind_{\bar{\K}}(Z_{-1}^{s})$.\label{line:-1to0}

\WHILE{$\sigma_{i}^{2}\leq1$}\label{line:phase1}

\STATE If $\sigma_{i}^{2}=1$, then go to next \textbf{while} after
this iteration. Set
\[
\begin{cases}
\sigma_{i+1}^{2}=1\wedge\sigma_{i}^{2}(1+\frac{1}{n})\quad\&\quad\rho_{i+1}=0 & \text{if }\sigma_{i}^{2}<1\,,\\
\sigma_{i+1}^{2}=\sigma_{i}^{2}\quad\&\quad\rho_{i+1}=\frac{1}{28l} & \text{if }\sigma_{i}^{2}=1\,.
\end{cases}
\]

\STATE Run \textbf{$(E_{i},\{Z_{i}^{s}\}_{s\in[m]})\gets\msf{Samp\text{-}Est}\,(\{Z_{i-1}^{s}\}_{s\in[m]},f_{i},f_{i+1},\epsilon)$}
and increment $i$.

\ENDWHILE

\WHILE{$\rho_{i}<n$}\label{line:phase2}

\STATE Set $\beta=1+\frac{1}{28nl}$ if $\beta\rho_{i}<n$, and $\beta=n/\rho_{i}$
otherwise.

\STATE \textbf{$(E_{i},\{Z_{i}^{s}\}_{s\in[m]})\gets\msf{Samp\text{-}Est}\,(\{Z_{i-1}^{s}\}_{s\in[m]},f_{i},f_{i+1},\epsilon)$}
for $\sigma_{i+1}^{2}=\frac{\sigma_{i}^{2}}{\beta},\rho_{i+1}=\beta\rho_{i}$.
Increment $i$.

\STATE \textbf{$(E_{i},\{Z_{i}^{s}\}_{s\in[m]})\gets\msf{Samp\text{-}Est}\,(\{Z_{i-1}^{s}\}_{s\in[m]},f_{i},f_{i+1},\epsilon)$}
for $\sigma_{i+1}^{2}=\beta\sigma_{i}^{2},\rho_{i+1}=\rho_{i}$. Increment
$i$.

\ENDWHILE

\WHILE{$\sigma_{i}^{2}\le C^{2}n$}\label{line:phase3}

\STATE If $\sigma_{i}^{2}=C^{2}n$, then \textbf{leave} this loop
after this iteration. Set
\[
\begin{cases}
\sigma_{i+1}^{2}=C^{2}n\wedge\sigma_{i}^{2}(1+\frac{\sigma_{i}^{2}}{C^{2}n})\quad\&\quad\rho_{i+1}=n & \text{if }\sigma_{i}^{2}<C^{2}n\,,\\
\sigma_{i+1}^{2}=\infty\quad\&\quad\rho_{i+1}=n & \text{if }\sigma_{i}^{2}=C^{2}n\,.
\end{cases}
\]

\STATE\textbf{$(E_{i},\{Z_{i}^{s}\}_{s\in[m]})\gets\msf{Samp\text{-}Est}\,(\{Z_{i-1}^{s}\}_{s\in[m]},f_{i},f_{i+1},\epsilon)$}
and increment $i$. 

\ENDWHILE

\STATE Return $E=n\,(2\pi n^{-1})^{n/2}(15+13l)\times E_{-1}E_{0}\cdots E_{i}$
as the estimate of the integral of $f$.

\end{algorithmic}
\end{algorithm}

\subsection{Streamlined analysis}

We show that our estimate $E_{-1}\cdots E_{p}$ concentrates around
$\E Y_{-1}\cdots\E Y_{p}$ with probability at least $0.9$ (say).

\subsubsection{Technical tools}

\paragraph{Two sources of complications.}

Unfortunately, the samples $Z_{i}$ in the algorithm follow an \textbf{approximate
distribution} $\bar{\mu}_{i}$ (not exactly the target $\mu_{i}$),
so this entails a coupling characterization of $\tv$-distance in
prior work. Once this issue is properly addressed, readers may want
to use a concentration of each $E_{i}$ around $\E E_{i}$ based on
$\var E_{i}$, but this does not suit our purpose. For example, in
Phase I, the (relative) variance $\var(E_{i}/\E E_{i})$ is bounded
by $1/n$, so this leads to $E_{i}=(1\pm\O(n^{-1/2}))\,\E E_{i}$
with high probability. Since there are $p=\Otilde(n)$ many phases,
the relative error of the product of $E_{i}$'s would blow up. Thus,
one should rely on the concentration of the \emph{product} $E_{-1}\cdots E_{p}$,
not just combining that of each estimate $E_{i}$. To this end, we
should be able to establish 
\[
\var(E_{-1}\cdots E_{p})\lesssim(\E[E_{-1}\cdots E_{p}])^{2}\quad\bigl(\text{equivalently, }\E[E_{-1}^{2}\cdots E_{p}^{2}]=\bpar{1+\O(1)}\,(\E[E_{-1}\cdots E_{p}])^{2}\bigr)\,.
\]
Unfortunately, another complication comes into play due to \textbf{dependence
of samples} \emph{within} each thread, so one also needs $\E[E_{-1}\cdots E_{p}]=(1\pm\O(1))\,\E E_{-1}\cdots\E E_{p}$.
Prior analysis of volume algorithms used the notion of \emph{$\alpha$-mixing}
(called $\mu$-independence therein) to address this technicality.
\begin{defn}
[$\alpha$-mixing, \cite{rosenblatt1956central}] Let $(\Omega,\mc F,\P)$
be a probability triplet. Given two sub $\sigma$-algebras $\mc A$
and $\mc B$, their $\alpha$-mixing coefficient is defined by
\[
\alpha(\mc A,\mc B):=\sup_{(A,B)\in\mc A\times\mc B}\abs{\P(A\cap B)-\P(A)\,\P(B)}\,.
\]
For two random variables $X,Y:\Omega\to\R$, its $\alpha$-mixing
coefficient is defined by 
\[
\alpha(X,Y)=\sup_{(A,B)\in\sigma(X)\times\sigma(Y)}\abs{\P(X\in A,Y\in B)-\P(X\in A)\,\P(Y\in B)}\,.
\]

This notion naturally arises when bounding the covariance between
$X$ and $Y$ through the following inequality due to \cite{davydov1968convergence}:
for $p,q,r\in\R_{>0}$ with $p^{-1}+q^{-1}+r^{-1}=1$, and $\alpha=\alpha(X,Y)$,
\[
\abs{\cov(X,Y)}\leq2\alpha^{1/p}\norm X_{q}\norm Y_{r}\,.
\]
Previous volume-computation work used this inequality with $p=1$
and $q=r=\infty$. Since the estimators $Y_{i}$ are unbounded, they
had to introduce several technical tools to handle this issue.
\end{defn}

\paragraph{New tools.}

We provide new approaches that can handle these technicalities in
a much cleaner way. As for the inexact distribution of samples (i.e.,
$\bar{\mu}_{i}$ instead of $\mu_{i}$), $\psann$ provides a provable
guarantee on $\norm{\frac{\D\bar{\mu}_{i}}{\D\mu_{i}}-1}_{L^{\infty}}$,
not just $\norm{\bar{\mu}_{i}-\mu_{i}}_{\tv}$, and this allows us
to bypass a coupling argument. As for the dependence of samples, we
use the notion of \emph{$\beta$-mixing} (or the \emph{coefficient
of absolute regularity}) instead of $\alpha$-mixing, which is viable
also due to stronger guarantees of our samplers.
\begin{defn}
[$\beta$-mixing, \cite{kolmogorov1960strong}] Let $(\Omega,\mc F,\P)$
be a probability triplet. Given two sub $\sigma$-algebras $\mc A$
and $\mc B$, their $\beta$-mixing coefficient is defined by
\[
\beta(\mc A,\mc B):=\norm{\P_{\sigma(\mc A\otimes\mc B)}-\P_{\mc A}\otimes\P_{\mc B}}_{\tv}\,.
\]
For two random variables $X,Y:\Omega\to\R$, it has an equivalent
definition of 
\[
\beta(X,Y)=\norm{\law(X,Y)-\law X\otimes\law Y}_{\tv}\,.
\]
\end{defn}

It holds that $2\alpha\leq\beta$ in general (see \cite[Proposition 1]{doukhan2012mixing}),
so $\beta$-mixing is stronger than $\alpha$-mixing. When $Y$ is
generated by taking $k$ steps of a Markov chin with kernel $P$ started
at $X$, it holds that (see \cite[Proposition 1]{davydov1974mixing})
\[
\beta(X,Y)=\int_{\Omega}\norm{\delta_{x}P^{k}-\law Y}_{\tv}\,(\law X)(\D x)\,.
\]
Recall that when $\psann$ generates $Y\sim\bar{\mu}_{i+1}$ from
$X\sim\bar{\mu}_{i}$, it follows from Lemma~\ref{lem:ps-annealing-iter}
that for $k=\Otilde(n^{2}(\sigma^{2}\vee l^{2})\polylog\frac{R}{\veps})$,
\[
\sup_{x\in\mc{\bar{K}}}\,\norm{\delta_{x}P^{k}-\mu_{i+1}}_{\tv}\leq\veps\,,
\]
and thus
\begin{equation}
\beta(X,Y)=\int_{\bar{\K}}\norm{\delta_{x}P^{k}-\bar{\mu}_{i+1}}_{\tv}\,\D\bar{\mu}_{i}(x)\leq\int_{\bar{\K}}\norm{\delta_{x}P^{k}-\mu_{i+1}}_{\tv}\,\D\bar{\mu}_{i}(x)+\norm{\mu_{i+1}-\bar{\mu}_{i+1}}_{\tv}\leq2\veps\,.\label{eq:beta-bound}
\end{equation}
Therefore, $\norm{\law(X,Y)-\law X\otimes\law Y}\leq2\veps$, and
it implies that we can ``decouple'' $X$ and $Y$ (i.e., use $\law X\otimes\law Y=\bar{\mu}_{i}\otimes\bar{\mu}_{i+1}$
instead of $\law(X,Y)$) at the cost of small probability $2\veps$. 

\subsubsection{Sample and query complexities}

We now analyze the variance of estimators and query complexities
at initialization, Phase I, Phase II, and Phase III. We begin with
a useful observation.
\begin{lem}
\label{lem:variance-estimator} For any $i=-1,0,\dots,p$, 
\[
\var_{\mu_{i}}\bpar{\frac{E_{i}}{\E_{\mu_{i}}E_{i}}}=\frac{1}{m}\var_{\mu_{i}}\bpar{\frac{Y_{i}}{\E_{\mu_{i}}Y_{i}}}=\frac{1}{m}\var_{\mu_{i}}\bpar{\frac{\D\mu_{i+1}}{\D\mu_{i}}}=\frac{1}{m}\,\chi^{2}(\mu_{i+1}\mmid\mu_{i})\,.
\]
\end{lem}

\begin{proof}
The first equality is obvious from independence of samples $Z_{i}^{s}$
across threads $s\in[m]$. For the second, due to 
\[
\frac{Y_{i}}{\E_{\mu_{i}}Y_{i}}=\frac{f_{i+1}}{f_{i}}\frac{F_{i}}{F_{i+1}}\,,
\]
and $\mu_{i+1}\ll\mu_{i}$, we clearly have 
\[
\var_{\mu_{i}}\bpar{\frac{Y_{i}}{\E_{\mu_{i}}Y_{i}}}=\var_{\mu_{i}}\bpar{\frac{\D\mu_{i+1}}{\D\mu_{i}}}\,,
\]
and the third one follows from the definition of $\chi^{2}$.
\end{proof}
We examine query and sample complexities for bounding the variance
of estimator $E_{i}$ by $v$ (i.e., $\E_{\mu_{i}}[E_{i}^{2}]\leq(1+v)\,(\E_{\mu_{i}}E_{i})^{2}$).

\paragraph{Initialization.}

At initialization, $\msf{Integral}$ draws samples from $\mu_{-1}$
and $\mu_{0}$ in Line~\ref{line:-1to0} through rejection sampling.
Below, we will set the target variance $v_{0}=\frac{\veps^{2}}{10^{9}}$,
which requires $m_{0}\geq10^{11}l/\veps^{2}$.
\begin{lem}
[Initialization] \label{lem:int-phase0} Given $v>0$, in Line~\ref{line:-1to0},
$\msf{Integral}$ needs $m\geq\frac{100l}{v}$ samples to bound the
relative variance $\var(E/\E E)$ by $v$, using $\O(m)$ queries
in expectation.
\end{lem}

\begin{proof}
For Line~\ref{line:-1to0}, we can show $\eu R_{\infty}(\mu_{0}\mmid\mu_{-1})=\log\O(l)$
as in the proof of Lemma~\ref{lem:init}: 
\begin{align*}
\frac{\D\mu_{0}}{\D\mu_{-1}} & =\frac{\exp(-\frac{n}{2}\,\norm x^{2})\cdot\ind_{\bar{\K}}(x,t)}{\exp(-\frac{n}{2}\,\norm x^{2})\cdot\ind_{I}(t)}\frac{\int_{\Rn}\exp(-\frac{n}{2}\,\norm x^{2})\,\D x\times\abs I}{\int_{\bar{\K}}\exp(-\frac{n}{2}\,\norm x^{2})\,\D x\D t}\\
 & \leq(15+13l)\,\frac{\int_{\Rn}\exp(-\frac{n}{2}\,\norm x^{2})\,\D x}{\int_{B_{1}^{n}(0)\times[0,1]}\exp(-\frac{n}{2}\,\norm x^{2})\,\D x}\leq(15+13l)\,\bpar{\P_{\mc N(0,n^{-1}I_{n})}(\norm X\leq1)}^{-1}\\
 & \lesssim l\,.
\end{align*}
Thus, by Lemma~\ref{lem:variance-estimator} and $\eu R_{2}=\log(1+\chi^{2})\leq\eu R_{\infty}$,
\[
\var_{\mu_{-1}}\bpar{\frac{E_{-1}}{\E_{\mu_{-1}}E_{-1}}}=\frac{1}{m}\,\chi^{2}(\mu_{0}\mmid\mu_{1})\leq\frac{100l}{m}\leq v\,.
\]
Also, computing $E_{-1}$ uses $m$ many evaluation queries. 
\end{proof}

\paragraph{Phase I.}

We now analyze the variance of $E_{i}$ under the update rule of $\sigma^{2}\gets\sigma^{2}(1+\frac{1}{n})$.
To this end, we recall a useful lemma from \cite[Lemma 3.2]{kalai2006simulated}:
\begin{lem}
\label{lem:slower_rate} Let $f:\R^{n+1}\to\R$ be an integrable logconcave
function. For $a>0$ and $Z(a)=\int f(z)^{a}\,\D z$, the function
$a\mapsto a^{n}Z(a)$ is logconcave in $a$.
\end{lem}

Below, we will set the target variance $v_{1,1}=\frac{\veps^{2}}{10^{9}n\log n}$
during $\sigma^{2}<1$ and $v_{1,2}=\frac{\veps^{2}}{10^{9}}$ when
$\sigma^{2}=1$, each of which requires $m_{1,1}\geq\frac{3\cdot10^{9}\log n}{\veps^{2}}$
and $m_{1,2}\geq\frac{2\cdot10^{9}}{\veps^{2}}$, respectively.
\begin{lem}
[Phase I: Line~\ref{line:phase1}] \label{lem:int-phase1} Given
$v>0$ and $\epsilon>0$, if $\sigma^{2}<1$, then $\msf{Integral}$
needs $m\geq\frac{3}{nv}$ samples to ensure $\var_{\mu}(E/\E_{\mu}E)\leq v$.
Also, when $\psann$ draws $m$ samples with law $\bar{\mu}$ satisfying
$\eu R_{\infty}(\bar{\mu}\mmid\mu)\leq\epsilon$ in each inner phase,
it requires $\Otilde(mn^{3}l^{2}\log^{4}\frac{R}{\epsilon})$ queries
throughout Phase I. When $\sigma^{2}=1$, $\msf{Integral}$ needs
$m\geq\frac{2}{v}$ samples and uses $\Otilde(mn^{2}l^{2}\log^{4}\frac{R}{\epsilon})$
queries.
\end{lem}

\begin{proof}
For $\alpha=1/n$, let us denote 
\[
\mu_{i}(x,t)\propto\exp\bpar{-\frac{1}{2\sigma^{2}/(1+\alpha)}\,\norm x^{2}}\big|_{\bar{\K}}\quad\&\quad\mu_{i+1}(x,t)\propto\exp\bpar{-\frac{1}{2\sigma^{2}}\,\norm x^{2}}\big|_{\bar{\K}}\,.
\]
It then follows from the definition of $L^{2}$-norm that 
\[
\Bnorm{\frac{\D\mu_{i+1}}{\D\mu_{i}}}_{L^{2}(\mu_{i})}^{2}=\frac{\int\exp(-\frac{1-\alpha}{2\sigma^{2}}\,\norm x^{2})\cdot\ind_{\bar{\K}}\times\int_{\bar{\K}}\exp(-\frac{1+\alpha}{2\sigma^{2}}\,\norm x^{2})\cdot\ind_{\bar{\K}}}{\bpar{\int_{\bar{\K}}\exp(-\frac{1}{2\sigma^{2}}\,\norm x^{2})\cdot\ind_{\bar{\K}}}^{2}}=\frac{F(1-\alpha)\,F(1+\alpha)}{F(1)^{2}}\,,
\]
where $F(s):=\int\exp(-\frac{s}{2\sigma^{2}}\,\norm x^{2})\cdot\ind_{\bar{\K}}$.
Using logconcavity of $s^{n+1}F(s)$ in $s$ (due to Lemma~\ref{lem:slower_rate}),
\[
\frac{F(1-\alpha)\,F(1+\alpha)}{F(1)^{2}}\leq\frac{1}{(1-\alpha^{2})^{n+1}}\leq1+\frac{3}{n}\,.
\]
Since $1+\chi^{2}(\mu_{i+1}\mmid\mu_{i})=\norm{\mu_{i+1}/\mu_{i}}_{L^{2}(\mu_{i})}^{2}$,
we can conclude that 
\[
\var_{\mu_{i}}\bpar{\frac{E_{i}}{\E_{\mu_{i}}E_{i}}}=\frac{1}{m}\,\chi^{2}(\mu_{i+1}\mmid\mu_{i})\leq\frac{3}{mn}\,.
\]
Hence, it suffices to take $m\geq\frac{3}{nv}$ samples for estimation.
Regarding the query complexity, there are $\O(n\log n)$ inner phases.
Also, $\mu_{i-1}$ is always $\O(1)$-warm with respect to $\mu_{i}$,
as in Lemma~\ref{lem:tgc-phase1}, thus $\psann$ uses $\Otilde(n^{2}l^{2}\log^{4}\frac{R}{\epsilon})$
queries in expectation (to obtain one sample which is $\epsilon$-close
to $\mu_{i}$ in $\eu R_{\infty}$). Therefore, Phase I requires $\Otilde(mn^{3}l^{2}\log^{4}\frac{R}{\epsilon})$
queries.

When $\sigma^{2}=1$, we let $\mu_{i}\propto\exp(-\frac{1}{2}\,\norm x^{2})|_{\bar{\K}}$
and $\mu_{i+1}\propto\exp(-\frac{1}{2}\,\norm x^{2}-\rho t)|_{\bar{\K}}$
for $\rho=\frac{1}{28l}$. Due to 
\[
\frac{\D\mu_{i+1}}{\D\mu_{i}}=\sup_{\bar{\K}}\frac{\exp(-\rho t)\int_{\bar{\K}}\exp(-\frac{1}{2}\,\norm x^{2})}{\int_{\bar{\K}}\exp(-\frac{1}{2}\,\norm x^{2}-\rho t)}\leq\sup_{\bar{\K}}\exp\bpar{\rho\,(t-\min_{\bar{\K}}t)}\leq e\,,
\]
it follows that for $m\geq2/v$,
\[
\var_{\mu_{i}}\bpar{\frac{E_{i}}{\E E_{i}}}=\frac{1}{m}\,\chi^{2}(\mu_{i+1}\mmid\mu_{i})\leq\frac{2}{m}\leq v\,.
\]
Since $\mu_{i-1}$ is $\O(1)$-warm with respect to $\mu_{i}$, $\psann$
uses $\Otilde(mn^{2}l^{2}\log^{4}\frac{R}{\epsilon})$ queries in
expectation.
\end{proof}

\paragraph{Phase II.}

Variance control in Phase II requires an additional lemma in \cite[Lemma 7.8]{cousins2018gaussian}.
We present it with a streamlined proof.
\begin{lem}
\label{lem:faster_rate} Suppose a logconcave function $f:\R^{n}\rightarrow\R$
has support contained in $B_{R}(0)$. Let $X$ be drawn from $\mu_{\sigma^{2}}\propto g(x,\sigma^{2}):=f(x)e^{-\|x\|^{2}/(2\sigma^{2})}$
and $Y=g(X,\sigma^{2}(1+\alpha))/g(X,\sigma^{2})$. For $\alpha\le1/2$,
\[
\var\bpar{\frac{Y}{\E Y}}=\chi^{2}(\mu_{\sigma^{2}(1+\alpha)}\mmid\mu_{\sigma^{2}})\le\exp\bpar{\frac{2R^{2}\alpha^{2}}{\sigma^{2}}}-1\,.
\]
\end{lem}

\begin{proof}
We can directly show that
\[
\Bnorm{\frac{\D\mu_{\sigma^{2}(1+\alpha)}}{\D\mu_{\sigma^{2}}}}_{L^{2}}^{2}=\frac{\int f(x)\exp(-\frac{1}{2\sigma^{2}}\,\norm x^{2})\times\int f(x)\exp(-\frac{1}{2\sigma^{2}}\frac{1-\alpha}{1+\alpha}\,\norm x^{2})}{\bpar{\int f(x)\exp(-\frac{1}{2\sigma^{2}\,(1+\alpha)}\,\norm x^{2})}^{2}}\,.
\]
Defining $g(s):=\log\bigl(\int f(x)\exp(-\frac{s}{2\sigma^{2}(1+\alpha)}\,\norm x^{2})\bigr)$,
we have
\begin{equation}
\eu R_{2}(\mu_{\sigma^{2}(1+\alpha)}\mmid\mu_{\sigma^{2}})=g(1+\alpha)+g(1-\alpha)-2g(1)=\int_{0}^{\alpha}\bpar{g'(1+t)-g'(1-t)}\,\D t=\int_{0}^{\alpha}\int_{1-t}^{1+t}g''(q)\,\D q\D t\,.\label{eq:renyi2-bound}
\end{equation}
For a logconcave distribution $\nu_{q}(x)\propto f(x)\exp(-\frac{q}{2\sigma^{2}(1+\alpha)}\,\norm x^{2})$,
we have
\begin{align*}
g''(q) & =-\frac{1}{2\sigma^{2}(1+\alpha)}\,\frac{\D}{\D q}\Bpar{\frac{\int\norm x^{2}f(x)\exp(-\frac{q}{2\sigma^{2}(1+\alpha)}\,\norm x^{2})}{\int f(x)\exp(-\frac{q}{2\sigma^{2}(1+\alpha)}\,\norm x^{2})}}\\
 & =\frac{1}{4\sigma^{4}(1+\alpha)^{2}}\,\Bbrack{\frac{\int\norm x^{4}f(x)\exp(-\frac{q}{2\sigma^{2}(1+\alpha)}\,\norm x^{2})}{\int f(x)\exp(-\frac{q}{2\sigma^{2}(1+\alpha)}\,\norm x^{2})}-\Bpar{\frac{\int\norm x^{2}f(x)\exp(-\frac{q}{2\sigma^{2}(1+\alpha)}\,\norm x^{2})}{\int f(x)\exp(-\frac{q}{2\sigma^{2}(1+\alpha)}\,\norm x^{2})}}^{2}}\\
 & =\frac{1}{4\sigma^{4}(1+\alpha)^{2}}\,\var_{\nu_{q}}(\norm X^{2})\underset{\eqref{eq:pi}}{\leq}\frac{1}{q\sigma^{2}(1+\alpha)}\,\E_{\nu_{q}}[\norm X^{2}]\,,
\end{align*}
where we used $\cpi(\nu_{q})\leq\frac{\sigma^{2}(1+\alpha)}{q}$ (Bakry-\'Emery)
in the last line.  Due to $\sup_{\supp\nu_{q}}\norm X\le R$,

\[
g''(q)\leq\frac{R^{2}}{q\sigma^{2}}\,.
\]
Putting this bound back to \eqref{eq:renyi2-bound} and integrating,
\[
\eu R_{2}(\mu_{\sigma^{2}(1+\alpha)}\mmid\mu_{\sigma^{2}})\leq\frac{R^{2}}{\sigma^{2}}\int_{0}^{\alpha}\int_{1-t}^{1+t}\frac{1}{q}\,\D q\D t=\frac{R^{2}}{\sigma^{2}}\int_{0}^{\alpha}\bigl(\log(1+t)-\log(1-t)\bigr)\,\D t\leq\frac{2R^{2}\alpha^{2}}{\sigma^{2}}\,,
\]
which completes the proof.
\end{proof}
We now establish sample and query complexities under the two update
rules: for $1\leq\beta\leq1+\frac{1}{28nl}$, (i) $\sigma^{2}\gets\beta^{-1}\sigma^{2}$
and $\rho\gets\beta\rho$, and (ii) $\sigma^{2}\gets\beta\sigma^{2}$.
Below, we will set the target variance $v_{2}=\frac{\veps^{2}}{2\cdot10^{9}n\log n}$,
which requires $m_{2}\geq\frac{2\cdot10^{9}\log nl}{\veps^{2}}$.
\begin{lem}
[Phase II: Line~\ref{line:phase2}] \label{lem:int-phase2} Given
$v>0$ and $\epsilon>0$, $\msf{Integral}$ needs $m\geq\frac{1}{nlv}$
samples to ensure $\var_{\mu}(E/\E_{\mu}E)\leq v$. Also, when $\psann$
draws $m$ samples with law $\bar{\mu}$ satisfying $\eu R_{\infty}(\bar{\mu}\mmid\mu)\leq\epsilon$
in each inner phase, it requires $\Otilde(mn^{3}l^{3}\log^{4}\frac{R}{\epsilon})$
queries throughout Phase II.
\end{lem}

\begin{proof}
In the first update of $\sigma_{i+1}^{2}=\sigma_{i}^{2}/\beta$ and
$\rho_{i+1}=\beta\rho_{i}$, we can write $\mu_{i}(x,t)\propto\exp(-U(x,t))$
and $\mu_{i+1}(x,t)\propto\exp(-\beta U(x,t))$ for $U(x,t)=(\frac{1}{2}\,\norm x^{2}+\rho t)/\ind_{\bar{\K}}(x,t)$,
where $\frac{1}{28l}\leq\rho<n$. Then, 
\[
\Bnorm{\frac{\D\mu_{i+1}}{\D\mu_{i}}}_{L^{2}(\mu_{i})}^{2}=\frac{\int_{\R^{n+1}}\exp(-(2\beta-1)\,U)\times\int_{\R^{n+1}}\exp(-U)}{\bpar{\int_{\R^{n+1}}\exp(-\beta U)}^{2}}\,.
\]
Using logconcavity of $s^{n+1}\int\exp(-sU)$ in $s$ with endpoints
$2\beta-1$ and $1$ (due to Lemma~\ref{lem:slower_rate}), 
\[
\frac{\int_{\R^{n+1}}\exp(-(2\beta-1)\,U)\times\int_{\R^{n+1}}\exp(-U)}{\bpar{\int_{\R^{n+1}}\exp(-\beta U)}^{2}}\leq\frac{\beta^{2n}}{(2\beta-1)^{n}}\leq1+\frac{1}{nl}\,.
\]
Since $1+\chi^{2}(\mu_{i+1}\mmid\mu_{i})=\norm{\mu_{i+1}/\mu_{i}}_{L^{2}(\mu_{i})}^{2}$,
we can conclude that 
\[
\var_{\mu_{i}}\bpar{\frac{E_{i}}{\E_{\mu_{i}}E_{i}}}=\frac{1}{m}\,\chi^{2}(\mu_{i+1}\mmid\mu_{i})\leq\frac{1}{mnl}\,,
\]
so it suffices to take $m\geq\frac{1}{nlv}$ samples for estimation. 

In the second update of $\sigma_{i+1}^{2}=\beta\sigma_{i}^{2}=1$
and $\rho_{i+1}=\rho_{i}$, we can write $\mu_{i}^{X}\propto g(x)\exp(-\frac{\beta}{2}\,\norm x^{2})$
and $\mu_{i+1}^{X}\propto g(x)\exp(-\frac{1}{2}\,\norm x^{2})$ for
some logconcave function $g$. Then, for $Y=\frac{\exp(-\norm x^{2}/2)}{\exp(-\beta\,\norm x^{2}/2)}$,
\begin{align}
\var_{\mu_{i}}\bpar{\frac{E_{i}}{\E_{\mu_{i}}E_{i}}} & =\frac{1}{m}\var_{\mu_{i}}\bpar{\frac{Y}{\E_{\mu_{i}}Y}}=\frac{1}{m}\var_{\mu_{i}}\bpar{\frac{\D\mu_{i+1}^{X}}{\D\mu_{i}^{X}}}=\frac{1}{m}\var_{\mu_{i}^{X}}\bpar{\frac{\D\mu_{i+1}^{X}}{\D\mu_{i}^{X}}}\label{eq:ex-faster-rate}\\
 & \underset{(i)}{\leq}\frac{1}{m}\,\Bpar{\exp\bpar{\frac{2C^{2}n\,(\beta-1)^{2}}{1/\beta}}-1}\leq\frac{1}{mnl}\,,\nonumber 
\end{align}
where we used Lemma~\ref{lem:faster_rate} in $(i)$.

Regarding the query complexity, we show that $\mu_{i-1}$ is $\O(1)$-warm
with respect to $\mu_{i}$. At the beginning of Phase II, for $t_{0}:=\inf_{\bar{\K}}t$,
\[
\frac{\D\mu_{i-1}}{\D\mu_{i}}=\sup_{\bar{\K}}\frac{\exp(-\frac{1}{2}\,\norm x^{2})}{\exp(-\frac{1}{2}\,\norm x^{2}-\frac{1}{28l}\,t)}\frac{\int_{\bar{\K}}\exp(-\frac{1}{2}\,\norm x^{2}-\frac{1}{28l}\,t)}{\int_{\bar{\K}}\exp(-\frac{1}{2}\,\norm x^{2})}\leq\sup_{\bar{\K}}\exp\bpar{\frac{1}{28l}\,(t-t_{0})}\leq e\,.
\]
In the first update,
\begin{align*}
\frac{\D\mu_{i}}{\D\mu_{i+1}} & =\sup_{\bar{\K}}\frac{\exp(-\frac{1}{2}\,\norm x^{2}-\rho t)}{\exp(-\frac{\beta}{2}\,\norm x^{2}-\beta\rho t)}\frac{\int_{\bar{\K}}\exp(-\frac{\beta}{2}\,\norm x^{2}-\beta\rho t)}{\int_{\bar{\K}}\exp(-\frac{1}{2}\,\norm x^{2}-\rho t)}\leq\sup_{\bar{\K}}\frac{\exp\bigl(-\frac{1}{2}\,\norm x^{2}-\rho\,(t-t_{0})\bigr)}{\exp\bigl(-\frac{\beta}{2}\,\norm x^{2}-\beta\rho\,(t-t_{0})\bigr)}\\
 & \leq\sup\exp\bpar{\frac{1}{56nl}\,\norm x^{2}+\frac{\rho}{28nl}\,(t-t_{0})}\leq e^{2}\,.
\end{align*}
In the second update, the warmness can be shown as in Lemma~\ref{lem:tgc-phase1}.
Putting these together, $\psann$ uses $\Otilde(n^{2}l^{2}\log^{4}\frac{R}{\epsilon})$
queries per inner phase, and there are $\O(nl\log nl)$ inner phases
in Phase II, so we need $\Otilde(mn^{3}l^{3}\log^{4}\frac{R}{\epsilon})$
queries throughout Phase II.
\end{proof}

\paragraph{Phase III. }

We establish sample and query complexity under the update rule of
$\sigma^{2}\gets\sigma^{2}\,(1+\frac{\sigma^{2}}{C^{2}n})$. Since
$\sigma^{2}$ increases from $1$ to $C^{2}n$, we can partition the
whole phase by $\O(\log C^{2}n)$ many doubling phases, where an initial
$\sigma^{2}$ doubles in each doubling phase. Below, we will set the
target variance $v_{3,1}=\frac{\veps^{2}\sigma^{2}}{10^{9}C^{2}n\log C^{2}n}$
during each doubling phase, where $\sigma^{2}$ is an initial value
in the doubling phase. This requires $m_{3,1}\geq\frac{4\cdot10^{9}\log C^{2}n}{\veps^{2}}$.
When $\sigma^{2}=C^{2}n$, we will set $v_{3,2}=\frac{\veps^{2}}{10^{9}}$,
which require $m_{3,2}\geq\frac{2\cdot10^{9}}{\veps^{2}}$.
\begin{lem}
[Phase III: Line~\ref{line:phase3}] \label{lem:int-phase3} Given
$v>0$ and $\epsilon>0$, $\msf{Integral}$ needs $m\geq\frac{4\sigma^{2}}{C^{2}nv}$
samples to ensure $\var_{\mu}(E/\E_{\mu}E)\leq v$. Also, when $\psann$
draws $m$ samples with law $\bar{\mu}$ satisfying $\eu R_{\infty}(\bar{\mu}\mmid\mu)\leq\epsilon$
in each inner phase, it requires $\Otilde(mC^{2}n^{3}l^{2}\log^{4}\frac{R}{\epsilon})$
queries during doubling of $\sigma^{2}$. When $\sigma^{2}=C^{2}n$,
$\msf{Integral}$ needs $m\geq\frac{2}{v}$ samples and uses $\Otilde(mn^{2}l^{2}\log^{4}\frac{R}{\epsilon})$
queries.
\end{lem}

\begin{proof}
Regarding the variance control, similar to \eqref{eq:ex-faster-rate},
we can use Lemma~\ref{lem:faster_rate} (with $R=C\sqrt{n}$ and
$\alpha=\frac{\sigma^{2}}{C^{2}n}$) to show that for $\sigma^{2}\in[1,C^{2}n]$,
\[
\var_{\mu_{i}}\bpar{\frac{E_{i}}{\E E_{i}}}=\frac{1}{m}\var\bpar{\frac{Y}{\E Y}}\leq\frac{1}{m}\,\Bpar{\exp\bpar{\frac{2C^{2}n\alpha^{2}}{\sigma^{2}}}-1}\leq\frac{4\sigma^{2}}{C^{2}mn}\,.
\]
Hence, it suffices to take $m\geq\frac{4\sigma^{2}}{C^{2}nv}$ samples
for estimation. Next, we can show that $\mu_{i-1}$ is $\O(1)$-warm
with respect to $\mu_{i}$ during Phase III, as in the proof of Lemma~\ref{lem:tgc-phase2}.
Since there are $\O(\frac{C^{2}n}{\sigma^{2}})$ inner phases (for
doubling of $\sigma^{2}$), and $\psann$ requires $\Otilde(n^{2}\sigma^{2}l^{2}\log^{4}\frac{R}{\epsilon})$
queries, Phase III uses $\Otilde(mC^{2}n^{3}l^{2}\log^{4}\frac{R}{\epsilon})$
queries in expectation during each doubling phase.

When $\sigma^{2}=C^{2}n$, we let $\mu_{i}\propto\exp(-\frac{1}{2C^{2}n}\,\norm x^{2}-nt)|_{\bar{\K}}$
and $\mu_{i+1}=\bar{\pi}\propto\exp(-nt)|_{\bar{\K}}$. Due to 
\[
\frac{\D\mu_{i+1}}{\D\mu_{i}}=\sup_{\bar{\K}}\frac{\exp(\frac{1}{2C^{2}n}\,\norm x^{2})\int_{\bar{\K}}\exp(-\frac{1}{2C^{2}n}\,\norm x^{2}-nt)}{\int_{\bar{\K}}\exp(-nt)}\leq\sqrt{e},
\]
it follows that for $m\geq2/v$,
\[
\var_{\mu_{i}}\bpar{\frac{E_{i}}{\E E_{i}}}=\frac{1}{m}\,\chi^{2}(\mu_{i+1}\mmid\mu_{i})\leq\frac{2}{m}\leq v\,,
\]
and to sample from $\mu_{i}$, $\psann$ uses $\Otilde(mn^{2}l^{2}\log^{4}\frac{R}{\epsilon})$
queries in expectation.
\end{proof}

\subsubsection{Final guarantee}

With these bounds on the variance of estimators in $\msf{Integral}$,
we can bound the overall variance of the product of estimates using
$m$ samples.  We then show that our final estimate in $\msf{Integral}$
is close to the integral $\int e^{-V}$ with high probability, streamlining
previous analysis in \cite{kannan1997random,lovasz2006simulated,cousins2018gaussian}.
\begin{lem}
\label{lem:estimation-accuracy} If $\msf{Integral}$ runs with sample
size $m=\veps^{-2}(10^{11}l\log nl\log C^{2}n)$ and sample accuracy
$\epsilon\leq\veps^{2}(10^{15}C^{2}nl^{2}\log^{3}(C^{2}nl))^{-1}$,
then it can compute a $(1+\veps)$-multiplicative approximation to
$\int e^{-V}$ with probability at least $0.98$.
\end{lem}

\begin{proof}
We note that the given sample size $m$ satisfies all the required
conditions (i.e., $m\geq m_{0},m_{1,*},m_{2},m_{3,*}$) to attain
a desired target variance in each phase.

We first \emph{assume} that all samples $\{Z_{i}^{s}\}_{0\leq i\leq p,s\in[m]}$
for estimations are independent (so $\{E_{i}\}_{i}$ are independent)
and exactly distributed according to $\mu_{i}$ (not $\bar{\mu}_{i}$).
Then, for $E=E_{-1}E_{0}\cdots E_{p}$,
\begin{align*}
\var E & =\E[E_{-1}^{2}]\cdots\E[E_{p}^{2}]-(\E E_{-1})^{2}\cdots(\E E_{p})^{2}=\prod_{i}\bigl(\var E_{i}+(\E E_{i})^{2}\bigr)-\prod_{i}(\E E_{i})^{2}\\
 & \leq\prod_{i}(1+v_{i})\,(\E E_{i})^{2}-\prod_{i}(\E E_{i})^{2}=\Bpar{\prod_{i}(1+v_{i})-1}\prod_{i}(\E E_{i})^{2}\leq\frac{\veps^{2}}{10^{8}}\prod_{i}(\E E_{i})^{2}\,,
\end{align*}
where the last line follows from
\begin{align*}
\prod_{i}(1+v_{i}) & \leq\exp\Bpar{v_{0}+v_{1,1}\times n\log n+v_{1,2}+v_{2}\times2nl\log nl+\bpar{v_{3,1}\times\frac{C^{2}n}{\sigma^{2}}}\log C^{2}n+v_{3,2}}\\
 & \leq\exp\bpar{\frac{6\veps^{2}}{10^{9}}}\leq1+\frac{\veps^{2}}{10^{8}}\,.
\end{align*}
By Chebyshev's inequality, 
\[
\P_{(\otimes_{i}\mu_{i})^{\otimes m}}\bigl(\abs{E-\E E}\geq\frac{\veps}{10^{3}}\,\E E\bigr)\leq0.01\,.
\]
Hence, \eqref{eq:final-integral-apprx} implies that with probability
at least $0.99$,
\begin{equation}
(1-\veps)\int f\leq E=n\,(2\pi n^{-1})^{n/2}(15+13l)\times E_{-1}E_{0}\cdots E_{i}\leq\bpar{1+\frac{\veps}{10^{3}}}\int f\,.\label{eq:final-desired}
\end{equation}

We now take into two sources of complications: (1) inexact distributions
of samples, and (2) dependence of samples within each thread. As for
(1), since $\psann$ ensures $\frac{\bar{\mu}_{i}}{\mu_{i}}\leq1+\epsilon$,
we have
\[
\frac{\law Z_{1}\otimes\cdots\otimes\law Z_{p}}{\mu_{1}\otimes\cdots\otimes\mu_{p}}=\bigotimes_{i}\frac{\bar{\mu}_{i}}{\mu_{i}}\leq e^{2p\epsilon}\,,
\]
and thus
\[
\P_{(\otimes\bar{\mu}_{i})^{\otimes m}}\bpar{\abs{E-\E E}\geq\frac{\veps}{10^{3}}\,\E E}\leq e^{2mp\epsilon}\,\P_{\otimes\mu_{i}}\bpar{\abs{E-\E E}\geq\frac{\veps}{10^{3}}\,\E E}\leq10^{-2}e^{2mp\epsilon}\,.
\]
Regarding (2), the triangle inequality results in
\begin{align*}
 & \Bnorm{\law(Z_{1},\dots,Z_{p})-\bigotimes_{i=1}^{p}\law Z_{i}}_{\tv}\leq\\
 & \norm{\law(Z_{1},\dots,Z_{p})-\law(Z_{1},\dots,Z_{p-1})\otimes\law Z_{p}}_{\tv}+\Bnorm{\law(Z_{1},\dots,Z_{p-1})\otimes\law Z_{p}-\bigotimes_{i=1}^{p}\law Z_{i}}_{\tv}\,.
\end{align*}
The first term equals $\norm{\law(Z_{p-1},Z_{p})-\law Z_{p-1}\otimes\law Z_{p}}_{\tv}$,
as $\{Z_{i}\}_{i}$ are generated by a Markov chain (see Lemma~\ref{lem:basic-TV}
for details). Applying the DPI to the second term, we obtain that
\begin{align*}
 & \Bnorm{\law(Z_{1},\dots,Z_{p})-\bigotimes_{i=1}^{p}\law Z_{i}}_{\tv}\\
 & =\norm{\law(Z_{p-1},Z_{p})-\law Z_{p-1}\otimes\law Z_{p}}_{\tv}+\Bnorm{\law(Z_{1},\dots,Z_{p-1})-\bigotimes_{i=1}^{p-1}\law Z_{i}}_{\tv}\,.
\end{align*}
As elaborated in \eqref{eq:beta-bound}, the first term is bounded
by $\epsilon$. Repeating the same argument on the second term, we
have
\[
\Bnorm{\law(Z_{1},\dots,Z_{p})-\bigotimes_{i=1}^{p}\law Z_{i}}_{\tv}\leq2p\epsilon\,.
\]
Using a similar argument based on the triangle inequality (see Lemma~\ref{lem:basic-TV}
for details), 
\[
\Bnorm{\law(\{Z_{i}^{s}\}_{i,s})-\Bigl(\bigotimes_{i=1}^{p}\law Z_{i}\Bigr)^{\otimes m}}_{\tv}=\Bnorm{\bigl(\law(Z_{1},\dots,Z_{p})\bigr)^{\otimes m}-\Bigl(\bigotimes_{i=1}^{p}\law Z_{i}\Bigr)^{\otimes m}}_{\tv}\leq2mp\epsilon\,.
\]
Therefore, it holds from the definition of $\tv$-distance that 
\begin{align*}
\P_{(\law(Z_{1},\dots,Z_{p}))^{\otimes m}}\bpar{\abs{E-\E E}\geq\frac{\veps}{10^{3}}\,\E E} & \le\P_{(\otimes_{i}\bar{\mu}_{i})^{\otimes m}}\bpar{\abs{E-\E E}\geq\frac{\veps}{10^{3}}\,\E E}+2mp\epsilon\\
 & \leq10^{-2}e^{2mp\epsilon}+2mp\epsilon\,,
\end{align*}
and this is bounded by $0.02$ when $\epsilon=(300mp)^{-1}\leq\veps^{2}(10^{15}C^{2}nl^{2}\log^{3}(C^{2}nl))^{-1}$.
Therefore, \eqref{eq:final-desired} holds with probability at least
$0.98$.
\end{proof}
We are now ready to prove the main result of this section, as claimed
in Theorem~\ref{thm:integration-intro}.
\begin{thm}
\label{thm:integration} In the setting of Problem~\ref{prob:integration},
for any integrable logconcave $f:\R^{n}\rightarrow\R$ given by a
well-rounded function oracle and any $\varepsilon>0$, with probability
at least $3/4$, $\msf{Integral}$ (Algorithm~\ref{alg:integral})
returns a $(1+\varepsilon)$-multiplicative approximation to the integral
of $f$ using $\Otilde(n^{3}/\varepsilon^{2})$ queries. For an arbitrary
logconcave $f$, the total query complexity is bounded by $\Otilde(n^{3.5}\polylog R+n^{3}/\varepsilon^{2})$. 
\end{thm}

\begin{proof}
By Lemma~\ref{lem:estimation-accuracy}, we must have the sample
size $m=\Otilde(\frac{l\log n}{\veps^{2}})$ and sample accuracy $\epsilon=\O(\frac{\veps^{2}}{C^{2}nl^{2}\log^{3}(C^{2}nl)})$.
These conditions will decide query and sample complexities in each
phase. At initialization (Lemma~\ref{lem:int-phase0}), under the
choices of $m$ and $\epsilon$, we need 
\[
\text{Query: }\O(m)=\Otilde\bpar{\frac{l\log n}{\veps^{2}}}\,.
\]
In Phase I (Lemma~\ref{lem:int-phase1}), we need
\begin{align*}
\sigma^{2}<1: & \text{Query: }\Otilde(mn^{3}l^{2}\log^{4}\frac{R}{\epsilon})=\Otilde\bpar{\frac{n^{3}l^{4}}{\veps^{2}}}\,.\\
\sigma^{2}=1: & \text{Query: }\Otilde(mn^{2}l^{2}\log^{4}\frac{R}{\epsilon})=\Otilde\bpar{\frac{n^{2}l^{3}}{\veps^{2}}}\,.
\end{align*}
In Phase II (Lemma~\ref{lem:int-phase2}), we need
\[
\text{Query: }\Otilde(mn^{3}l^{3}\log^{4}\frac{R}{\epsilon})=\Otilde\bpar{\frac{n^{3}l^{4}}{\veps^{2}}}\,.
\]
In Phase III (Lemma~\ref{lem:int-phase3}), we need
\[
\text{Query: }\Otilde(mC^{2}n^{3}l^{2}\log^{4}\frac{R}{\epsilon}\times\log C^{2}n)=\Otilde\bpar{\frac{C^{2}n^{3}l^{3}}{\veps^{2}}}\,.
\]
Putting these together with $l=C^{2}\vee\log\frac{8e}{\veps}$, the
query complexity of the algorithm is 
\[
\Otilde(\frac{C^{5}n^{3}}{\veps^{2}})\,.
\]

For general logconcave $f$, we first use the rounding algorithm to
bring $f$ to well-rounded position, which takes $\Otilde(n^{3.5})$
queries and then apply the above for well-rounded $f$. 
\end{proof}

\begin{acknowledgement*}
We are grateful to Andre Wibisono for helpful discussions and to Sinho
Chewi for pointing us to the Holley-Stroock principle and the reference
\cite{wang2013analysis} used in the proof of Lemma~\ref{lem:lsi-annealing}.
This work was supported in part by NSR Award CCF-2106444 and a Simons
Investigator grant. 
\end{acknowledgement*}
\bibliographystyle{alpha}
\bibliography{main}

\newcommand{\etalchar}[1]{$^{#1}$}
\begin{thebibliography}{DMLM03}

\bibitem[AK91]{applegate1991sampling}
David Applegate and Ravi Kannan.
\newblock Sampling and integration of near log-concave functions.
\newblock In {\em Symposium on Theory of Computing (STOC)}, pages 156--163,
  1991.

\bibitem[BGL14]{bakry2014analysis}
Dominique Bakry, Ivan Gentil, and Michel Ledoux.
\newblock {\em Analysis and geometry of {M}arkov diffusion operators}, volume
  103.
\newblock Springer, 2014.

\bibitem[CCSW22]{chen2022improved}
Yongxin Chen, Sinho Chewi, Adil Salim, and Andre Wibisono.
\newblock Improved analysis for a proximal algorithm for sampling.
\newblock In {\em Conference on Learning Theory (COLT)}, pages 2984--3014.
  PMLR, 2022.

\bibitem[CEL{\etalchar{+}}22]{chewi2021analysis}
Sinho Chewi, Murat~A Erdogdu, Mufan Li, Ruoqi Shen, and Shunshi Zhang.
\newblock Analysis of {L}angevin {Monte Carlo} from {P}oincar{\'e} to
  l{og-Sobolev}.
\newblock In {\em Conference on Learning Theory (COLT)}, pages 1--2. PMLR,
  2022.

\bibitem[Cha04]{chafai2004entropies}
Djalil Chafa{\"i}.
\newblock Entropies, convexity, and functional inequalities, on
  {$\Phi$}-entropies and {$\Phi$}-sobolev inequalities.
\newblock {\em Journal of Mathematics of Kyoto University}, 44(2):325--363,
  2004.

\bibitem[Che24]{chewi2023log}
Sinho Chewi.
\newblock {\em Log-concave sampling}.
\newblock Book draft available at \url{https://chewisinho.github.io}, 2024.

\bibitem[CV18]{cousins2018gaussian}
Ben Cousins and Santosh~S Vempala.
\newblock Gaussian {C}ooling and {$O^{*}(n^3)$} algorithms for volume and
  {G}aussian volume.
\newblock {\em SIAM Journal on Computing (SICOMP)}, 47(3):1237--1273, 2018.

\bibitem[Dav68]{davydov1968convergence}
Yu~A Davydov.
\newblock Convergence of distributions generated by stationary stochastic
  processes.
\newblock {\em Theory of Probability \& Its Applications}, 13(4):691--696,
  1968.

\bibitem[Dav74]{davydov1974mixing}
Yu~A Davydov.
\newblock Mixing conditions for {M}arkov chains.
\newblock {\em Theory of Probability \& Its Applications}, 18(2):312--328,
  1974.

\bibitem[DFK91]{dyer1991random}
Martin Dyer, Alan Frieze, and Ravi Kannan.
\newblock A random polynomial-time algorithm for approximating the volume of
  convex bodies.
\newblock {\em Journal of the ACM (JACM)}, 38(1):1--17, 1991.

\bibitem[DMLM03]{del2003contraction}
Pierre Del~Moral, Michel Ledoux, and Laurent Miclo.
\newblock On contraction properties of {M}arkov kernels.
\newblock {\em Probability Theory and Related Fields}, 126(3):395--420, 2003.

\bibitem[Dou12]{doukhan2012mixing}
Paul Doukhan.
\newblock {\em Mixing: properties and examples}, volume~85.
\newblock Springer Science \& Business Media, 2012.

\bibitem[Fol99]{folland1999real}
Gerald~B Folland.
\newblock {\em Real analysis: modern techniques and their applications},
  volume~40.
\newblock John Wiley \& Sons, 1999.

\bibitem[HS87]{holley1987logarithmic}
Richard Holley and Daniel~W Stroock.
\newblock Logarithmic {S}obolev inequalities and stochastic {I}sing models.
\newblock {\em Journal of Statistical Physics}, 46(5/6):1159--1194, 1987.

\bibitem[JLLV21]{jia2021reducing}
He~Jia, Aditi Laddha, Yin~Tat Lee, and Santosh~S Vempala.
\newblock Reducing isotropy and volume to {KLS}: an {$O^*(n^3 \psi^2)$} volume
  algorithm.
\newblock In {\em Symposium on Theory of Computing (STOC)}, pages 961--974,
  2021.

\bibitem[JLLV24]{jia2024reducingisotropyvolumekls}
He~Jia, Aditi Laddha, Yin~Tat Lee, and Santosh~S Vempala.
\newblock Reducing isotropy and volume to {KLS}: {F}aster rounding and volume
  algorithms.
\newblock {\em arXiv preprint arXiv:2008.02146}, 2024.

\bibitem[Kla23]{klartag2023logarithmic}
Bo\'az Klartag.
\newblock Logarithmic bounds for isoperimetry and slices of convex sets.
\newblock {\em Ars Inveniendi Analytica}, 2023.

\bibitem[KLS97]{kannan1997random}
Ravi Kannan, L{\'a}szl{\'o} Lov{\'a}sz, and Mikl{\'o}s Simonovits.
\newblock Random walks and an {$O^*(n^5)$} volume algorithm for convex bodies.
\newblock {\em Random Structures \& Algorithms (RS\&A)}, 11(1):1--50, 1997.

\bibitem[KLSV22]{kook2022sampling}
Yunbum Kook, Yin~Tat Lee, Ruoqi Shen, and Santosh~S Vempala.
\newblock Sampling with {R}iemannian {H}amiltonian {M}onte {C}arlo in a
  constrained space.
\newblock In {\em Advances in Neural Information Processing Systems (NeurIPS)},
  volume~35, pages 31684--31696, 2022.

\bibitem[KR60]{kolmogorov1960strong}
Andrei~Nikolaevich Kolmogorov and Yu~A Rozanov.
\newblock On strong mixing conditions for stationary {G}aussian processes.
\newblock {\em Theory of Probability \& Its Applications}, 5(2):204--208, 1960.

\bibitem[KV06]{kalai2006simulated}
Adam~Tauman Kalai and Santosh~S Vempala.
\newblock Simulated annealing for convex optimization.
\newblock {\em Mathematics of Operations Research}, 31(2):253--266, 2006.

\bibitem[KV24]{kook2024gaussian}
Yunbum Kook and Santosh~S Vempala.
\newblock Gaussian cooling and {D}ikin walks: {T}he interior-point method for
  logconcave sampling.
\newblock In {\em Conference on Learning Theory (COLT)}, pages 3137--3240.
  PMLR, 2024.

\bibitem[KVZ24]{kook2024inout}
Yunbum Kook, Santosh~S Vempala, and Matthew~S Zhang.
\newblock In-and-{O}ut: {A}lgorithmic diffusion for sampling convex bodies.
\newblock In {\em To appear in Advances in Neural Information Processing
  Systems (NeurIPS)}, 2024.

\bibitem[KZ24]{kook2024covariance}
Yunbum Kook and Matthew~S Zhang.
\newblock Covariance estimation using {M}arkov chain {M}onte {C}arlo.
\newblock {\em arXiv preprint arXiv:2410.17147}, 2024.

\bibitem[KZ25]{kook2024renyi}
Yunbum Kook and Matthew~S Zhang.
\newblock R{\'e}nyi-infinity constrained sampling with $d^3$ membership
  queries.
\newblock In {\em To appear in Symposium on Discrete Algorithms (SODA)}, 2025.

\bibitem[Led94]{ledoux1994simple}
Michel Ledoux.
\newblock A simple analytic proof of an inequality by {P}. {B}user.
\newblock {\em Proceedings of the American Mathematical Society},
  121(3):951--959, 1994.

\bibitem[Lov91]{lovasz1991compute}
L{\'a}szl{\'o} Lov{\'a}sz.
\newblock {\em How to compute the volume?}
\newblock DIMACS, Center for Discrete Mathematics and Theoretical Computer
  Science, 1991.

\bibitem[LS90]{lovasz1990mixing}
L{\'a}szl{\'o} Lov{\'a}sz and Mikl{\'o}s Simonovits.
\newblock The mixing rate of {M}arkov chains, an isoperimetric inequality, and
  computing the volume.
\newblock In {\em Symposium on Foundations of Computer Science (FOCS)}, pages
  346--354. IEEE, 1990.

\bibitem[LS93]{lovasz1993random}
L{\'a}szl{\'o} Lov{\'a}sz and Mikl{\'o}s Simonovits.
\newblock Random walks in a convex body and an improved volume algorithm.
\newblock {\em Random Structures \& Algorithms (RS\&A)}, 4(4):359--412, 1993.

\bibitem[LST21]{lee2021structured}
Yin~Tat Lee, Ruoqi Shen, and Kevin Tian.
\newblock Structured logconcave sampling with a restricted {G}aussian oracle.
\newblock In {\em Conference on Learning Theory (COLT)}, pages 2993--3050.
  PMLR, 2021.

\bibitem[LV06a]{lovasz2006fast}
L{\'a}szl{\'o} Lov{\'a}sz and Santosh~S Vempala.
\newblock Fast algorithms for logconcave functions: {S}ampling, rounding,
  integration and optimization.
\newblock In {\em Symposium on Foundations of Computer Science (FOCS)}, pages
  57--68. IEEE, 2006.

\bibitem[LV06b]{lovasz2006hit}
L{\'a}szl{\'o} Lov{\'a}sz and Santosh~S Vempala.
\newblock Hit-and-run from a corner.
\newblock {\em SIAM Journal on Computing (SICOMP)}, 35(4):985--1005, 2006.

\bibitem[LV06c]{lovasz2006simulated}
L{\'a}szl{\'o} Lov{\'a}sz and Santosh~S Vempala.
\newblock Simulated annealing in convex bodies and an {${O}^{*}(n^{4})$} volume
  algorithm.
\newblock {\em Journal of Computer and System Sciences (JCSS)}, 72(2):392--417,
  2006.

\bibitem[LV07]{lovasz2007geometry}
L{\'a}szl{\'o} Lov{\'a}sz and Santosh~S Vempala.
\newblock The geometry of logconcave functions and sampling algorithms.
\newblock {\em Random Structures \& Algorithms (RS\&A)}, 30(3):307--358, 2007.

\bibitem[LV24]{lee2024eldan}
Yin~Tat Lee and Santosh~S Vempala.
\newblock Eldan's stochastic localization and the {KLS} conjecture:
  {I}soperimetry, concentration and mixing.
\newblock {\em Annals of Mathematics}, 199(3):1043--1092, 2024.

\bibitem[Mil09]{milman2009role}
Emanuel Milman.
\newblock On the role of convexity in isoperimetry, spectral gap and
  concentration.
\newblock {\em Inventiones Mathematicae}, 177(1):1--43, 2009.

\bibitem[Pao06]{paouris2006concentration}
Grigoris Paouris.
\newblock Concentration of mass on convex bodies.
\newblock {\em Geometric and Functional Analysis (GAFA)}, 16(5):1021--1049,
  2006.

\bibitem[Ros56]{rosenblatt1956central}
Murray Rosenblatt.
\newblock A central limit theorem and a strong mixing condition.
\newblock {\em Proceedings of the National Academy of Sciences (PNAS)},
  42(1):43--47, 1956.

\bibitem[RW00]{ruskai2000study}
Mary~Beth Ruskai and Elisabeth Werner.
\newblock Study of a class of regularizations of $1/|x|$ using {G}aussian
  integrals.
\newblock {\em SIAM Journal on Mathematical Analysis (SIMA)}, 32(2):435--463,
  2000.

\bibitem[VEH14]{van2014renyi}
Tim Van~Erven and Peter Harremos.
\newblock R{\'e}nyi divergence and {K}ullback-{L}eibler divergence.
\newblock {\em IEEE Transactions on Information Theory}, 60(7):3797--3820,
  2014.

\bibitem[VW19]{vempala2019rapid}
Santosh~S Vempala and Andre Wibisono.
\newblock Rapid convergence of the unadjusted {L}angevin algorithm:
  {I}soperimetry suffices.
\newblock {\em Advances in Neural Information Processing Systems (NeurIPS)},
  32, 2019.

\bibitem[Wan13]{wang2013analysis}
Feng-Yu Wang.
\newblock {\em Analysis for diffusion processes on {R}iemannian manifolds},
  volume~18.
\newblock World Scientific, 2013.

\end{thebibliography}

\appendix

\section{Functional inequalities and logconcave geometry}

\paragraph{Functional inequalities.}
\begin{lem}
[Bakry-\'Emery] \label{lem:bakry-emery} If the potential $V$
is $\alpha$-strongly convex, then a probability measure $\pi\propto\exp(-V)$
satisfies
\[
\cpi(\pi)\leq\clsi(\pi)\leq\frac{1}{\alpha}\,.
\]
\end{lem}

\begin{lem}
[Bounded perturbation, {\cite{holley1987logarithmic}}] \label{lem:lsi-bdd-perturbation}
Suppose that $\pi\in\mc P(\Rn)$ satisfies \eqref{eq:lsi} with constants
$\clsi(\pi)$. If $\mu\in\mc P(\Rn)$ satisfies $c\leq\frac{\D\mu}{\D\pi}\leq C$
for $c,C\in\R_{>0}$, then 
\[
\clsi(\mu)\leq\frac{C}{c}\,\clsi(\pi)\,.
\]
\end{lem}

\begin{lem}
[Data-processing inequality] \label{lem:DPI} For $\mu,\nu\in\mc P(\Rn)$,
Markov kernel $P$, $f$-divergence $D_{f}$, and $q\in[1,\infty]$,
it holds that 
\[
D_{f}(\mu P\mmid\nu P)\leq D_{f}(\mu\mmid\nu)\,,\quad\text{and}\quad\eu R_{q}(\mu P\mmid\nu P)\leq\eu R_{q}(\mu\mmid\nu)\,.
\]
\end{lem}

\begin{lem}
[Basic facts on TV] \label{lem:basic-TV} When $\{Z_{i}\}_{i\in[p]}$
are generated by a Markov chain,
\[
\norm{\law(Z_{1},\dots,Z_{p})-\law(Z_{1},\dots,Z_{p-1})\otimes\law Z_{p}}_{\tv}=\norm{\law(Z_{p-1},Z_{p})-\law Z_{p-1}\otimes\law Z_{p}}_{\tv}\,.
\]
Also, it holds that $\norm{\mu^{\otimes m}-\nu^{\otimes m}}_{\tv}\leq m\,\norm{\mu-\nu}_{\tv}$.
\end{lem}

\begin{proof}
For the first claim, let us denote by $f_{1,\dots,p}$ the probability
density of $(Z_{1},\dots,Z_{p})$. Then, 
\begin{align*}
 & \norm{\law(Z_{1},\dots,Z_{p})-\law(Z_{1},\dots,Z_{p-1})\otimes\law Z_{p}}_{\tv}=\frac{1}{2}\int\abs{f_{1,\dots,p}(x,y)-f_{1,\dots,p-1}(x)f_{p}(y)}\,\D x\D y\\
 & =\half\int\abs{f_{p|1,\dots,p-1}(y|x)-f_{p}(y)}f_{1,\dots,p-1}(x)\,\D x\D y=\half\int\abs{f_{p|p-1}(y|x')-f_{p}(y)}f_{p-1}(x')\,\D x'\D y\\
 & =\half\int\abs{f_{p-1,p}(x',y)-f_{p-1}(x')f_{p}(y)}\,\D x'\D y=\norm{\law(Z_{p-1},Z_{p})-\law Z_{p-1}\otimes\law Z_{p}}_{\tv}\,.
\end{align*}

For the second claim, using the DPI for both terms in the RHS (or
just the definition of $\tv$),
\begin{align*}
\norm{\mu^{\otimes m}-\nu^{\otimes m}}_{\tv} & \leq\norm{\mu^{\otimes m-1}\otimes\mu-\nu^{\otimes m-1}\otimes\mu}_{\tv}+\norm{\nu^{\otimes m-1}\otimes\mu-\nu^{\otimes m-1}\otimes\nu}_{\tv}\\
 & \leq\norm{\mu^{\otimes m-1}-\nu^{\otimes m-1}}_{\tv}+\norm{\mu-\nu}_{\tv}\,.
\end{align*}
Applying the same argument to the first term, we finish the proof.
\end{proof}

\paragraph{Geometry of logconcave distributions.}

We recall the most basic property of logconcave functions and distributions:
\begin{lem}
[Dinghas, Leindler, Pr\'ekopa] The product, minimum, and convolution
of logconcave functions is logconcave. Any marginals of logconcave
distributions are logconcave.
\end{lem}

The following explains why we can focus on the ground set $\msf L_{\pi,g}$
of a logconcave distribution $\pi$.
\begin{lem}
[{\cite[Lemma 5.16]{lovasz2007geometry}}] \label{lem:LC-tail}
Let $\pi$ be a logconcave probability density on $\Rn$. For $\beta\geq2$,
\[
\P_{\pi}\bpar{\pi(X)\leq e^{-\beta(n-1)}\max\pi}\leq(e^{1-\beta}\beta)^{n-1}\,.
\]
\end{lem}

We recall two useful tail bounds of logconcave distributions.
\begin{lem}
[Exponential decay, {\cite[Lemma 5.17]{lovasz2007geometry}}] \label{lem:LC-exponential-decay}
For logconcave $\pi\in\mc P(\Rn)$ with $\E_{\pi}[\norm X^{2}]=R^{2}$,
it holds that for any $t\ge1$,
\[
\P_{\pi}(\abs X\geq Rt)\leq e^{-t+1}\,.
\]
\end{lem}

\begin{lem}
[Paouris] \label{lem:version-Paouris} There exists a universal constant
$c>0$ such that for any $s$-isotropic centered logconcave $\pi\in\mc P(\Rn)$
and $t\geq1$, 
\[
\P_{\pi}(\norm X\geq t\sqrt{sn})\leq\exp(-ct\sqrt{n})\,.
\]
\end{lem}

\begin{proof}
For $\Sigma$ the covariance matrix of $\pi$, we note that $\nu:=(\Sigma^{-1/2})_{\#}\pi$
is isotropic logconcave. Then, for $E=\{X\in\Rn:\norm X_{\Sigma^{-1}}\geq t\sqrt{n}\}$
and some universal constant $c>0$,
\[
\pi(E)=\int\ind_{E}(x)\,\D\pi(x)=\int\ind_{E}(\Sigma^{1/2}y)\,\D\nu(y)=\nu\bpar{\Rn\backslash B_{t\sqrt{n}}(0)}\leq\exp(-ct\sqrt{n})\,,
\]
where the last line follows from Paouris' theorem \cite{paouris2006concentration}.
Due to $s^{-1}I_{n}\preceq\Sigma^{-1}$, we have $\P_{\pi}(\norm X\geq t\sqrt{sn})\leq\exp(-ct\sqrt{n})$,
from which the claim easily follows.
\end{proof}
The reverse direction of the H\"older inequality holds for logconcave
distributions.
\begin{lem}
[Reverse H\"{o}lder, {\cite[Theorem 5.22]{lovasz2007geometry}}]
\label{lem:reverse-Holder} For logconcave $\pi\in\mc P(\Rn)$, it
holds that $(\E_{\pi}[\norm X^{k}])^{1/k}\leq2k\,\E_{\pi}\norm X$
for $k\geq1$.
\end{lem}

The following folklore reveals a lower regularity of logconcave distributions:
\begin{lem}
[{\cite[Lemma 5.13]{lovasz2007geometry}}] \label{lem:ball-in-isotropy}
The support of any $s$-isotropic logconcave distribution on $\Rn$
contains a ball of radius $\nicefrac{1}{\sqrt{s}e}$ centered at its
mean.
\end{lem}

The following on one-dimensional isotropic logconcave density is crucial
when proving Lemma~\ref{lem:diameter-level-set}.
\begin{lem}
[{\cite[Lemma 5.5 (b)]{lovasz2007geometry}}] \label{lem:iso-onedim}
Let $f:\R\to\R_{\geq0}$ be an isotropic logconcave density function.
Then, $f(0)\geq1/8$.
\end{lem}

\end{document}